%% file: main.tex
\documentclass[10pt]{article}

\usepackage[preprint]{tmlr}

\usepackage[utf8]{inputenc}
\usepackage{amsmath,amsfonts,amssymb,amsthm}
\usepackage{algorithmic}
\usepackage{algorithm}
\usepackage{array}
\usepackage{subfig}
\usepackage{float}
\usepackage{textcomp}
\usepackage{url}
\usepackage{verbatim}
\usepackage{graphicx}
\usepackage{enumitem}
\usepackage{xparse}
\usepackage{booktabs}
\usepackage[table]{xcolor}
\usepackage{makecell}
\usepackage[percent]{overpic}
\usepackage{hyperref}
\hypersetup{hidelinks}

\newcolumntype{C}{>{\centering\arraybackslash}p{1.2em}}

\definecolor{gtbg}{RGB}{210,245,210}
\definecolor{mdsbg}{RGB}{255,230,150}
\definecolor{cedbg}{RGB}{255,210,210}
\definecolor{shade}{RGB}{248,248,248}
\definecolor{good}{RGB}{0,120,0}
\definecolor{bad}{RGB}{180,0,0}

\theoremstyle{plain}
\newtheorem{theorem}{Theorem}[section]
\newtheorem{proposition}[theorem]{Proposition}
\newtheorem{lemma}[theorem]{Lemma}
\newtheorem{corollary}[theorem]{Corollary}

\theoremstyle{definition}
\newtheorem{definition}[theorem]{Definition}

\theoremstyle{remark}
\newtheorem{remark}[theorem]{Remark}

\providecommand{\appendices}{}

\def\month{MM}
\def\year{YYYY}
\def\openreview{\url{https://openreview.net/forum?id=XXXX}}

\title{Sierpi\'nski--Knopp Wasserstein Distance for Persistence Diagrams and Applications to 2-Wasserstein Approximation}

\author{\name Sebastien Tchitchek \email sebastien.tchitchek@etu.sorbonne-universite.fr\\
\addr CNRS and Sorbonne Universit\'e
\AND
\name Julien Tierny \email julien.tierny@sorbonne-universite.fr\\
\addr CNRS and Sorbonne Universit\'e}

\begin{document}

\input{notations.tex}

\maketitle

\begin{abstract}
\input{abstract_tmlr.tex}

\end{abstract}

\input{introduction_Section1_}
\input{preliminaries_Section2_}

\input{continuous_edit_distance_between_time_varying_persistence_diagrams_Section3_}

\input{continuous_edit_distance_geodesics_between_time_varying_persistence_diagrams_Section4_}
\input{applications_Section6_}

\input{results_Section7_}
\input{limitations_section}
\input{conclusion_Section8_}

% Acknowledgments should normally be omitted in the anonymous submission.
% Uncomment for the accepted or preprint version.
%\subsubsection*{Acknowledgments}
%This work is partially supported by the European Commission grant ERC-2019-COG
%``TORI'' (ref. 863464, \url{https://erc-tori.github.io/}).

\bibliography{refs}
\bibliographystyle{tmlr}

\clearpage
\section*{Appendices}
\appendix

\input{appendix_skot_from_ieee.tex}
\input{supplementary_experiments.tex}

\end{document}

%% file: notations.tex
%%%%%%%%%%%%%%%%%%%%%%%%%%%%%%%%%%%%%%%%%%%%%%%%%%%%%%%%%%%%%%%%%%%%%%%%
% notations.tex  –  Shortcuts and notations for the IEEE-TVCG paper “CED”
%%%%%%%%%%%%%%%%%%%%%%%%%%%%%%%%%%%%%%%%%%%%%%%%%%%%%%%%%%%%%%%%%%%%%%%%

\newcommand{\encommentaire}[1]{}

\newcommand{\toDiscuss}[1]{#1}

\newcommand{\julien}[1]{#1}
\newcommand{\sebastien}[1]{#1}

%color used when orange text (\sebastien) has been overwritten
%\newcommand{\sebastienBis}[1]{\textcolor{red}{#1}}
\newcommand{\sebastienBis}[1]{#1}

\newcommand{\Scorrection}[1]{#1}

\renewcommand{\sectionautorefname}{Sec.}
\renewcommand{\subsectionautorefname}{Sec.}
\renewcommand{\subsubsectionautorefname}{Sec.}

\renewcommand{\figureautorefname}{Fig.}
\renewcommand{\tableautorefname}{Tab.}

\newcommand{\N}{N}  

\newcommand{\RNB}{\mathbb{R}}      

\newcommand{\NNB}{\mathbb{N}}      

\newcommand{\algebra}{\mathcal{M}}  

\newcommand{\borelian}{\mathcal{B}}

\newcommand{\candidate}{\ensuremath{B}}

\newcommand{\candidateTwo}{\ensuremath{C}}

\newcommand{\frechetEnergy}{\mathcal{E}}

\newcommand{\measure}{\mu}    

\newcommand{\completion}{\mathcal{Z}}

\newcommand{\scalarField}{\ensuremath{f}}   

\newcommand{\variableTime}{\ensuremath{t}}  

\newcommand{\manifold}{\mathcal{M}}  

\newcommand{\dimensionvariable}{\ensuremath{d}}

\newcommand{\timedScalarFieldsSeq}{\ensuremath{U}}  

\newcommand{\PDSeq}{\ensuremath{V}}  

\newcommand{\sampleTVPD}{\ensuremath{V}}  

\newcommand{\divisibilitySymbol}{\mid} 

\newcommand{\sublevelset}{L^{-}_{w}}     

\newcommand{\PDS}{\mathcal{D}}   

\newcommand{\PD}{\ensuremath{X}} 

\newcommand{\birthPair}{\ensuremath{b}} 

\newcommand{\deathPair}{\ensuremath{d}} 

\newcommand{\persistancepair}{\ensuremath{x}} 

\newcommand{\persistancepairTwo}{\ensuremath{y}} 

\newcommand{\PDBijection}{\psi}        

\newcommand{\floorCost}{\ensuremath{c}} 

\newcommand{\wasserstein}{\ensuremath{W}}  

\newcommand{\diago}{\Lambda}  

\newcommand{\diagoProj}[1]{\pi\left(#1\right)}    

\newcommand{\deletionSet}{\mathcal{D}} 

\newcommand{\substitutionSet}{\mathcal{S}} 

\newcommand{\insertionSet}{\mathcal{I}} 

\newcommand{\morseIndex}{\mathcal{I}} 

\newcommand{\persistencePairNumber}{\ensuremath{I}} 

\newcommand{\persistencePairNumberBis}{\ensuremath{J}} 

\newcommand{\persistencePairNumberBisBis}{\ensuremath{K}} 

\newcommand{\metricSpace}{\ensuremath{\mathcal{X}}}

\newcommand{\metricDistance}{\ensuremath{d}}

\newcommand{\paramDelta}{\Delta}     

\newcommand{\paramDeltaone}{\Delta_1}  

\newcommand{\paramDeltatwo}{\Delta_2}  

\newcommand{\paramWeight}{\alpha}    

\newcommand{\paramPenalty}{\beta}            

\newcommand{\scalarVariableOne}{\varepsilon} 

\newcommand{\scalarVariableTwo}{\gamma}

\newcommand{\paramIntervalSize}{\eta}        

\newcommand{\paramSGD}{\rho_{\mathrm{s}}}

\newcommand{\paramGGD}{\ensuremath{k}}

\newcommand{\paramGGDBis}{\ensuremath{M}} 

\newcommand{\stepGre}{\rho_{\mathrm{g}}}

\newcommand{\paramGD}{\ensuremath{T}}

\newcommand{\paramGDBis}{\ensuremath{t}}  

\newcommand{\WtwoGeodesic}{\gamma}

\newcommand{\metricGeodesic}{\gamma}

\newcommand{\boundarySet}{\ensuremath{A}}

\newcommand{\PASet}{\mathcal{A}}

\newcommand{\initDPOne}{\ensuremath{A}}

\newcommand{\initDPTwo}{\ensuremath{B}}

\newcommand{\Subdspace}{\ensuremath{s}}   

\newcommand{\TVPDspace}{\ensuremath{\mathcal{S}}}  

\newcommand{\CTVPDspace}{\ensuremath{C}}   

\newcommand{\ITVPDspace}{\ensuremath{I}}   

\newcommand{\PCTVPDspace}{\mathrm{PC}}     

\newcommand{\Img}{\mathrm{Im}}   

\newcommand{\Dom}{\mathrm{Dom}}   

\newcommand{\SubdInterval}{I}    

\newcommand{\TVPDp}{\ensuremath{P}} 

\newcommand{\TVPDq}{\ensuremath{Q}}

\newcommand{\TVPDf}{\ensuremath{F}}

\newcommand{\TVPDg}{\ensuremath{G}}

\newcommand{\TVPDX}{\ensuremath{X}}

\newcommand{\TVPDY}{\ensuremath{Y}}

\newcommand{\functionf}{\ensuremath{f}}

\newcommand{\functiong}{\ensuremath{g}}

\newcommand{\partialAssignment}{\ensuremath{f}}

\newcommand{\partialAssignmentBis}{\ensuremath{g}}

\newcommand{\partialAssignmentBisBis}{\ensuremath{h}}

\newcommand{\dilatedWtwoGeodesic}{\mathcal{H}}

\newcommand{\geodesicCED}{\ensuremath{G}}

\newcommand{\TDAP}{TDA}

\newcommand{\TDAPP}{TDA\ }

\newcommand{\localDistance}{\ensuremath{d}}

\newcommand{\localDistanceAppendix}{\ensuremath{D}}

\newcommand{\assignmentCost}{\operatorname{cost}} 

\newcommand{\TVPDP}{TVPD\ }

\newcommand{\TVPDPP}{TVPD}

\newcommand{\TVPDsP}{TVPDs\ } 

\newcommand{\TVPDsPP}{TVPDs} 

\newcommand{\TVPDM}{\mathrm{TVPD}}

\newcommand{\CEDP}{CED\ }  

\newcommand{\CEDPP}{CED}  

\newcommand{\CEDM}{\mathrm{CED}}

\newcommand{\DP}{\delta}

\newcommand{\dimension}{\ensuremath{d}}

\newcommand{\intervalBounda}{\ensuremath{a_{i}}}

\newcommand{\intervalBoundaBis}{\ensuremath{a_{i,n}}}

\newcommand{\intervalBoundaBisBis}{\ensuremath{a_{i,n+1}}}

\newcommand{\intervalBoundb}{\ensuremath{b_{i}}}

\newcommand{\intervalBoundc}{\ensuremath{c_{j}}}

\newcommand{\intervalBoundd}{\ensuremath{d_{j}}}

\newcommand{\criticalPointOne}{\ensuremath{c}}

\newcommand{\criticalPointTwo}{\ensuremath{c'}}

\newcommand{\limitl}{\ensuremath{l}}

\newcommand{\variablei}{\ensuremath{i}}

\newcommand{\variablej}{\ensuremath{j}}

\newcommand{\variablek}{\ensuremath{k}}

\newcommand{\variableK}{\ensuremath{K}}

\newcommand{\variableKLips}{\ensuremath{K}}

\newcommand{\variablel}{\ensuremath{l}}

\newcommand{\variableL}{\ensuremath{L}}

\newcommand{\variableM}{\ensuremath{M}}

\newcommand{\variablen}{\ensuremath{n}}

\newcommand{\variableN}{\ensuremath{N}}

\newcommand{\variabler}{\ensuremath{r}}

\newcommand{\variables}{\ensuremath{s}}

\newcommand{\variablev}{\ensuremath{v}}

\newcommand{\variablew}{\ensuremath{w}}

\newcommand{\variablex}{\ensuremath{x}}

\newcommand{\variablexBis}{\ensuremath{x}}

\newcommand{\variablexBisBis}{\ensuremath{x}}

\newcommand{\variabley}{\ensuremath{y}}

\newcommand{\variableyMetricSpace}{\ensuremath{y}}

\newcommand{\variablez}{\ensuremath{z}}

\newcommand{\variableZ}{\ensuremath{Z}}

\newcommand{\rayon}{\ensuremath{r}}

\newcommand{\ball}{\mathcal{B}}

\newcommand{\intervalI}{\ensuremath{I}}

\newcommand{\averageNumberDeltaS}{\ensuremath{n}}

\newcommand{\averageNumberPersistencePair}{\ensuremath{p}}

\newcommand{\averageNumberPersistencePairPerTVPD}{\ensuremath{P}}

\newcommand{\sampleSize}{\ensuremath{N}}

\newcommand{\bigLandau}{\ensuremath{\mathcal{O}}}

\newcommand{\I}{[0,1]}
\newcommand{\dd}{\,\mathrm d}
\newcommand{\cDel}{c_{\mathrm{del}}}
\newcommand{\cCre}{c_{\mathrm{cre}}}
\newcommand{\cbd}{c}
\newcommand{\X}{\mathcal X}
\newcommand{\Y}{\mathcal Y}
\newcommand{\Xbar}{\overline{\mathcal X}}
\newcommand{\Ybar}{\overline{\mathcal Y}}
\newcommand{\qbar}{\overline q}
\newcommand{\dX}{d_{\mathcal X}}
\newcommand{\dY}{d_{\mathcal Y}}
\newcommand{\dXbar}{\overline d_{\mathcal X}}
\newcommand{\dYbar}{\overline d_{\mathcal Y}}
\newcommand{\dAh}{\overline d_{\mathcal Ah}}
\newcommand{\Diag}{\mathrm{Diag}}
\newcommand{\Left}{\mathrm{Left}}
\newcommand{\Top}{\mathrm{Top}}
\newcommand{\R}{\mathbb R}
\newcommand{\Atri}{\mathcal A_{\triangle}}
\newcommand{\AtriSet}{\{(x,y)\in[0,1]^2:\ x\le y\}}
\newcommand{\dist}{\operatorname{dist}}
\newcommand{\distset}[2]{\dist\!\big(#1,#2\big)}

%%%%%%%%%%%%%%%%%%%%%%%%%%%%%%%%%%%%%%%%%%%%%%%%%%%%%%%%%%%%%%%%%%%%%%%%
% End of notations.tex
%%%%%%%%%%%%%%%%%%%%%%%%%%%%%%%%%%%%%%%%%%%%%%%%%%%%%%%%%%%%%%%%%%%%%%%%

%% file: abstract_tmlr.tex
\julien{This paper introduces the Sierpi\'nski--Knopp (SK) Wasserstein distance, a fast metric between persistence diagrams. The SK-Wasserstein distance, denoted $d_{\mathrm{SK}}$,} maps diagram
points and their diagonal projections to the unit interval
\julien{via the Sierpi\'nski--Knopp space-filling
curve on the upper diagonal triangle. The encoded point sets are then efficiently matched via one-dimensional optimal assignment, in \(O(N\log N)\) steps, yielding an explicit
diagonal-aware point assignment between the two input persistence
diagrams.} \julien{We show that the SK-Wasserstein distance controls} the
classical \(2\)-Wasserstein distance \julien{between diagrams}, admits an explicit isometric embedding
into a Hilbert space, and induces a positive-definite Gaussian kernel, making the resulting geometry
directly compatible with Euclidean and kernel-based learning methods.
\julien{A tighter surrogate dissimilarity, noted \(W_\Gamma\), is also introduced based on the point  assignments along the curve.} Experiments on 12 scientific collections comprising 227 diagrams show median Spearman correlations with \(W_2\) of \(0.879\) for
\(d_{\mathrm{SK}}\) and \(0.924\) for \(W_\Gamma\),
\julien{reflecting consistent pairwise dissimilarity rankings}.
The median per-collection speedup of \(d_{\mathrm{SK}}\) \julien{over state-of-the-art approximations of \(W_2\)} is \(626\times\), while the
aggregate speedup over the full benchmark is \(2100\times\). Average-linkage partitions obtained from
\(d_{\mathrm{SK}}\) and \(W_\Gamma\) each exactly match the
corresponding \(W_2\) partition on 8 of the 12 collections. Hilbert
\(k\)-means and Gaussian spectral clustering, both based on
\(d_{\mathrm{SK}}\), achieve mean adjusted Rand indices (ARI) of \(0.756\)
and \(0.800\), respectively, with respect to the benchmark reference
partitions, compared to
\(0.750\) obtained by average linkage on \(W_2\). The Gaussian \(d_{\mathrm{SK}}\) kernel supports other
kernel-based analysis tasks, as illustrated by its use
for contiguous segmentation of ordered diagram collections in our experiments. A C++ implementation is provided for reproducibility at \href{https://github.com/sebastien-tchitchek/SK-Wasserstein-Reproducibility}{https://github.com/sebastien-tchitchek/SK-Wasserstein-Reproducibility}.

%% file: introduction_Section1_.tex
\section{Introduction}\label{sec:introduction}

Advancements in data acquisition and numerical simulation have
made increasingly large and complex datasets available for analysis. In this
context, topological data analysis (TDA) provides a family of techniques for
summarizing geometric and scalar data through concise topological descriptors
\cite{book,oudot2015persistence}. Among them, persistence diagrams are one of
the most widely used representations: they encode the birth and death of
topological features across a filtration and provide a compact summary of the
salient structures of a
% scalar field
\julien{dataset
\cite{B94,frosini99,robins99,edelsbrunner02}.}

A central task in persistence-based analysis is the comparison of persistence diagrams. Standard metrics, such as the \(2\)-Wasserstein distance, compare diagrams through optimal
\julien{assignments}
% matchings
in the birth--death plane, while allowing unmatched points to be paired with the diagonal \cite{Munkres1957,Bertsekas81,kerber2017geometry}. This distance has strong theoretical foundations, including stability guarantees~\cite{cohen2005stability}, and has enabled
% statistical and learning
\julien{analysis}
tasks on diagram spaces~\cite{mileyko2011probability,turner2013frechetmeansdistributionspersistence,lacombe2018large,vidal2019progressive,pont_tvcg22,sisouk_tvcg24, sisouk2026robustbarycenterspersistencediagrams}. However, its repeated evaluation can become costly, especially when large collections of diagrams must be compared or when distances are used inside iterative algorithms. This motivates the design of fast diagram distances that remain faithful to the \(2\)-Wasserstein geometry.

This paper addresses this issue by introducing
\julien{the \emph{Sierpi\'nski--Knopp (SK) Wasserstein distance},
denoted $d_{\mathrm{SK}}$,}
% \emph{Sierpi\'nski--Knopp
% Optimal Transport} (SKOT),
a fast
\julien{metric}
% distance
between
% normalized
persistence
diagrams. The construction keeps the diagonal-aware augmentation used in the
standard \(2\)-Wasserstein comparison of persistence diagrams, where unmatched
features may be paired with the diagonal through their orthogonal projections.
The key difference is that, instead of solving the resulting
\julien{assignment}
% transport
problem
in the persistence
\julien{upper diagonal}
triangle \julien{(denoted \(\Atri\))},
\julien{$d_{\mathrm{SK}}$}
replaces it with a one-dimensional
assignment problem on \([0,1]\). This is achieved by fixing a Sierpi\'nski--Knopp space-filling curve \(S:[0,1]\to\Atri\) and using its first-hit selector on each augmented
diagram to assign a scalar value in \([0,1]\) to every diagram point, thereby
producing two equal-length lists of scalar values. These lists,
equivalently viewed as equal-mass atomic measures on \([0,1]\), are then compared through the
classical one-dimensional \(1\)-Wasserstein assignment problem, and
\(d_{\mathrm{SK}}\) is defined as the square root of its optimal cost. \Scorrection{Beyond the distance value, the one-dimensional optimal assignment yields an explicit
diagonal-aware assignment between the two input diagrams, which can be
reused by downstream methods requiring point correspondences.}

The benefit of this construction is twofold. First, one-dimensional optimal
% transport
\julien{assignment}
admits a monotone closed-form solution, which reduces the evaluation of
% SKOT
\julien{$d_{\mathrm{SK}}$}
to sorting lists of scalar values. Second, the \(1/2\)-H\"older regularity of the SK
curve ties this one-dimensional comparison back to the original geometry of the
persistence triangle. In particular,
\julien{the classical $2$-Wasserstein distance between persistence diagrams is controlled, up to a constant factor, by \toDiscuss{\Scorrection{${d_{\mathrm{SK}}}$}}.}
% the quadratic Wasserstein distance between
% normalized persistence diagrams is controlled, up to a constant factor, by
% \(\sqrt{\mathrm{SKOT}}\).
This makes
\toDiscuss{\Scorrection{${d_{\mathrm{SK}}}$}}
% \(\sqrt{\mathrm{SKOT}}\)
a computationally
attractive surrogate for repeated \(W_2\)-type evaluations. Moreover, \julien{\Scorrection{$d_{\mathrm{SK}}^{\,2}$}} admits
an explicit cumulative \(L^1\) representation, which yields an isometric Hilbert
embedding of
\toDiscuss{\Scorrection{${d_{\mathrm{SK}}}$}},
% \(\sqrt{\mathrm{SKOT}}\)
and a positive-definite Gaussian \(d_{\mathrm{SK}}\) kernel on
% normalized
persistence diagrams.

\subsection{Related work}\label{sec:related_work}

The literature related to our work can be grouped into two main
families:
\textit{(i)} persistence-diagram metrics, transport approximations,
and Hilbert-compatible representations, and
\textit{(ii)} space-filling curves and one-dimensional orderings.

\noindent\textbf{(i) Persistence-diagram metrics, transport
approximations, and representations:}
The bottleneck and \(p\)-Wasserstein distances compare persistence
diagrams through partial matchings in the birth--death plane, with
unmatched points sent to the diagonal. They underlie stability,
statistical analysis, and barycenter constructions on diagram spaces
\cite{cohen2005stability,mileyko2011probability,
turner2013frechetmeansdistributionspersistence}.
Their evaluation nevertheless requires solving
% a matching
\julien{an assignment}
problem
\cite{Munkres1957,Bertsekas81}.
Geometric auction algorithms improve its practical performance
\cite{kerber2017geometry}, while near-linear approximation schemes
have been developed specifically for the diagram
\(1\)-Wasserstein distance~\cite{chen2021approximation}. Repeated
transport computations nevertheless remain costly in large-scale
means and clustering procedures
\cite{lacombe2018large,vidal2019progressive,lacombe2018large,pont_tvcg22,sisouk_tvcg24,sisouk2026robustbarycenterspersistencediagrams}.

For statistical learning, persistence landscapes and persistence
images map diagrams into linear spaces
\cite{bubenik2015landscapes,adams2017persistenceimages}, while stable
diagram kernels provide RKHS representations
\cite{reininghaus2015multiscale,kusano2016pwgk}. Embedding the
standard diagram metrics themselves into Hilbert spaces is subject to
unavoidable metric-distortion limitations
\cite{carriere2019distortion}.

A complementary learned representation is the domain-oblivious
persistence-diagram hashing of Qin et al.~\cite{qin2022domain}.
Their method maps each diagram to a short binary code and compares the
resulting representations using Hamming distance, yielding
orders-of-magnitude faster comparisons and scalability to large
diagram collections. However, the comparison does not produce an
explicit pointwise assignment, and its agreement with Wasserstein
geometry is evaluated empirically rather than through a deterministic
approximation bound.

The closest transport-based construction to ours is the Sliced
Wasserstein kernel for persistence diagrams
\cite{carriere2017sliced}\julien{, inspired from sliced optimal transport \cite{sisouk_tmlr25}}. It augments diagrams with diagonal
projections, computes one-dimensional \(W_1\) distances after linear
projections, and integrates these costs over all directions. For
diagrams of bounded cardinality, its comparison theorem has the form
\[
\frac{d_1}{2M}
\leq
\mathrm{SW}
\leq
2\sqrt{2}\,d_1,
\]
where \(d_1\) is the first diagram Wasserstein distance and \(M\)
depends on the cardinality bound. Thus, this result concerns the
\(1\)-Wasserstein diagram geometry, whereas our objective is a direct
control of the quadratic diagram distance \(W_2\) by
\(
d_{\mathrm{SK}}
.
\)
\Scorrection{$d_{\mathrm{SK}}^{}$} also replaces the integration over projection directions by one
deterministic injective scalar encoding and one sorting-based
one-dimensional transport problem. The sliced construction relies on a family of direction-dependent
one-dimensional assignments and therefore does not directly return a
single diagram-level assignment in the birth--death plane. In
contrast, the SK ordering yields one explicit diagonal-aware
assignment between the augmented diagrams, which is available to
downstream methods requiring point correspondences.

\noindent\textbf{(ii) Space-filling curves and one-dimensional
orderings:}
Space-filling curves provide continuous surjections from an interval
onto a higher-dimensional domain and have long been used to impose
one-dimensional orderings on multidimensional data
\cite{Sagan1994,Bader2013}.

Hierarchical spatial embeddings provide a related but distinct
approach. For EMD-based image retrieval, Indyk and Thaper embed the
ground geometry into a randomly shifted quadtree metric
\cite{indyk2003fast}. The corresponding tree-based \(W_1\) distance
admits a sparse \(\ell_1\) representation, enabling fast indexed
nearest-neighbour search with approximation guarantees
\cite{backurs2020scalable}. Unlike our construction, this is a
multiscale tree embedding rather than a single scalar selector, and it
is not tailored to the diagonal-assignment structure of persistence
diagrams.

Hilbert-curve orderings have also been
used to construct inexpensive Wasserstein surrogates for empirical
distributions of equal mass
\cite{bernton2019abcwasserstein,li2024hilbertcurve}. These methods
order samples along a space-filling curve to induce a coupling. In
the Hilbert Curve Projection distance, for example, the cost of this
coupling is subsequently evaluated in the original ambient space.
This ambient-cost principle is conceptually closest to our planar
surrogate \(W_\Gamma\).

Our construction is instead tailored to persistence diagrams. It
uses the Sierpi\'nski--Knopp curve directly on the normalized
persistence triangle and incorporates the diagonal through the
standard augmentation by orthogonal projections. The first-hit
selector provides an injective scalar encoding, so distinct points of
the persistence triangle are not identified. Moreover, the
\(1/2\)-H\"older regularity of the curve links assignment between the
scalar codes to assignment in the birth--death plane
\cite{Shchepin2020}. Retaining the one-dimensional assignment cost
defines \Scorrection{$d_{\mathrm{SK}}^{}$} itself, yields the \(W_2\) control proved below, and
leads to the explicit Hilbertian and Gaussian-kernel structures
developed in \autoref{sec:skot_hilbert_kernel}.

\subsection{Contributions}\label{sec:contributions}

This paper makes the following contributions:

\begin{itemize}

\item
\textit{An SK-Wasserstein metric distance \(d_{\mathrm{SK}}\) for persistence
diagrams.}
We introduce \(d_{\mathrm{SK}}\), a distance between normalized
persistence diagrams based on the exact first-hit selector \(\iota\) of the
Sierpi\'nski--Knopp space-filling curve.  In
practice, \(\iota\) is approximated by its finite level-\(L\)
version \(\iota_L\), yielding the numerical approximation
\(d_{\mathrm{SK},L}\). Using the standard diagonal
augmentation, the original two-dimensional partial-assignment problem
is replaced by a one-dimensional \(1\)-Wasserstein assignment problem
between equal-mass atomic encodings on \([0,1]\). We prove that
\(d_{\mathrm{SK}}\) defines a metric on normalized persistence
diagrams.

\item
\textit{A sorting-based \(O(N\log N)\) evaluation algorithm for \(d_{\mathrm{SK},L}\).}
We exploit the monotone
structure of one-dimensional \(W_1\) to derive a closed-form sorting formula for
\Scorrection{$d_{\mathrm{SK}}^{}$}. Given the selector values of the relevant diagram points and their
diagonal projections, \Scorrection{$d_{\mathrm{SK},L}$} between two diagrams of total size \(N\) can be
computed in \(O(N\log N)\) time. 

\item
\textit{A quantitative connection with the diagram
\(2\)-Wasserstein distance.}
Using the \(1/2\)-H\"older regularity of the SK curve, we establish the
universal bound
\[
W_2(X_1,X_2)
\leq
\sqrt{2}\,d_{\mathrm{SK}}(X_1,X_2).
\]
The monotone one-dimensional assignment also induces an admissible
assignment in the birth--death plane. \Scorrection{Thus, the construction returns both a distance value and an explicit
diagonal-aware point correspondence between $X_1$ and $X_2$.} Re-evaluating this assignment
with the quadratic diagram cost defines the planar surrogate
\(W_\Gamma\), which satisfies
\[
W_2
\leq
W_\Gamma
\leq
\sqrt{2}\,d_{\mathrm{SK}}.
\]
Unlike \(d_{\mathrm{SK}}\), \(W_\Gamma\) is not guaranteed to be a
metric, but it provides a numerically tighter surrogate for \(W_2\)
in median.

\item
\textit{An explicit Hilbertian geometry and a positive-definite
Gaussian kernel.}
We derive a cumulative \(L^1\) representation of
\(d_{\mathrm{SK}}^2\), from which we obtain an explicit isometric
embedding of \(d_{\mathrm{SK}}\) into a Hilbert space and a
positive-definite Gaussian \(d_{\mathrm{SK}}\) kernel. These results make the
SK-Wasserstein distance \(d_{\mathrm{SK}}\)  directly compatible with
Euclidean embedding methods and learning pipelines that require vector
representations, and
kernel-based learning methods on normalized persistence diagrams \cite{Guella2020GaussianHilbert,Sriperumbudur2010HilbertEmbedding}.

\item
\textit{Experimental validation and applications on scientific
ensembles.}
We evaluate \(d_{\mathrm{SK}}\) and \(W_\Gamma\) on 12 scientific
collections containing 227 persistence diagrams, with complete pairwise numerical \(W_2\) reference matrices computed with TTK's auction-based
solver. The experiments demonstrate
numerical stabilization at \(L=30\), strong preservation of the
\(W_2\) geometry, and a median \(626\times\) speedup for
\(d_{\mathrm{SK},30}\). They further demonstrate the practical use of
the Hilbertian and Gaussian-kernel representations of \(d_{\mathrm{SK}}\) for clustering and
contiguous segmentation of ordered diagram collections.

\end{itemize}

%% file: preliminaries_Section2_.tex
\section{Preliminaries}\label{sec:preliminaries}

\noindent
This section presents the theoretical background required for our work. It formalizes our input data, introduces persistence diagrams and a standard metric for their comparison, and specifies the SK space-filling curve used in our definition of \Scorrection{$d_{\mathrm{SK}}^{}$}. We refer the reader to standard textbooks~\cite{book,oudot2015persistence} for an introduction to TDA, and to~\cite{Bader2013} for a general presentation of space-filling curves.

\subsection{Input data}\label{sec:input_data}

We consider as input a finite collection of PL scalar fields
\(
\mathcal F=\{f_i:\mathcal M_i\to\mathbb R\}_{i=1}^{n_{\mathcal F}},
\)
where each \(\mathcal M_i\) is a PL \((d_i)\)--manifold, with \(d_i\le 3\) in our applications. 

For a fixed index \(i\in\NNB\) and a threshold
\(\variablew\in\RNB\), we denote by
\(
\sublevelset( fi)
=
 f_{i}^{-1}\bigl((-\infty,\variablew]\bigr)
\)
the sublevel set of \( f_i\) at value \(\variablew\).
When \(\variablew\) increases from \(-\infty\) to \(+\infty\), the topology of
\(\sublevelset( f_i)\) remains unchanged except when
\(\variablew\) crosses a critical value of \( f_i\). The
critical points
\(\criticalPointOne\in\manifold_{ f_{i}}\) are classified
according to their Morse index \(\morseIndex(\criticalPointOne)\): index \(0\)
corresponds to minima, index \(1\) to \(1\)-saddles, index
\(\dimension_{\manifold_{ f_{i}}}-1\) to
\((\dimension_{\manifold_{ f_{i}}}-1)\)-saddles, and index
\(\dimension_{\manifold_{ f_{i}}}\) to maxima. In practice, as
is standard~\cite{edelsbrunner2001hierarchical,edelsbrunner1990simulation}, we
assume that all critical points are isolated and non-degenerate. By the Elder rule~\cite{book}, each topological feature created during the
sublevel-set sweep, such as a connected component, a cycle, or a void, is
associated with a pair of critical points
\((\criticalPointOne,\criticalPointTwo)\). The first critical point
\(\criticalPointOne\) creates the feature, while the second critical point
\(\criticalPointTwo\) destroys it. These points satisfy
\(
 f_{i}(\criticalPointOne)
<
 f_{i}(\criticalPointTwo),
\text{ and }
\morseIndex(\criticalPointOne)
=
\morseIndex(\criticalPointTwo)-1.
\)
The feature is therefore born at the older critical point
\(\criticalPointOne\) and dies at the younger critical point
\(\criticalPointTwo\), and the pair
\((\criticalPointOne,\criticalPointTwo)\) is called a persistence pair. For
example, when two connected components merge at a critical point
\(\criticalPointTwo\), the component that was created most recently disappears,
whereas the older component survives. Each persistence pair is represented by its birth and death values,
\(
(\birthPair,\deathPair)
=
\bigl(
 f_{i}(\criticalPointOne),
 f_{i}(\criticalPointTwo)
\bigr).
\) The persistence diagram associated with \(f_i\) is the finite multiset
of all such pairs,
\[
X_i^{\mathrm{raw}}
=
\{(b_i^1,d_i^1),\dots,(b_i^{m_i},d_i^{m_i})\}.
\]

\noindent As illustrated in
\autoref{fig:noNoisy-NoisyPersistenceDiagrams}, features with large topological
significance correspond to points \((\birthPair,\deathPair)\) located far from
the diagonal
\(
\Delta
=
\{(\birthPair,\deathPair)\in\RNB^2:\birthPair=\deathPair\}.
\)
Such points have a large lifespan
\(\deathPair-\birthPair\), called their persistence. Conversely, pairs produced
by small-amplitude noise typically have short lifespans and therefore lie close
to the diagonal.

\begin{figure}[t]
  \centering
  \includegraphics[width=0.8\linewidth]{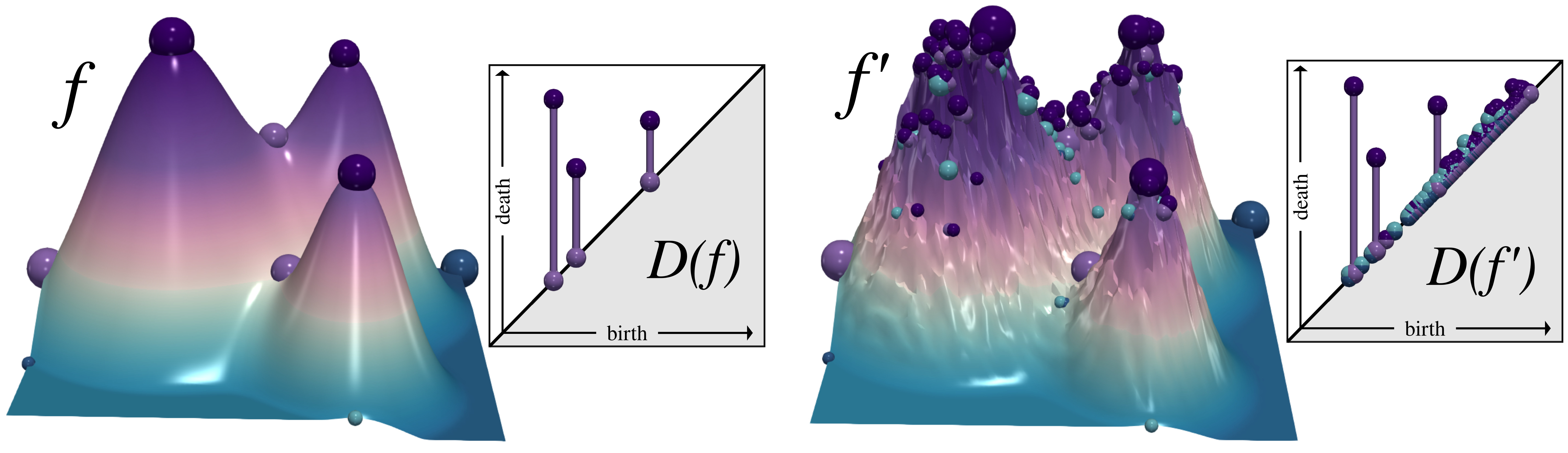}
  \caption{Persistence diagrams \(\PDS(\functionf)\) and
\(\PDS(\functionf')\) associated respectively with the noise-free
scalar field \(\functionf\) (left) and its noisy counterpart
\(\functionf'\) (right). In both diagrams, the three dominant peaks
appear as highly persistent pairs, while the numerous points close to
the diagonal in \(\PDS(\functionf')\) correspond to low-persistence
features introduced by the background noise.}
  \label{fig:noNoisy-NoisyPersistenceDiagrams}
\end{figure}
\noindent

\Scorrection{In order to fit the setting in which \Scorrection{$d_{\mathrm{SK}}^{}$} is defined, that is of normalized persistence diagrams, while preserving the relative amplitude of features across the collection, we normalize all
persistence diagrams of the collection jointly: assuming that at least one diagram is
nonempty, let
\[
v_{\min}
:=
\min_{\substack{1\leq i\leq n_{\mathcal F}\\1\leq r\leq m_i}}
b_i^r,
\qquad
v_{\max}
:=
\max_{\substack{1\leq i\leq n_{\mathcal F}\\1\leq r\leq m_i}}
d_i^r.
\]
We apply the same increasing affine map to every point of every
diagram and define
\[
 X_i
:=
\left\{
\left(
\frac{b_i^r-v_{\min}}{v_{\max}-v_{\min}},
\frac{d_i^r-v_{\min}}{v_{\max}-v_{\min}}
\right)
\right\}_{r=1}^{m_i}.
\]
Thus, every \( X_i\) is supported in the normalized
persistence triangle
\[
\Atri
:=
\{(x,y)\in[0,1]^2:x\leq y\}.
\] Since the same affine transformation is used
for the entire collection, the relative birth, death, and persistence
scales among its diagrams are preserved. This normalization depends
only on the persistence diagrams and can therefore also be applied to
diagrams obtained from point clouds or other filtered data. If all diagrams are already supported in \(\Atri\), no further normalization is required and we
simply set \(X_i=X_i^{\mathrm{raw}}\).}

\Scorrection{Thanks to the global normalization, the persistence diagram \(X_i\) can also be viewed as an integer atomic measure supported in the normalized persistence triangle
\(
\Atri.
\)
More precisely,
\[
X_i=\sum_{r=1}^{m_i}\delta_{(b_i^r,d_i^r)},\,
(b_i^r,d_i^r)\in\Atri\setminus\Delta,
\]}

\Scorrection{where, for simplicity, \(b_i^r\) and \(d_i^r\) now denote the
normalized birth and death coordinates. Repeated points in the
multiset are represented by repeated Dirac masses.}

\Scorrection{Although our experiments use persistence diagrams derived from PL
scalar fields, the framework developed in this paper only requires a
collection of persistence diagrams, which can be jointly normalized
regardless of their data source. It therefore applies directly to
diagrams obtained from other filtered data, such as point clouds
equipped with Vietoris--Rips filtrations.}

In what follows, we
enumerate the $\persistencePairNumberBisBis$ points of a persistence diagram
$\PD$ as
$\PD=\{\persistancepair^{1},\dots,\persistancepair^{\persistencePairNumberBisBis
}\}$, and denote $\PD_{\varnothing}=\{\}$ the empty persistence diagram.

\subsection{\julien{Metric for} persistence diagrams}
\label{sec:persistence_diagrams}

Let
\[
\PD_1=\{\persistancepair_1^1,\dots,
\persistancepair_1^{\persistencePairNumberBisBis_1}\},
\qquad
\PD_2=\{\persistancepair_2^1,\dots,
\persistancepair_2^{\persistencePairNumberBisBis_2}\}
\]
be two persistence diagrams. Since their numbers of off-diagonal points may differ, we augment each diagram with diagonal projections of the off-diagonal points of the other diagram. More precisely, we set
\(
\PD_1^\ast
=
\PD_1
\cup
\bigl\{
\Pi(\persistancepair):
\persistancepair\in \PD_2\setminus\Delta
\bigr\},
\PD_2^\ast
=
\PD_2
\cup
\bigl\{
\Pi(\persistancepair):
\persistancepair\in \PD_1\setminus\Delta
\bigr\},
\)
where the diagonal projection is given by
\[
\Pi(\birthPair,\deathPair)
=
\left(
\frac{\birthPair+\deathPair}{2},
\frac{\birthPair+\deathPair}{2}
\right).
\]
The two augmented diagrams have the same cardinality, which we denote by
\(
\persistencePairNumberBisBis
:=
|\PD_1^\ast|
=
|\PD_2^\ast|.
\)
We write
\[
\PD_1^\ast
=
\{
\persistancepair_{*1}^{1},
\dots,
\persistancepair_{*1}^{\persistencePairNumberBisBis}
\},
\qquad
\PD_2^\ast
=
\{
\persistancepair_{*2}^{1},
\dots,
\persistancepair_{*2}^{\persistencePairNumberBisBis}
\},
\]
with multiplicities. Let
\(
\intervalI_{\persistencePairNumberBisBis}
=
\{1,\dots,\persistencePairNumberBisBis\}.
\)
A bijection
\(
\psi:\intervalI_{\persistencePairNumberBisBis}
\to
\intervalI_{\persistencePairNumberBisBis}
\)
defines a one-to-one assignment between the points of the two augmented diagrams. We use the diagonal-aware squared cost
\[
\floorCost(\persistancepair,\persistancepairTwo)
=
\begin{cases}
0,
& \persistancepair\in\Delta
\text{ and }
\persistancepairTwo\in\Delta,\\[2pt]
\|\persistancepair-\persistancepairTwo\|_2^2,
& \text{otherwise},
\end{cases}
\]
where \(\|\cdot\|_2\) is the Euclidean norm in \(\RNB^2\). The corresponding \(2\)-Wasserstein distance is
\begin{equation}
\wasserstein_2(\PD_1,\PD_2)
=
\min_{\psi:\,\intervalI_{\persistencePairNumberBisBis}
\to
\intervalI_{\persistencePairNumberBisBis}
\ \mathrm{bijective}}
\left(
\sum_{\variablej=1}^{\persistencePairNumberBisBis}
\floorCost
\bigl(
\persistancepair_{*1}^{\variablej},
\persistancepair_{*2}^{\psi(\variablej)}
\bigr)
\right)^{1/2}.
\label{eq:wasserstein}
\end{equation}

The diagonal augmentation encodes the standard convention that features may be deleted or created by matching off-diagonal points to the diagonal. Thus, \( \wasserstein_2 \) measures the minimum root-sum-of-squares cost required to transport one persistence diagram to the other, while assigning zero cost to matches between diagonal points; see \autoref{fig:2-2-L2-PD-Wasserstein_Distance}.

\begin{figure}[t]
  \centering
  \includegraphics[width=0.78\linewidth]{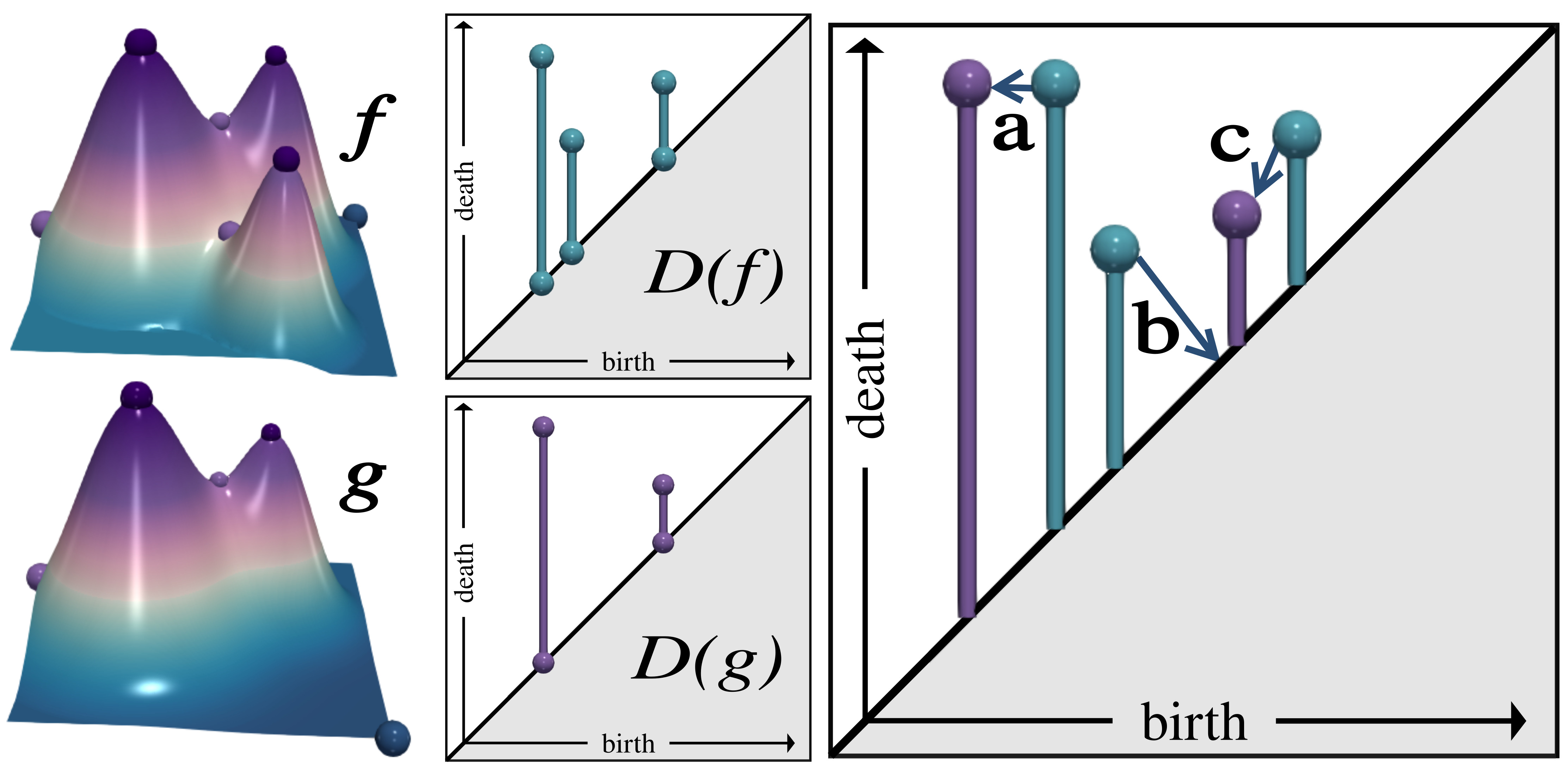}
  \caption{Left: two synthetic scalar fields, \(\functionf\) (top) and
\(\functiong\) (bottom). Center: their corresponding persistence
diagrams, \(\PDS(\functionf)\) and \(\PDS(\functiong)\). Right: the
optimal \(2\)-Wasserstein assignment \(\psi\), between
$\PDS(\functionf)$ and $\PDS(\functiong)$ visualized by arrows.
For clarity, the diagonal augmentations are omitted and only the
off-diagonal matchings are represented. The
squared \(2\)-Wasserstein distance is the sum of the squared arrow
lengths:
\(
a^{2}+b^{2}+c^{2}
=
\wasserstein_{2}\bigl(\PDS(\functionf),\PDS(\functiong)\bigr)^{2}.
\)}
  \label{fig:2-2-L2-PD-Wasserstein_Distance}
\end{figure}

\subsection{Sierpiński--Knopp space-filling curve}
\label{sec:skot_encoding}

We first present the finite-level recursive SK construction and the numerical
selector \(\iota_L\), before introducing the exact SK space-filling curve \(S\) and
its first-hit selector \(\iota\) as their limiting counterparts.

\smallskip
\noindent
\textbf{Finite-level recursive SK construction and numerical selector.}

Starting from the persistence triangle
\[
\Atri
=
\operatorname{conv}\{(0,0),(1,1),(0,1)\},
\]
the level-\(0\) construction consists of the single cell \(\Atri\),
associated with the parameter interval \([0,1]\). We equip this root cell with the ordered hypotenuse
\(((0,0),(1,1))\), whose endpoints are its entry and exit vertices, respectively. At each level-$(\ell+\!1)$ refinement
step, every current level-$\ell$ right isosceles triangular cell is subdivided into
two similar level-$(\ell+\!1)$ right isosceles cells, while its associated level-$\ell$ dyadic interval is split into two level-$(\ell+\!1)$ dyadic subintervals.

More precisely, let
\(
T=\operatorname{conv}\{r,p,q\}
\)
be a level-\(\ell\) cell, where \(p\) and \(q\) are respectively the entry and exit endpoints
of its hypotenuse, \(r\) is the remaining vertex, and
\(
m:=\frac{p+q}{2}
\)
is the midpoint of the hypotenuse. We subdivide \(T\) into
\(
T^{(0)}
=
\operatorname{conv}\{r,p,m\},
T^{(1)}
=
\operatorname{conv}\{r,m,q\}.
\)
The two children \(T^{(0)}\) and \(T^{(1)}\) inherit the ordered
hypotenuses \((p,r)\) and \((r,q)\), respectively. The SK traversal designates the level-$(\ell+1)$ cell $T^{(0)}$ as the first-visited child and the level-$(\ell+1)$ cell $T^{(1)}$ as
the second-visited child. Accordingly, if $I=[a,b]$ denotes the current level-$\ell$ dyadic interval
associated with $T$, we associate $T^{(0)}$ with the left half of $I$, yielding
a level-$(\ell+1)$ dyadic subinterval
\(
I^{(0)}:=[a,\ \tfrac{a + b}{2}],
\)
and we associate $T^{(1)}$ with the right half, yielding a level-$(\ell+1)$
dyadic subinterval
\(
I^{(1)}:=[\tfrac{a+b}{2},\ b].
\)

After \(L\) refinements, this construction yields \(2^L\) ordered
triangular cells
\(
T_0^{(L)},\ldots,T_{2^L-1}^{(L)},
\)
associated respectively with the dyadic intervals
\(
I_k^{(L)}
=
\left[
\frac{k}{2^L},
\frac{k+1}{2^L}
\right],
k=0,\ldots,2^L-1.
\)

\autoref{fig:sk_curve_levels} visualizes the cell visitation order
induced by this recursive construction for \(L=1,\ldots,9\).

\begin{figure}[t]
  \centering
  \includegraphics[width=.98\linewidth]
  {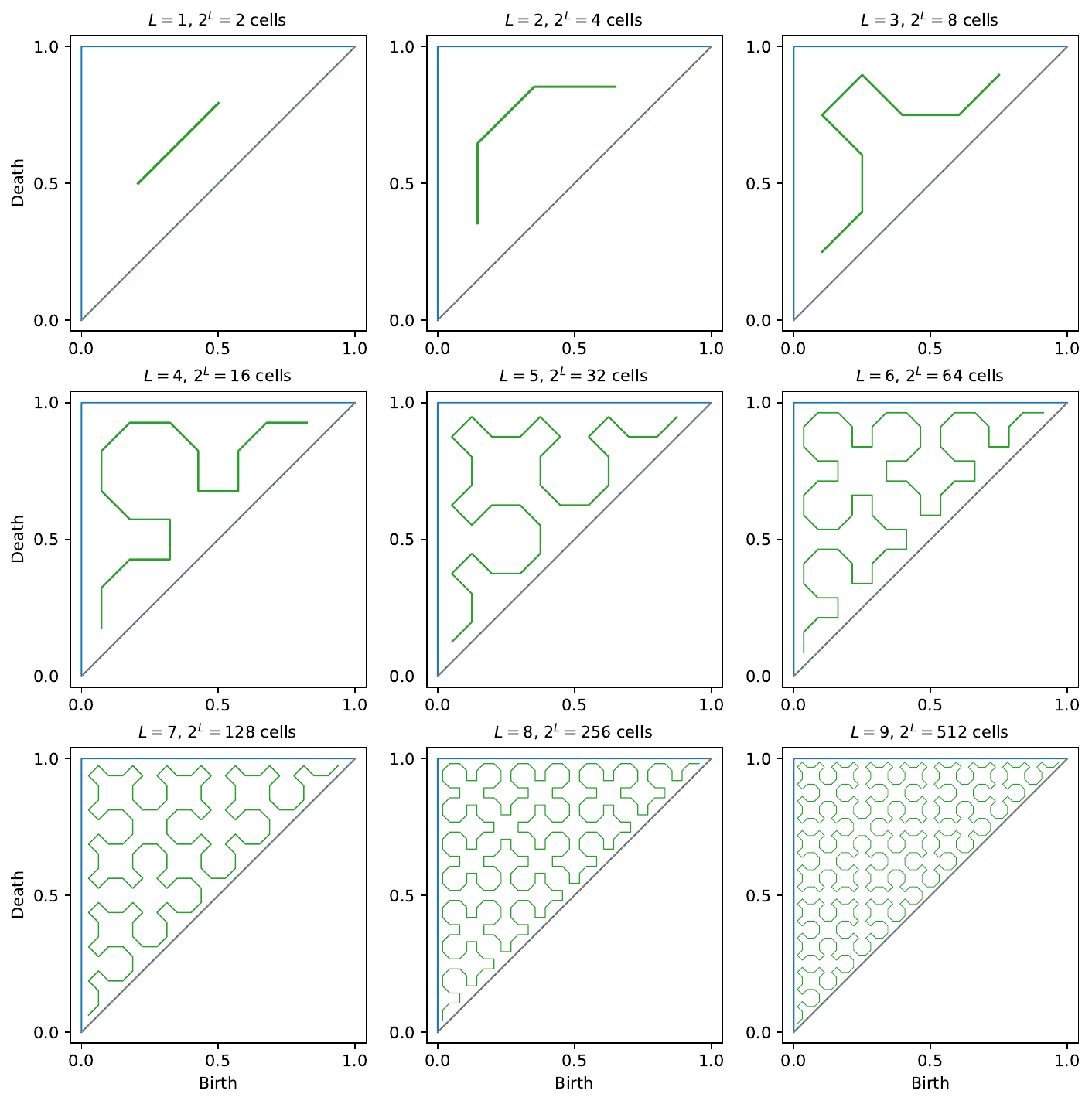}
  \caption{Finite-level SK cell traversals for \(L=1,\ldots,9\). At level
  \(L\), the polyline connects the incenters of the \(2^L\)
  triangular cells in visitation order, which coincides with the
  order of their associated dyadic subintervals of \([0,1]\).
  }
  \label{fig:sk_curve_levels}
\end{figure}

For a point \(z\in\Atri\), let
\(
k_L(z)
:=
\min
\left\{
k\in\{0,\ldots,2^L-1\}:
z\in T_k^{(L)}
\right\}.
\)
The minimum implements the first-visited-cell convention when \(z\)
belongs to the common boundary of several cells. We define the
finite-level selector by
\(
\iota_L(z)
:=
\frac{k_L(z)}{2^L},
\)
that is, the left endpoint of the first level-\(L\) dyadic interval
whose associated cell contains \(z\).

In practice, \(k_L(z)\) is computed by descending the binary
refinement tree. Initialize $k\leftarrow 0$ and set the
current right isosceles triangle to the root persistence triangle
$\Atri$. For $\ell=1,\dots,L$, form the two right isosceles child
triangles and decide which child contains $z$. This membership test is carried
out in constant time with an oriented-area predicate: with
$\mathrm{cross}(u,v):=u_x v_y-u_y v_x$, for a triangle $(a,b,c)$ one evaluates
the signs of $\mathrm{cross}(z-a,b-a)$, $\mathrm{cross}(z-b,c-b)$, and
$\mathrm{cross}(z-c,a-c)$ (with a small tolerance), and $z$ lies inside if and
only if these signs are consistent (boundary included). If $z$ lies in the
first-visited child, set $k\leftarrow 2k$; otherwise set $k\leftarrow 2k+1$;
update the current triangle accordingly and iterate. If $z$ lies on the common
boundary of the two children (up to numerical tolerance), we select the
first-visited child, consistently
with the definition of \(k_L(z)\). After \(L\) steps, the computed index satisfies
\[
k=k_L(z).
\]
Hence,
\[
I_k^{(L)}
=
\left[
\frac{k}{2^L},
\frac{k+1}{2^L}
\right]
\]
is the dyadic interval associated with the first-visited level-\(L\)
cell containing \(z\). The algorithm therefore returns
\[
\iota_L(z)=\frac{k}{2^L},
\]
the left endpoint of \(I_k^{(L)}\). Since each refinement level
requires constant work, the overall cost is \(O(L)\) per query point.

\smallskip
\noindent
\textbf{Exact SK curve and first-hit selector.}

The ordered cells determine a continuous polygonal map at each
refinement level: for \(k=0,\ldots,2^L-1\), let
\(e_k^{(L)}\) and \(q_k^{(L)}\) denote respectively the entry and exit
vertices of \(T_k^{(L)}\), that is, the two endpoints of its
hypotenuse ordered according to the SK traversal. For \(t\in I_k^{(L)}\), let
\[
\theta:=2^Lt-k\in[0,1]
\]
and define
\[
S_L(t)
:=
(1-\theta)e_k^{(L)}
+
\theta q_k^{(L)}.
\]
Thus, \(S_L\) traverses the level-\(L\) cell \(T_k^{(L)}\) linearly from its entry
vertex to its exit vertex. Using these entry and exit vertices, let
\[
S_L:[0,1]\to\Atri
\]
denote the polygonal approximation obtained by traversing
each level-\(L\) cell \(T_k^{(L)}\) from its entry vertex
\(e_k^{(L)}\) to its exit vertex \(q_k^{(L)}\) over the associated
dyadic interval \(I_k^{(L)}\). The SK visitation convention ensures that consecutive segments share
their endpoints, so \(S_L:[0,1]\to\Atri\) is continuous. Then, by the convergence result for the recursive SK construction
\cite{Sagan1994,Bader2013}, the sequence
\((S_L)_{L\geq0}\) converges uniformly, as \(L\to\infty\), to a
continuous surjection
\[
S:[0,1]\to\Atri,
\]
called the Sierpi\'nski--Knopp space-filling curve of $\Atri$. Since $S$ is surjective, every point $z\in\Atri$ is hit at least once. Therefore, we can define the associated first-hit selector
\[
\iota(z):=\min S^{-1}(\{z\})=\min\{t\in[0,1]:S(t)=z\}\in[0,1].
\]

Because $S$ is surjective, the fiber $S^{-1}(\{z\})\neq\varnothing,\forall\,z\in\Atri$, and since $S$ is
continuous and $\{z\}$ is closed, $S^{-1}(\{z\})$ is closed in $[0,1]$.
Therefore $S^{-1}(\{z\})$ is a non-empty closed subset of the compact $[0,1]$, hence it
attains its minimum, and $\iota(z)=\min S^{-1}(\{z\})$ is well defined. As a result, the finite-level selectors \(\iota_L\) converge
pointwise to the first-hit selector \(\iota\) as \(L\to\infty\)
\cite{Bader2013,Sagan1994}. Moreover,
\(
S(\iota(z))=z,\forall\,z\in\Atri,
\)
hence $\iota:\Atri\to[0,1]$ is injective (indeed, $\iota(z)=\iota(z') \implies z=S(\iota(z))=S(\iota(z'))=z' $), this will later guarantee that the
\Scorrection{assignment} problem on $[0,1]$ between the one-dimensional diagram-encoded measures induces a metric on
normalized persistence diagrams. The selector $\iota$ is also Borel measurable. Indeed, for any $t\in[0,1]$,
\[
\iota^{-1}([0,t])=\{z\in\Atri:\iota(z)\le t\}=S([0,t]),
\]
because $\iota(z)\le t$ if and only if $z$ is hit by $S$ at some parameter
$s\le t$. Since $S$ is continuous, $S([0,t])$ is a compact subset of $\Atri$,
hence closed in $\Atri$, and therefore Borel. Consequently, $\iota$ is Borel measurable; this measurability will ensure that our first-layer
one-dimensional encoding of normalized persistence diagrams is well defined.

\smallskip
\noindent
\textbf{H\"older regularity.} ~\cite{Shchepin2020} proves that on a unit-area right isosceles triangle
$T_{\text{ref}}$ the reference SK space-filling curve $S_{\text{ref}}:[0,1]\to T_{\text{ref}}$ is
$1/2$-H\"older; more precisely,
\[
\|S_{\text{ref}}(t)-S_{\text{ref}}(s)\|_2^2 \le 4\,|t-s|
\qquad\forall\,s,t\in[0,1].
\]

Our persistence triangle $\Atri$ is a right isosceles triangle of area $1/2$,
whereas the reference right isosceles triangle $T_{\text{ref}}$ in~\cite{Shchepin2020} has
area $1$. Therefore, $\Atri$ is a scaled copy of $T_{\text{ref}}$ with scale
$\alpha=1/\sqrt{2}$, up to an isometry of $\mathbb{R}^2$ (i.e., a rotation or reflection,
followed by a translation). Let $\varphi:\,T_{\text{ref}}\to\Atri$ be such a similarity, i.e., $\varphi(x)=\alpha R x+b$, where $R$ is
orthogonal and $b$ is a translation.

Choosing the orientation of $S_{\text{ref}}$ consistently with the SK traversal introduced above, our SK space-filling curve \(S\) on $\Atri$ is obtained from the reference
curve \(S_{\text{ref}}\) through the similarity
\[
S=\varphi\circ S_{\text{ref}}:[0,1]\to\Atri.
\]

Therefore, for all $s,t\in[0,1]$, using the fact that
$R$ is orthogonal (so $\|Rw\|_2=\|w\|_2$), we obtain
\[
\begin{aligned}
\|S(t)-S(s)\|_2^2
&=\|\varphi(S_{\text{ref}}(t))-\varphi(S_{\text{ref}}(s))\|_2^2 \\
&=\|\alpha R(S_{\text{ref}}(t)-S_{\text{ref}}(s))\|_2^2 \\
&=\alpha^2\,\|R(S_{\text{ref}}(t)-S_{\text{ref}}(s))\|_2^2 \\
&=\alpha^2\,\|S_{\text{ref}}(t)-S_{\text{ref}}(s)\|_2^2 \\
&\le 4\alpha^2\,|t-s|
=2\,|t-s|.
\end{aligned}
\]
Thus $S$ satisfies the $1/2$-H\"older bound
$\|S(t)-S(s)\|_2^2 \le 2\,|t-s|$ for all $s,t\in[0,1]$, and we set
$C_{\mathrm{SK}}:=2$ in the sequel. This regularity will later allow us to relate \Scorrection{$d_{\mathrm{SK}}^{}$} to the $2$-Wasserstein
distance between normalized persistence diagrams.

%% file: continuous_edit_distance_between_time_varying_persistence_diagrams_Section3_.tex
\section{Sierpiński--Knopp Wasserstein distance between normalized persistence diagrams}\label{sec:skot}

\noindent
In this section, we introduce the Sierpiński--Knopp Wasserstein distance \Scorrection{$d_{\mathrm{SK}}^{}$} between normalized persistence diagrams. First, we provide an overview of its construction and summarize its key
properties.
Then, we detail the mathematical formulation of \Scorrection{$d_{\mathrm{SK}}^{}$}, which is built in two conceptual layers:
(i) a deterministic one-dimensional encoding of normalized diagram points induced by the SK space-filling curve on the persistence triangle, and (ii) a one-dimensional
\Scorrection{assignment} problem on $[0,1]$ between the induced encodings. Next, we show how \Scorrection{$d_{\mathrm{SK}}^{}$} can be evaluated efficiently in practice through a
sorting-based formula, leading to an $O(N\log N)$ procedure for normalized diagrams of total
size $N$. Finally, we relate \Scorrection{$d_{\mathrm{SK}}^{}$} to the 2-Wasserstein distance between normalized persistence diagrams.

\subsection{Overview}\label{sec:skot_overview}

Here, we provide a high-level overview of the 
the \Scorrection{$d_{\mathrm{SK}}^{}$} distance. We first describe how normalized diagrams are handled with respect to the
diagonal and encoded onto the unit interval using the SK space-filling curve. We then define \Scorrection{$d_{\mathrm{SK}}^{}$}
as a one-dimensional \Scorrection{assignment} problem on these encodings, and summarize its main properties.

\medskip
The \Scorrection{$d_{\mathrm{SK}}^{}$} construction combines two conceptual layers.

\smallskip
\noindent
\textbf{(i) A diagonal-aware 1D encoding of normalized persistence diagrams.}
Let $\PD=\{\persistancepair^{1},\dots,\persistancepair^{\persistencePairNumberBisBis
}\}$, be a normalized persistence diagram. We can see $\PD=\sum_{k=1}^K \,\delta_{\persistancepair^{k}}$ as an integer atomic measure on $\Atri$. Using the associated first-hit selector
$\iota(z)=\min S^{-1}(\{z\})$ of the SK space-filling curve we can map points in $\Atri$ to scalar values in $[0,1]$. This yields two integer atomic measures on
$[0,1]$:
\[
\mu_X:=\iota_{\#}X,
\qquad
\nu_X:=\iota_{\#}(\Pi_{\#}X),
\]
encoding respectively the off-diagonal atoms of $X$ and their diagonal
projections, both have equal total mass $|X|$.

\smallskip
\noindent
\textbf{(ii) \Scorrection{$d_{\mathrm{SK}}^{}$} distance between normalized persistence diagrams as a 1D \Scorrection{assignment} between their augmented encodings.}

We compare normalized diagrams by evaluating the classical 1D $1$-Wasserstein distance between their
augmented 1D encodings on $([0,1],|\cdot|)$. For two diagrams $X_1,X_2$ we define
\[
\mathrm{\Scorrection{d_{\mathrm{SK}}^{}}}(X_1,X_2)
:=\sqrt{W_1\big(\mu_{X_1}+\nu_{X_2},\ \mu_{X_2}+\nu_{X_1}\big)}.
\]
Equivalently, setting $\bar X_1:=X_1+\Pi_{\#}X_2$ and $\bar X_2:=X_2+\Pi_{\#}X_1$
(which have equal mass), one has
$\mathrm{\Scorrection{d_{\mathrm{SK}}^{}}}(X_1,X_2)=\sqrt{W_1(\iota_{\#}\bar X_1,\iota_{\#}\bar X_2)}$. This augmentation mirrors the standard persistence-diagram \Scorrection{assignment} convention in
which unmatched off-diagonal points may be matched to the diagonal. In \Scorrection{$d_{\mathrm{SK}}^{}$}, this effect is encoded by adding diagonal projections to the compared measures; the difference is that the \Scorrection{resulting assignment problem} is no longer solved directly in the persistence triangle, but through a 1D 1-Wasserstein problem on $[0,1]$.

\autoref{fig:dsk_assignment_pipeline} provides a finite-level illustration of the two-layer
construction at \(L=6\): diagram points and their
diagonal projections are encoded on \([0,1]\), the augmented scalar
lists are matched monotonically, and the resulting pairing is lifted
back to the normalized persistence triangle.

\begin{figure}[t]
  \centering
  \includegraphics[width=.98\linewidth]
  {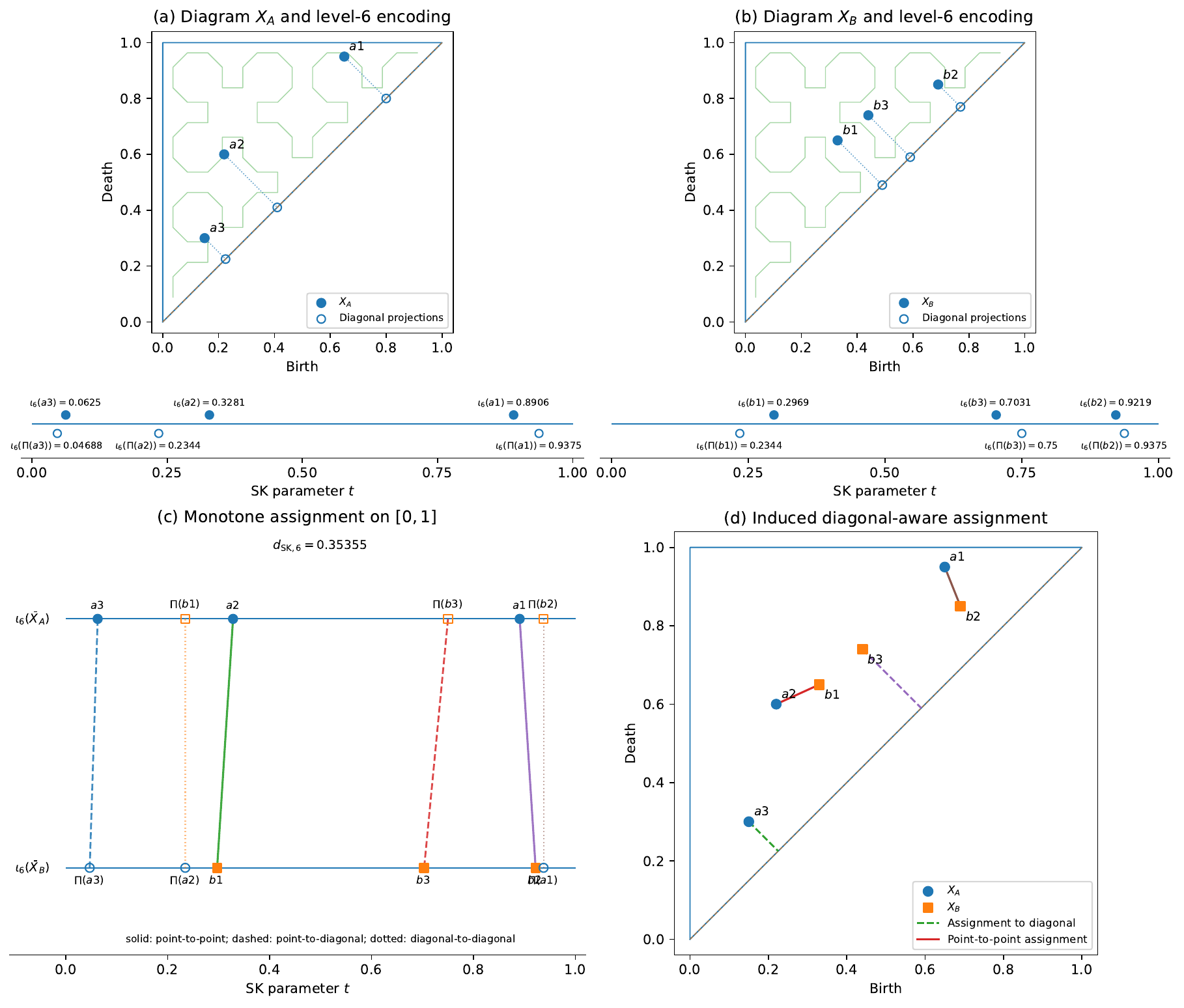}
  \caption{
  Finite-level construction of \(d_{\mathrm{SK},6}\) on two toy
  normalized persistence diagrams. (a)--(b) The points of \(X_A\)
  and \(X_B\), together with their diagonal projections, are encoded
  by the level-\(6\) selector \(\iota_6\). (c) Sorting the two
  augmented scalar lists yields the monotone optimal assignment on
  \([0,1]\), from which \(d_{\mathrm{SK},6}\) is computed.
  (d) Retaining point identities induces an explicit partial
  assignment between \(X_A\) and \(X_B\), with unmatched points
  assigned to \(\Delta\). Diagonal-to-diagonal pairs are omitted from
  the final panel.
  }
  \label{fig:dsk_assignment_pipeline}
\end{figure}

\medskip
\noindent
\textbf{Key properties.}
First, \Scorrection{$d_{\mathrm{SK}}^{}$} defines a metric on the space of normalized persistence diagrams, and its square $\mathrm{\Scorrection{d_{\mathrm{SK}}^{\,2}}}$ does as well (see \autoref{app:thm:SKOT_metric} and \autoref{app:cor:dSK_metric}). Second, because one-dimensional $W_1$ has a
closed-form monotone optimal coupling, \Scorrection{$d_{\mathrm{SK}}^{}$} between two normalized persistence diagrams $X_1$ and $X_2$ can be evaluated by sorting two length-$N$ lists of $\iota$-values (with $N=|X_1|+|X_2|$), giving an $O(N\log N)$ evaluation procedure once the relevant $\iota$-values have been computed. Third, thanks to the $1/2$-H\"older regularity of the SK space-filling curve, the Wasserstein distance $W_{2}$ is controlled, up to a constant factor, by $\mathrm{\Scorrection{d_{\mathrm{SK}}^{}}}$, which makes $\mathrm{\Scorrection{d_{\mathrm{SK}}^{}}}$ a fast surrogate metric for repeated $W_{2}$ evaluations on normalized persistence diagrams. Finally, \Scorrection{$d_{\mathrm{SK}}^{\,2}$} admits an explicit $L^1$ signature representation. In particular, there is an explicit Hilbert embedding $\Phi$ into $L^2$ such that $\mathrm{\Scorrection{d_{\mathrm{SK}}^{}}}(X_1,X_2)=\lVert\Phi(X_1)-\Phi(X_2)\rVert_{L^2}$, and the associated Gaussian kernel $\exp(-\mathrm{\Scorrection{d_{\mathrm{SK}}^{}}}(X_1,X_2)^2/(2\sigma^2))$ is positive definite. Thus,
\Scorrection{\(d_{\mathrm{SK}}\)} supports both Hilbertian embedding
methods and kernel-based learning.

\subsection{Deterministic diagonal-aware 1D encoding of normalized persistence diagrams}
\label{sec:skot_layer1}

The first layer of \Scorrection{$d_{\mathrm{SK}}^{}$} maps each normalized persistence diagram to a diagonal-aware one-dimensional encoding on the unit interval. This reduction is entirely driven by the fixed SK space-filling curve introduced above and is designed to preserve both the location of off-diagonal atoms and their interaction with the diagonal $\Delta$:

A normalized persistence diagram is represented as a finite integer atomic measure
\[
X=\sum_{r=1}^m a_r\,\delta_{z_r},
\qquad
z_r\in\Atri,\ \ a_r\in\mathbb{N},
\]
more precisely, with all atoms lying in $\Atri\setminus\Delta$.

Given the SK curve \(S:[0,1]\to\Atri\), each persistence pair \(z_{r}\in\Atri\) is assigned its first-hit value via the selector \(\iota:\Atri\to[0,1]\):
\[
\iota(z_{r}):=\min\{t\in[0,1]:S(t)=z_{r}\}.
\]
The minimum is well defined, and \(\iota\) is injective, as shown in~\autoref{sec:skot_encoding}. This injectivity is important: distinct points of the persistence triangle receive distinct scalar values, so the one-dimensional encoding does not collapse geometric information. In particular, the two sets \(\iota(\Atri\setminus\Delta)\) and \(\iota(\Delta)\) are disjoint, which means that off-diagonal atoms and diagonal atoms remain separated after encoding. The selector is also Borel measurable, so pushing atomic measures through \(\iota\) is well defined.

This leads to two associated integer atomic measures on \([0,1]\):
\[
\mu_X:=\iota_{\#}X,
\qquad
\nu_X:=\iota_{\#}(\Pi_{\#}X).
\]
more explicitly, at the level of atoms,
\[
\mu_X=\sum_{r=1}^m a_r\,\delta_{\iota(z_r)},
\qquad
\nu_X=\sum_{r=1}^m a_r\,\delta_{\iota(\Pi(z_r))}.
\]
The first measure \(\mu_X\) records where the diagram places mass along the SK parameterization, while the second measure \(\nu_X\) records where the same mass would lie after projection onto the diagonal. Both measures have the same total mass, namely \(|X|\). Thus, the first layer of \Scorrection{$d_{\mathrm{SK}}^{}$} does not produce a single list of scalar values, but rather a diagonal-aware pair of one-dimensional measures. While \(\mu_X\) alone already determines \(X\) because \(\iota\) is injective, the companion measure \(\nu_X\) is essential because persistence-diagram comparisons must also account for matches to the diagonal.

It is convenient to combine these two measures into the signed measure
\[
\sigma_X:=\mu_X-\nu_X,
\qquad
\sigma_X([0,1])=0.
\]
This signed object can be viewed as a compact one-dimensional signature of the diagram: positive atoms represent off-diagonal contributions carried by \(X\), while negative atoms represent the corresponding diagonal cancellations. Because \(\mu_X\) and \(\nu_X\) are supported on disjoint subsets of \([0,1]\), this signed encoding still retains the full information of the original normalized diagram. \Scorrection{As shown in \autoref{sec:skot_hilbert_kernel}, this representation is
the key object behind the cumulative \(L^1\) formulation of the squared distance
\(d_{\mathrm{SK}}^{\,2}\), from which the Hilbertian representation
of \(d_{\mathrm{SK}}\) and its Gaussian kernel follow.}

At this stage, no \Scorrection{assignment} problem has yet been solved. The first layer simply maps each diagram \(X\) to the pair of one-dimensional measures
\[
X\longmapsto(\mu_X,\nu_X).
\]

\subsection{\Scorrection{$d_{\mathrm{SK}}^{}$} distance between
normalized persistence diagrams as a 1D \Scorrection{assignment} between their
augmented encodings}
\label{sec:skot_layer2}

The second layer of \Scorrection{$d_{\mathrm{SK}}^{}$} turns the diagonal-aware one-dimensional encodings introduced above into a distance between normalized persistence diagrams. Let \(X_1\) and \(X_2\) be two normalized persistence diagrams, and let
\[
(\mu_{X_1},\nu_{X_1})
\qquad\text{and}\qquad
(\mu_{X_2},\nu_{X_2})
\]
denote their corresponding one-dimensional encodings on \([0,1]\) produced by the first layer. A direct comparison of \(\mu_{X_1}\) and \(\mu_{X_2}\) alone would only account for matches between off-diagonal atoms. However, as in the standard \Scorrection{assignment} formulation for persistence diagrams, unmatched off-diagonal points must also be allowed to match to the diagonal. In \Scorrection{$d_{\mathrm{SK}}^{}$}, this effect is encoded directly in one dimension by augmenting each off-diagonal measure with the diagonal projection measure of the other diagram.

More precisely, we introduce the two augmented measures
\[
\xi_{1}:=\mu_{X_1}+\nu_{X_2},
\qquad
\xi_{2}:=\mu_{X_2}+\nu_{X_1}.
\]
Both are finite nonnegative measures on \([0,1]\), and they have the same total mass:
\[
\xi_{1}([0,1])=\xi_{2}([0,1])=|X_1|+|X_2|.
\]
Equivalently, if we define the augmented diagrams
\[
\bar X_1:=X_1+\Pi_{\#}X_2,
\qquad
\bar X_2:=X_2+\Pi_{\#}X_1,
\]
then
\[
\xi_{1}=\iota_{\#}\bar X_1,
\qquad
\xi_{2}=\iota_{\#}\bar X_2.
\]
Thus, at this step, comparing two normalized persistence diagrams reduces to comparing two equal-mass atomic measures on the unit interval.

We then define \Scorrection{\(d_{\mathrm{SK}}\)} as the square root of
the one-dimensional \(1\)-Wasserstein distance between these augmented encodings. More precisely, for two finite nonnegative Borel measures \(\alpha,\beta\) on \([0,1]\) with equal total mass, we write
\[
W_1(\alpha,\beta)
:=
\inf_{\pi\in\Pi(\alpha,\beta)}
\int_{[0,1]^2}|s-t|\,d\pi(s,t),
\]
where \(\Pi(\alpha,\beta)\) denotes the set of couplings of \(\alpha\) and \(\beta\). Here, a coupling \(\pi\) is a joint measure on \([0,1]^2\) whose first and second marginals are \(\alpha\) and \(\beta\), respectively (i.e., \((\mathrm{pr}_1)_\#\pi=\alpha\) and \((\mathrm{pr}_2)_\#\pi=\beta\), where \(\mathrm{pr}_1(s,t)=s\) and \(\mathrm{pr}_2(s,t)=t\)). The quantity \(|s-t|\) is the cost of moving one unit of mass from \(s\) to \(t\), so \(W_1(\alpha,\beta)\) is the minimum total transport cost over all couplings of \(\alpha\) and \(\beta\). \Scorrection{$d_{\mathrm{SK}}^{}$} is then defined by
\[
\mathrm{\Scorrection{d_{\mathrm{SK}}^{}}}(X_1,X_2):=\sqrt{W_1(\xi_{1},\xi_{2})
}=
\sqrt{W_1\bigl(\mu_{X_1}+\nu_{X_2},\,\mu_{X_2}+\nu_{X_1}\bigr)}.
\]

Equivalently, its squared value satisfies
\[
\mathrm{\Scorrection{d_{\mathrm{SK}}^{}}}(X_1,X_2)^2
=
W_1(\xi_1,\xi_2)
=
W_1\bigl(
\mu_{X_1}+\nu_{X_2},
\mu_{X_2}+\nu_{X_1}
\bigr).
\]

\Scorrection{Unlike the cost used in \(W_2\), the ground cost
\(|s-t|\) is applied to every pair in the one-dimensional coupling.
Consequently, pairing two distinct diagonal projections contributes a
positive cost to \(d_{\mathrm{SK}}^{}\), whereas a
diagonal-to-diagonal pair has zero cost in \(W_2\). This is one source
of discrepancy between \(d_{\mathrm{SK}}\) and \(W_2\).}

A key algebraic identity is
\[
(\mu_{X_1}+\nu_{X_2})-(\mu_{X_2}+\nu_{X_1})
=
(\mu_{X_1}-\nu_{X_1})-(\mu_{X_2}-\nu_{X_2})
=
\sigma_{X_1}-\sigma_{X_2}.
\]
In other words, the signed difference between the two augmented measures is exactly the difference of the signed measures introduced in the first layer. In particular, for every threshold \(t\in[0,1]\),
\[
(\xi_1-\xi_2)([0,t])=(\sigma_{X_1}-\sigma_{X_2})([0,t]).
\]

Since one-dimensional \(W_1\) depends only on the cumulative discrepancy
\[
t\longmapsto (\xi_1-\xi_2)([0,t]),
\]
\cite{peyre2019computational}, and this discrepancy is exactly
\[
t\longmapsto (\sigma_{X_1}-\sigma_{X_2})([0,t]),
\]
the value of $\mathrm{\Scorrection{d_{\mathrm{SK}}^{}}}(X_1,X_2)^2$ is entirely determined by the signed difference \(\sigma_{X_1}-\sigma_{X_2}\). This identity is precisely what makes the cumulative \(L^1\) representation of the squared distance \Scorrection{$d_{\mathrm{SK}}^{\,2}$} and, in turn, the
Hilbertian and kernel interpretations of
\Scorrection{\(d_{\mathrm{SK}}\)} developed in the next section.

In summary, the second layer of \Scorrection{$d_{\mathrm{SK}}^{}$} replaces the original two-dimensional partial-assignment problem between normalized persistence diagrams by a one-dimensional \Scorrection{assignment} problem on \([0,1]\), while still retaining the diagonal-matching mechanism through the augmented encodings. This is the central tradeoff behind \Scorrection{$d_{\mathrm{SK}}^{}$}: the comparison is performed in one dimension, which makes the distance computationally much simpler and, as shown in the next subsection, reducible to a sorting-based \(O(N\log N)\) procedure for two normalized persistence diagrams \(D_1\) and \(D_2\), where \(N=|D_1|+|D_2|\), while the \(1/2\)-H\"older regularity of the SK curve ensures that the \Scorrection{$d_{\mathrm{SK}}^{}$} geometry remains tied to the original two-dimensional geometry of normalized persistence diagrams.

\subsection{Sorting-based evaluation and computational complexity of \Scorrection{$d_{\mathrm{SK}}^{}$}}
\label{sec:skot_computation}

Another key advantage of reducing the original two-dimensional geometry to a one-dimensional \Scorrection{assignment} problem is that the transport on the line admits a closed-form monotone solution. If
\[
A:=\{\iota(x):x\in X_1\}\cup\{\iota(\Pi(y)):y\in X_2\},
\qquad
B:=\{\iota(y):y\in X_2\}\cup\{\iota(\Pi(x)):x\in X_1\},
\]
denote the two multisets of scalar values associated with \(\xi_{1}\) and \(\xi_{2}\), and if
\[
a_{(1)}\le\cdots\le a_{(N)},
\qquad
b_{(1)}\le\cdots\le b_{(N)},
\qquad
N:=|X_1|+|X_2|,
\]
are their sorted versions, then the optimal \Scorrection{assignment} is obtained by matching the \(k\)-th smallest value on the left with the \(k\)-th smallest value on the right. As a result (see \autoref{app:prop:1D_EMD}), 
\[
\mathrm{\Scorrection{d_{\mathrm{SK}}^{}}}(X_1,X_2)
=
\left(
\sum_{k=1}^{N}
|a_{(k)}-b_{(k)}|
\right)^{1/2}.
\]

\Scorrection{Keeping the identity of the diagram point associated with each scalar
value, the same sorted pairing induces an explicit assignment between
the augmented diagrams \(\bar X_1\) and \(\bar X_2\). Equivalently, it
defines a partial assignment between \(X_1\) and \(X_2\), with
unmatched points assigned to the diagonal.}

For the finite-level numerical evaluation, the exact selector
\(\iota\) is replaced by its finite-level selector \(\iota_L\).
Computing \(\iota_L\) for a diagram point or a diagonal projection
requires traversing \(L\) successive levels of the SK refinement, with
constant work at each level. Since the two diagrams contain \(N\)
points in total and generate \(N\) diagonal projections, constructing
the multisets \(A\) and \(B\) costs \(O(NL)\). 

Sorting the two lists \(A\) and \(B\) costs
\(O(N\log N)\), while evaluating the
resulting monotone assignment and taking the final square root costs
\(O(N)\). Consequently, the complete evaluation of
\(d_{\mathrm{SK},L}\), including the computation of all selector
values, has complexity
\[
O\!\left(NL+N\log N\right)
=
O\!\left(N(L+\log N)\right).
\]
For a fixed refinement level \(L\), as used in our implementation,
this reduces to \(O(N\log N)\).

This sorting-based formula is the main algorithmic advantage of \Scorrection{$d_{\mathrm{SK}}^{}$}. After the geometric information of the diagrams has been encoded on the unit interval, the full comparison reduces to a one-dimensional monotone \Scorrection{assignment} problem with a closed-form solution. In this sense, \Scorrection{$d_{\mathrm{SK}}^{}$} preserves the diagonal-aware structure of persistence-diagram matching while replacing the original two-dimensional optimal \Scorrection{assignment} computation by a simple and efficient procedure on sorted scalar codes.

\subsection{Controlling the 2-Wasserstein distance with \texorpdfstring{$\mathrm{\Scorrection{d_{\mathrm{SK}}^{}}}$}{\Scorrection{$d_{\mathrm{SK}}^{}$}}}
\label{sec:skot_w2}

We explain in this subsection how the \(1/2\)-H\"older regularity of the SK space-filling curve ties the one-dimensional \Scorrection{assignment} geometry of \Scorrection{$d_{\mathrm{SK}}^{}$} back to the original two-dimensional Wasserstein geometry of normalized persistence diagrams. More precisely, it allows us to control the quadratic 2-Wasserstein distance by \(\mathrm{\Scorrection{d_{\mathrm{SK}}^{}}}\), thereby justifying \(d_{\mathrm{SK}}\) as a fast surrogate metric for repeated \(W_{2}\)-type evaluations.

Recall that, for two normalized persistence diagrams \(X_1\) and \(X_2\), we introduced the augmented diagrams
\[
\bar X_1:=X_1+\Pi_{\#}X_2,
\qquad
\bar X_2:=X_2+\Pi_{\#}X_1,
\]
which have the same total mass. We also consider the diagonal-aware squared cost
\[
c_\Delta(x,y):=
\begin{cases}
\|x-y\|_2^2, & x\notin\Delta,\ y\notin\Delta,\\[2pt]
d_\Delta(x)^2, & x\notin\Delta,\ y\in\Delta,\\[2pt]
d_\Delta(y)^2, & x\in\Delta,\ y\notin\Delta,\\[2pt]
0, & x\in\Delta,\ y\in\Delta,
\end{cases}
\]
where
\[
d_\Delta(z):=\inf_{u\in[0,1]}\|z-(u,u)\|_2
\qquad (z\in\Atri)
\]
denotes the Euclidean distance to the diagonal. The associated quadratic Wasserstein distance between normalized persistence diagrams is then defined by
\[
W_{2}^2(X_1,X_2)
:=
\inf_{\gamma\in\Pi(\bar X_1,\bar X_2)}
\int_{\Atri\times\Atri} c_\Delta(x,y)\,d\gamma(x,y).
\]

On the other hand, by construction,
\[
\mathrm{\Scorrection{d_{\mathrm{SK}}^{}}}(X_1,X_2)^{2}=W_1(\iota_{\#}\bar X_1,\iota_{\#}\bar X_2),
\]
where the ground cost on \([0,1]\) is \(|t-s|\). Let \(\pi^\star\) be the optimal monotone coupling between \(\iota_{\#}\bar X_1\) and \(\iota_{\#}\bar X_2\), so that
\[
\mathrm{\Scorrection{d_{\mathrm{SK}}^{}}}(X_1,X_2)^2
=
\int_{[0,1]^2}|t-s|\,d\pi^\star(t,s).
\]
Pushing this coupling forward through the SK curve yields
\[
\Gamma_{12}:=(S,S)_{\#}\pi^\star.
\]
Since \(S\circ\iota=\mathrm{id}_{\Atri}\), the marginals of \(\Gamma_{12}\) are exactly \(\bar X_1\) and \(\bar X_2\), so \(\Gamma_{12}\) is an admissible coupling for \(W_{2}(X_1,X_2)\). \Scorrection{In
the finite-diagram setting considered here, \(\Gamma_{12}\) encodes precisely the
diagonal-aware assignment between \( X_1\) and \( X_2\) obtained by monotonically pairing the sorted scalar lists (see~\autoref{sec:skot_computation}).}

Now let \(x:=S(t)\) and \(y:=S(s)\). If both points are off-diagonal, then \(c_\Delta(x,y)=\|x-y\|_2^2\). If one of them lies on the diagonal, say \(y\in\Delta\), then
\[
c_\Delta(x,y)=d_\Delta(x)^2\le \|x-y\|_2^2,
\]
because \(y\) is itself a diagonal point. If both points lie on the diagonal, then \(c_\Delta(x,y)=0\). Therefore, in all cases,
\[
c_\Delta(S(t),S(s))\le \|S(t)-S(s)\|_2^2.
\]
By the \(1/2\)-H\"older bound established in Section~\ref{sec:skot_encoding},
\[
\|S(t)-S(s)\|_2^2\le C_{\mathrm{SK}}\,|t-s|.
\]
Combining the two inequalities and integrating against \(\pi^\star\), we obtain
\[
\int_{\Atri\times\Atri} c_\Delta(x,y)\,d\Gamma_{12}(x,y)
=
\int_{[0,1]^2} c_\Delta(S(t),S(s))\,d\pi^\star(t,s)
\le
C_{\mathrm{SK}}\int_{[0,1]^2}|t-s|\,d\pi^\star(t,s)
=
C_{\mathrm{SK}}\,\mathrm{\Scorrection{d_{\mathrm{SK}}^{}}}(X_1,X_2)^{2}.
\]
Since \(\Gamma_{12}\) is only one admissible coupling among all couplings defining \(W_{2,}(X_1,X_2)\), taking the infimum yields
\[
W_{2}^2(X_1,X_2)\le C_{\mathrm{SK}}\,
\mathrm{\Scorrection{d_{\mathrm{SK}}^{}}}(X_1,X_2)^2.
\]
Equivalently,
\[
W_{2}(X_1,X_2)\le \sqrt{C_{\mathrm{SK}}}\,\mathrm{\Scorrection{d_{\mathrm{SK}}^{}}}(X_1,X_2).
\]

For the normalized SK curve used throughout this section, we have \(C_{\mathrm{SK}}=2\), so in particular
\[
W_{2}(X_1,X_2)\le \sqrt{2}\,\mathrm{\Scorrection{d_{\mathrm{SK}}^{}}}(X_1,X_2).
\]

A formal statement and proof of this bound are given in
\autoref{app:thm:W2_le_SKOT}.

Thus, although \Scorrection{$d_{\mathrm{SK}}^{}$} is defined through
a one-dimensional assignment problem on \([0,1]\), the
\(1/2\)-Hölder regularity of the SK curve ensures that it controls the
two-dimensional \(2\)-Wasserstein distance between persistence
diagrams. This is precisely why \(d_{\mathrm{SK}}\) can be used as a computationally efficient surrogate for the diagonal-aware quadratic Wasserstein distance between normalized persistence diagrams.

\begin{remark}[A tighter but non-metric surrogate]
Recall that the proof above constructs a specific admissible coupling
\[
\Gamma_{12}:=(S,S)_{\#}\pi^\star \in \Pi(\bar X_1,\bar X_2),
\]
where \(\pi^\star\) is the monotone optimal coupling on \([0,1]\). Define its associated quadratic cost by
\[
W_{\Gamma}^2(X_1,X_2)
:=
\int_{\Atri\times\Atri} c_\Delta(x,y)\,d\Gamma_{12}(x,y),
\]
or equivalently
\[
W_{\Gamma}(X_1,X_2)
:=
\left(
\int_{\Atri\times\Atri} c_\Delta(x,y)\,d\Gamma_{12}(x,y)
\right)^{1/2}.
\]

\Scorrection{Thus, \(d_{\mathrm{SK}}^{}\) and \(W_\Gamma^{}\) evaluate the
correspondence induced by the same one-dimensional coupling with
different costs: the former uses \(|s-t|\) on \([0,1]\), whereas the
latter uses \(c_\Delta\) in \(\Atri\), so diagonal-to-diagonal pairs
have zero cost, as in \(W_2\).}

By construction, one has
\[
W_{2}^2(X_1,X_2)
\le
W_{\Gamma}^2(X_1,X_2)
\le
C_{\mathrm{SK}}\,\mathrm{\Scorrection{d_{\mathrm{SK}}^{}}}(X_1,X_2)^{2},
\]
or equivalently,
\[
W_{2}(X_1,X_2)
\le
W_{\Gamma}(X_1,X_2)
\le
\sqrt{C_{\mathrm{SK}}}\,\mathrm{\Scorrection{d_{\mathrm{SK}}^{}}}(X_1,X_2).
\]
Thus \(W_{\Gamma}\) can be viewed as a tighter upper surrogate for \(W_{2}\) than the coarse bound \(\sqrt{C_{\mathrm{SK}}}\,\mathrm{\Scorrection{d_{\mathrm{SK}}^{}}}\), since it evaluates the \Scorrection{assignment} cost directly in the persistence triangle along the specific correspondence induced by \(\pi^\star\).

\Scorrection{However, unlike \(d_{\mathrm{SK}}\), \(W_\Gamma\) is not a metric,
as shown by the following singleton counterexample. Let
\[
X=\{x\},
\qquad
Y=\{y\},
\qquad
Z=\{z\},
\]
with
\[
x=(0.40,0.85),
\qquad
y=(0.50,0.80),
\qquad
z=(0.55,0.75).
\]
For the SK traversal fixed above, the level-\(8\) selector values satisfy
\[
\iota_8(x)
<
\iota_8(y)
<
\iota_8(\Pi(x))
<
\iota_8(z)
<
\iota_8(\Pi(y))
=
\iota_8(\Pi(z)).
\]
Since each exact first-hit value lies in the level-\(8\) dyadic
interval associated with its finite-level value, and these intervals
are pairwise disjoint and strictly ordered for the five distinct
locations, this finite-level ordering certifies the exact ordering:
\[
\iota(x)
<
\iota(y)
<
\iota(\Pi(x))
<
\iota(z)
<
\iota(\Pi(y))
=
\iota(\Pi(z)).
\]
The induced monotone assignments therefore match \(x\) with \(y\)
and \(y\) with \(z\), whereas the comparison between \(X\) and \(Z\)
assigns both off-diagonal points to the diagonal. Consequently,
\[
\begin{aligned}
W_\Gamma(X,Y)
&=
\|x-y\|_2
\approx
0.112,\\
W_\Gamma(Y,Z)
&=
\|y-z\|_2
\approx
0.071,\\
W_\Gamma(X,Z)
&=
\sqrt{d_\Delta(x)^2+d_\Delta(z)^2}
\approx
0.348.
\end{aligned}
\]
Hence,
\[
W_\Gamma(X,Z)
>
W_\Gamma(X,Y)+W_\Gamma(Y,Z),
\]
since \(0.348>0.112+0.071\). The failure arises because the
SK-induced assignments selected for the three pairs are not mutually
compatible. The failure of the triangle inequality for \(W_\Gamma\) is illustrated
in \autoref{fig:wgamma_triangle_counterexample}.}

\begin{figure}[t]
  \centering
  \includegraphics[width=\linewidth]
  {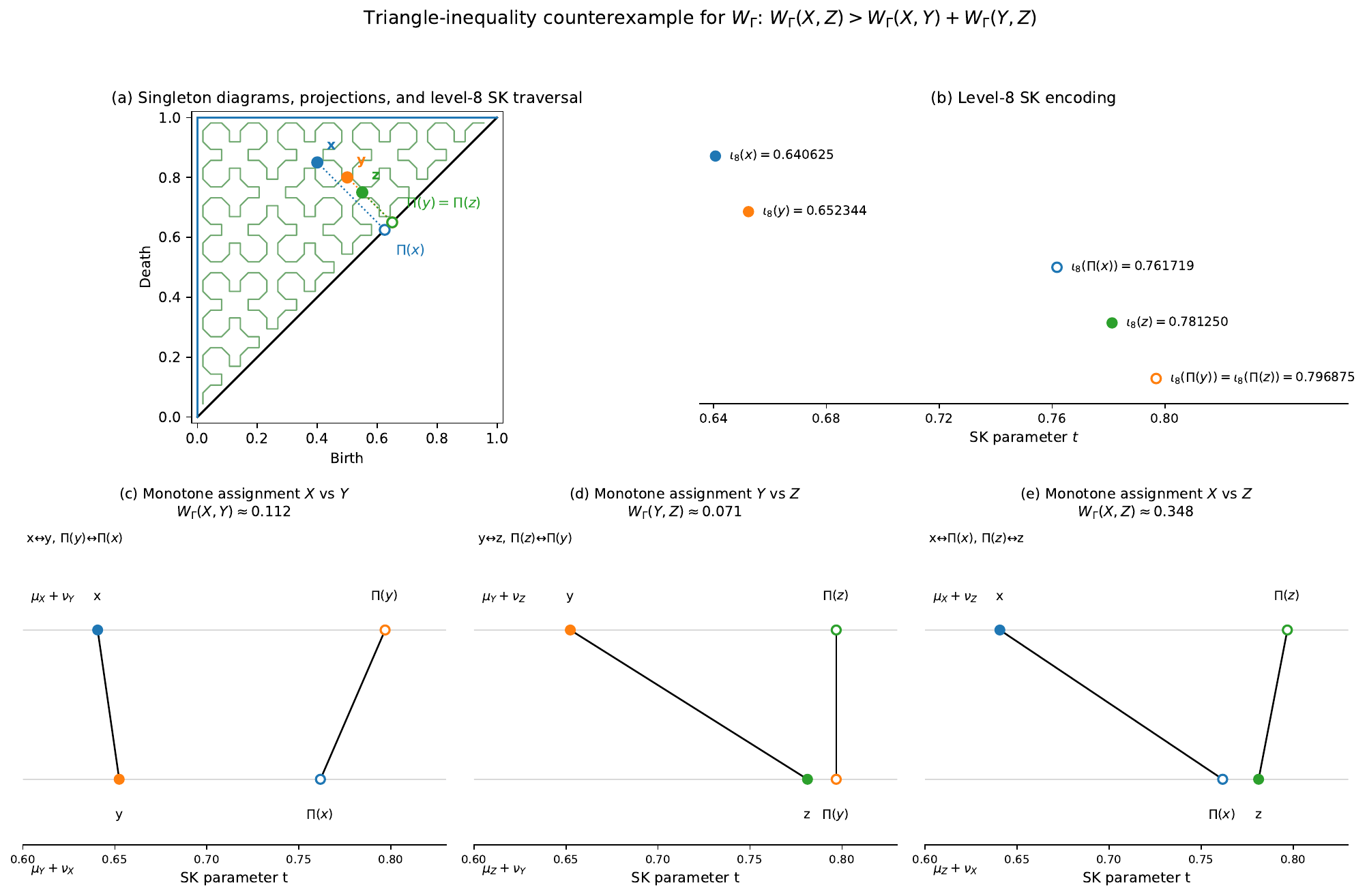}
  \caption{
  Triangle-inequality counterexample for \(W_\Gamma\).
  (a) Singleton diagrams \(X=\{x\}\), \(Y=\{y\}\), and \(Z=\{z\}\),
  their diagonal projections, and the level-\(8\) SK traversal on
  \(\Atri\).
  (b) The level-\(8\) selector values certify the exact first-hit
  ordering
  \(
  \iota(x)
  <
  \iota(y)
  <
  \iota(\Pi(x))
  <
  \iota(z)
  <
  \iota(\Pi(y))
  =
  \iota(\Pi(z)).
  \)
  (c)--(e) Monotone assignments induced for the three pairwise
  comparisons. They yield
  \(W_\Gamma(X,Y)\approx0.112\),
  \(W_\Gamma(Y,Z)\approx0.071\), and
  \(W_\Gamma(X,Z)\approx0.348\). Consequently,
  \(
  W_\Gamma(X,Z)
  >
  W_\Gamma(X,Y)+W_\Gamma(Y,Z),
  \)
  violating the triangle inequality.
  }
  \label{fig:wgamma_triangle_counterexample}
\end{figure}

\Scorrection{This non-metric behavior is also confirmed numerically for
\(W_{\Gamma,30}\) in our experiments.}

Nevertheless, because it may provide a sharper approximation of \(W_{2}\) in practice, we also evaluate its empirical accuracy, alongside \(\mathrm{\Scorrection{d_{\mathrm{SK}}^{}}}\), in the Experimental evaluation section.
\end{remark}

\smallskip
\noindent
\textbf{Dataset-specific converse bounds.}
The forward control \(W_{2}(X_1,X_2)\lesssim \mathrm{\Scorrection{d_{\mathrm{SK}}^{}}}(X_1,X_2)\) is universal, but the converse direction is necessarily more data-dependent, because the selector \(\iota\) is only Borel measurable and need not be globally Lipschitz on \(\Atri\). Nevertheless, on any fixed finite family \(\mathcal F\) of normalized persistence diagrams, one can recover a quantitative reverse comparison. More precisely, if
\[
\mathcal Z_{\mathcal F}
:=
\bigcup_{X\in\mathcal F}
\Bigl(\operatorname{spt}(X)\cup \Pi(\operatorname{spt}(X))\Bigr)
\subset \Atri
\]
denotes the finite set of relevant atoms, then the empirical Lipschitz constant
\[
L_{\mathcal F}
:=
\max_{\substack{x,y\in\mathcal Z_{\mathcal F}\\ x\neq y}}
\frac{|\iota(x)-\iota(y)|}{\|x-y\|_2}
\]
is finite, and the following bound holds for every
\(X_1,X_2\in\mathcal F\) (see \autoref{app:prop:SKOT_reverse_finite_dataset}):
\[
\mathrm{\Scorrection{d_{\mathrm{SK}}^{}}}(X_1,X_2)^{2}
\le
L_{\mathcal F}\,\sqrt{2\bigl(|X_1|+|X_2|\bigr)}\;W_{2}(X_1,X_2).
\]
In particular, if all diagrams in \(\mathcal F\) satisfy \(|X|\le M\), this simplifies to
\[
\mathrm{\Scorrection{d_{\mathrm{SK}}^{}}}(X_1,X_2)^2\le {2L_{\mathcal F}\sqrt{M}\;}{W_{2}(X_1,X_2)}.
\]
A coarser but easier-to-evaluate version is obtained by replacing \(L_{\mathcal F}\) with a separation-based constant: if \(\delta_{\mathcal F}\) denotes the minimum pairwise Euclidean distance between distinct off-diagonal atoms in the dataset and \(\eta_{\mathcal F}\) the minimum Euclidean distance to the diagonal, then the following bound holds (see \autoref{app:cor:SKOT_reverse_dataset_simple}):
\[
\mathrm{\Scorrection{d_{\mathrm{SK}}^{}}}(X_1,X_2)^2
\le
{\frac{2\sqrt M}{\min(\delta_{\mathcal F},\eta_{\mathcal F})}\;} {W_{2}(X_1,X_2)}
\qquad\forall\,X_1,X_2\in\mathcal F.
\]
Then, despite no universal reverse inequality is expected in full generality, \Scorrection{$d_{\mathrm{SK}}^{}$} and \(W_2\) remain quantitatively comparable on any fixed finite collection of normalized persistence diagrams.

%% file: continuous_edit_distance_geodesics_between_time_varying_persistence_diagrams_Section4_.tex
\section{Hilbertian Geometry and Gaussian Kernel Induced by
\texorpdfstring{$\mathrm{\Scorrection{d_{\mathrm{SK}}^{}}}$}{dSK}}
\label{sec:skot_hilbert_kernel}

\noindent

This section establishes two structural properties of
\Scorrection{\(d_{\mathrm{SK}}\)} on the space of normalized
persistence diagrams. We first show that the squared distance
\(\Scorrection{d_{\mathrm{SK}}}^2\) admits an explicit
cumulative \(L^1\) representation on \([0,1]\). This representation
yields an explicit Hilbert-space embedding whose Hilbert distance is
exactly \Scorrection{\(d_{\mathrm{SK}}\)}. It also shows that the
squared distance is conditionally negative definite and therefore
induces a positive definite Gaussian kernel on normalized persistence
diagrams.

\subsection{Cumulative \texorpdfstring{\(L^1\)}{L1}
representation of the squared distance}
\label{sec:skot_l1_signature}

Recall that, for a normalized persistence diagram \(X\), the first
layer of \Scorrection{\(d_{\mathrm{SK}}\)} yields the signed measure
\[
\sigma_X
:=
\mu_X-\nu_X
\qquad
\text{on }[0,1].
\]
We define its cumulative signature by
\[
H_X(t)
:=
\sigma_X([0,t]),
\qquad
t\in[0,1].
\]

By the algebraic identity established in
\autoref{sec:skot_layer2},
\[
(\mu_{X_1}+\nu_{X_2})
-
(\mu_{X_2}+\nu_{X_1})
=
\sigma_{X_1}-\sigma_{X_2}.
\]
Combining this identity with the cumulative formula for
one-dimensional \(W_1\) yields
\[
\mathrm{\Scorrection{d_{\mathrm{SK}}^{}}}(X_1,X_2)^2
=
\int_0^1
\left|
(\sigma_{X_1}-\sigma_{X_2})([0,t])
\right|
\,dt
=
\int_0^1
\left|
H_{X_1}(t)-H_{X_2}(t)
\right|
\,dt.
\]
Hence, the squared distance is exactly the \(L^1\) distance between
the cumulative signatures:
\[
\mathrm{\Scorrection{d_{\mathrm{SK}}^{}}}(X_1,X_2)^2
=
\left\|
H_{X_1}-H_{X_2}
\right\|_{L^1([0,1])}.
\]

Equivalently, the metric itself is given by
\[
\mathrm{\Scorrection{d_{\mathrm{SK}}^{}}}(X_1,X_2)
=
\left\|
H_{X_1}-H_{X_2}
\right\|_{L^1([0,1])}^{1/2}.
\]
Thus, the complete
\Scorrection{\(d_{\mathrm{SK}}\)} geometry is encoded by a
one-dimensional cumulative signal.

This representation is also injective at the diagram level. Indeed,
the cumulative function \(H_X\) determines the finite signed measure
\(\sigma_X\). Moreover, since \(\mu_X\) and \(\nu_X\) are supported
on the disjoint sets
\(
\iota(\Atri\setminus\Delta)
\text{ and }
\iota(\Delta),
\)
respectively, they are the positive and negative parts of
\(\sigma_X\). Therefore, \(\sigma_X\) determines \(\mu_X\) and
\(\nu_X\); finally, the injectivity of the exact selector \(\iota\)
implies that \(\mu_X\) determines the original normalized persistence
diagram \(X\).

The cumulative \(L^1\) representation and the injectivity of this
encoding are formally proved in \autoref{app:prop:SKOT_L1}.

\subsection{Hilbertian geometry of
\texorpdfstring{$\mathrm{\Scorrection{d_{\mathrm{SK}}^{}}}$}{dSK}}
\label{sec:skot_hilbert}

The \(L^1\) signature representation immediately implies that \(\mathrm{\Scorrection{d_{\mathrm{SK}}^{}}}\) is not merely a metric, but a Hilbertian one. Indeed, define for each normalized persistence diagram \(X\) the function
\[
\Phi(X):[0,1]\times\mathbb{R}\to\mathbb{R},
\qquad
\Phi(X)(t,u):=\mathbf{1}_{\{u<H_X(t)\}}-\mathbf{1}_{\{u<0\}}.
\]

For each fixed \(t\), the function
\[
u\longmapsto\Phi(X)(t,u)
\]
is supported on the interval between \(0\) and \(H_X(t)\). Since
\(H_X\) is the cumulative function of a finite signed atomic measure,
it is bounded and integrable on \([0,1]\). Consequently,
\[
\Phi(X)
\in
L^2\bigl([0,1]\times\mathbb{R},dt\,du\bigr).
\]

Moreover, for any two reals \(a,b\),
\[
\int_{\mathbb{R}}
\Bigl(
\mathbf{1}_{\{u<a\}}-\mathbf{1}_{\{u<b\}}
\Bigr)^2\,du
=
|a-b|.
\]
Applying this identity pointwise with \(a=H_{X_1}(t)\) and \(b=H_{X_2}(t)\), and then integrating in \(t\), gives
\[
\|\Phi(X_1)-\Phi(X_2)\|_{L^2([0,1]\times\mathbb{R})}^2
=
\int_0^1 |H_{X_1}(t)-H_{X_2}(t)|\,dt
=
\mathrm{\Scorrection{d_{\mathrm{SK}}^{\,2}}}(X_1,X_2).
\]
Therefore,
\[
d_{\mathrm{SK}}(X_1,X_2)=
\|\Phi(X_1)-\Phi(X_2)\|_{L^2([0,1]\times\mathbb{R})}.
\]

A formal proof of this explicit isometric Hilbert embedding is
provided in~\autoref{app:prop:SKOT_Hilbert_embed}.

In other words, \Scorrection{\(d_{\mathrm{SK}}\)} is the pullback of an ordinary Hilbert space distance through the explicit embedding
\(\Phi\). In particular, this makes
\Scorrection{\(d_{\mathrm{SK}}\)} directly compatible with
Euclidean embedding methods and learning pipelines that require vector
representations,  while preserving the diagram-level structure encoded by \Scorrection{$d_{\mathrm{SK}}^{}$}.

\subsection{Gaussian kernel induced by
\texorpdfstring{$\mathrm{\Scorrection{d_{\mathrm{SK}}^{}}}$}{dSK}}
\label{sec:skot_kernel}

The Hilbertian representation also provides a natural Gaussian kernel
for normalized persistence diagrams. Since
\[
\mathrm{\Scorrection{d_{\mathrm{SK}}^{}}}(X_1,X_2)^2
=
\left\|
\Phi(X_1)-\Phi(X_2)
\right\|_
{L^2([0,1]\times\mathbb{R})}^2,
\]
the squared distance
\[
(X_1,X_2)
\longmapsto
\mathrm{\Scorrection{d_{\mathrm{SK}}^{}}}(X_1,X_2)^2
\]
is conditionally negative definite (see~\autoref{app:prop:SKOT_CND}). By Schoenberg's theorem, for
every \(\sigma>0\), the function
\[
k_\sigma(X_1,X_2)
:=
\exp\left(
-\frac{
\mathrm{\Scorrection{d_{\mathrm{SK}}^{}}}(X_1,X_2)^2
}{
2\sigma^2
}
\right)
\]
is a positive definite kernel on normalized persistence diagrams (see~\autoref{app:cor:SKOT_kernel_PD}).

Using the explicit embedding \(\Phi\), this kernel can equivalently
be written as
\[
k_\sigma(X_1,X_2)
=
\exp\left(
-\frac{
\left\|
\Phi(X_1)-\Phi(X_2)
\right\|_
{L^2([0,1]\times\mathbb{R})}^2
}{
2\sigma^2
}
\right),
\]
that is, as the restriction of the standard Gaussian kernel on the Hilbert space \(L^2([0,1]\times\mathbb{R})\)
to the embedded space of normalized persistence diagrams. This kernel is a natural companion to the metric
\Scorrection{\(d_{\mathrm{SK}}\)}: both encode the same Hilbertian
geometry, while the kernel provides a similarity function directly
usable in kernel-based methods. In particular, it gives a principled
way of applying support vector machines, kernel ridge regression,
Gaussian processes, spectral or kernel clustering, and
MMD-based distributional comparisons to normalized persistence
diagrams. Under standard results for Gaussian kernels on Hilbert spaces, the
kernel is also characteristic and, on compact subsets of the
embedded diagram space, universal
\cite{Guella2020GaussianHilbert,Sriperumbudur2010HilbertEmbedding};
see~\autoref{app:thm:SKOT_kernel_characteristic} and~\autoref{app:cor:SKOT_kernel_universal}, respectively.

Finally, from a computational viewpoint, evaluating \(k_\sigma(X_1,X_2)\) only requires one computation of \(\mathrm{\Scorrection{d_{\mathrm{SK}}^{\,2}}}(X_1,X_2)\) followed by one exponential. Therefore, given the relevant values of the selector \(\iota\), the kernel can be evaluated in the same \(O(N\log N)\) time as \Scorrection{$d_{\mathrm{SK}}^{}$} itself, where \(N=|X_1|+|X_2|\).

%% file: applications_Section6_.tex
\section{Experimental evaluation}\label{sec:applications}

\subsection{Scientific benchmark and experimental setup}
\label{sec:scientific_ensembles}

We evaluate \Scorrection{the two SK-based constructions, the SK-Wasserstein
distance \(d_{\mathrm{SK}}\) and the SK-induced planar surrogate
\(W_\Gamma\),} on 12 scientific ensemble collections derived from
scalar fields previously used in the evaluation of the Wasserstein
distance between merge trees~\cite{pont2021wasserstein}. The benchmark contains 227
persistence diagrams and \(3{,}376{,}524\) persistence pairs in total.
Collection sizes range from 7 to 48 diagrams, while individual
diagrams contain between 22 and \(102{,}735\) pairs. The collections
and their main characteristics are summarized in
\autoref{tab:dataset_summary}.

Each collection is accompanied by metadata assigning its \(n\)
scalar fields to \(k\) reference groups. These labels define the reference
partition supplied with the benchmark.

For each scalar field, we compute a classical persistence diagram
with the Discrete Morse Sandwich backend of TTK~\cite{ttk17, ttk19,
guillou2023discretemorsesandwichfast, leguillou_tvcg24}. We
retain all available critical-pair types and homological dimensions,
without persistence thresholding, cardinality truncation, or
subsampling. We use the collection-wise normalization introduced in
\autoref{sec:input_data}.

For every collection, we compute a complete pairwise matrix of the classical $2$-Wasserstein distance using TTK's auction-based solver \cite{kerber2017geometry, vidal2019progressive},
with relative-precision threshold \(\texttt{DeltaLim}=0.01\), without
persistence thresholding, cardinality truncation, or subsampling. Throughout the remainder of this section, \(W_2\) denotes this
numerical TTK approximation.

All distance constructions below are evaluated on exactly
the same normalized diagrams.

\medskip
\noindent
\textbf{SK-based matrices and refinement levels.}
We next specify the pairwise matrices and refinement levels used
throughout the evaluation.

For each collection \(c\), containing \(n_c\) persistence diagrams,
and each refinement level
\[
L\in\{10,20,30,40\},
\]
we compute the pairwise matrices associated with
\(d_{\mathrm{SK},L}\) and \(W_{\Gamma,L}\). Here,
\(d_{\mathrm{SK},L}\) denotes the level-\(L\) numerical approximation
of the exact SK-Wasserstein distance \(d_{\mathrm{SK}}\): it is
obtained by replacing the exact first-hit selector \(\iota\) in the
definition of \(d_{\mathrm{SK}}\) with the finite-level
selector \(\iota_L\) introduced in
\autoref{sec:skot_encoding}. Likewise, \(W_{\Gamma,L}\) denotes the
SK-induced planar surrogate of \autoref{sec:skot_w2}, computed from
the monotone assignment induced by the same finite-level selector \(\iota_L\).

\input{tables/table_dataset_summary.tex}

\subsection{Convergence with respect to the selector refinement level}
\label{sec:selector_convergence_results}

\subsubsection{Protocol}

We first evaluate the influence of the selector refinement level
\(L\) on the two SK-based constructions. For each construction
\[
m\in\{d_{\mathrm{SK}},W_\Gamma\},
\]
let
\[
D_{c,L}^{m}\in\mathbb{R}^{n_c\times n_c}
\]
denote its symmetric pairwise matrix for collection \(c\). We use
\(L=40\) as a high-resolution numerical reference and compare the
matrices obtained at
\[
L\in\{10,20,30\}
\]
with their corresponding \(L=40\) matrix. The \(L=40\) computation is
used only as a finite numerical reference and is not assumed to
coincide with the exact \(L\to\infty\) limit.

For \(d_{\mathrm{SK}}\), this empirical refinement study is
complemented by the perturbation bound in
\autoref{app:prop:SKOT_selector_approx}: a selector approximation
error of at most \(\varepsilon\) on all persistence pairs of \(X\) and \(Y\) and
on their diagonal projections induces an error of at most
\(2(|X|+|Y|)\varepsilon\) on
\(d_{\mathrm{SK}}(X,Y)^2\).

\begin{itemize}

\item We first measure the numerical difference between the matrices
using the relative Frobenius error
\[
E_c^{m}(L)
:=
\frac{
  \left\|D_{c,L}^{m}-D_{c,40}^{m}\right\|_F
}{
  \left\|D_{c,40}^{m}\right\|_F
}.
\]

\item To assess whether changing \(L\) modifies the global ordering of
the pairwise dissimilarities, let
\[
\operatorname{ut}(D)
:=
\bigl(D_{ij}\bigr)_{1\leq i<j\leq n_c}
\]
denote the vector of strict upper-triangular entries of a symmetric
matrix \(D\). This vector contains each distinct pairwise
dissimilarity exactly once. We compute
\[
\rho_c^{m}(L)
:=
\rho_{\mathrm{S}}
\left(
  \operatorname{ut}\!\left(D_{c,L}^{m}\right),
  \operatorname{ut}\!\left(D_{c,40}^{m}\right)
\right),
\]
where \(\rho_{\mathrm{S}}\) denotes Spearman's rank correlation
coefficient. Thus, \(\rho_c^{m}(L)=1\) means that all distinct
pairwise dissimilarities are ranked in the same order at refinement
levels \(L\) and \(40\), although their numerical values may differ.

\item We additionally evaluate local-neighbourhood preservation. For
a pairwise matrix \(D\), let
\[
\mathcal{N}_3(D,i)
\]
denote the set of the three nearest neighbours of sample \(i\),
excluding \(i\) itself. We define
\[
\operatorname{NN@3}_c^{m}(L)
:=
\frac{1}{3n_c}
\sum_{i=1}^{n_c}
\left|
  \mathcal{N}_3\!\left(D_{c,L}^{m},i\right)
  \cap
  \mathcal{N}_3\!\left(D_{c,40}^{m},i\right)
\right|.
\]
Hence, \(\operatorname{NN@3}_c^{m}(L)=1\) means that every diagram in
collection \(c\) has exactly the same three nearest neighbours at
levels \(L\) and \(40\).

\end{itemize}

Across the 12 collections, we report the median and maximum values of
\(E_c^{m}(L)\), together with the minimum values of
\(\rho_c^{m}(L)\) and \(\operatorname{NN@3}_c^{m}(L)\). The latter
two quantities therefore characterize the worst-performing
collection at each refinement level. These convergence diagnostics are used to
select the refinement level employed in the subsequent evaluation.

\subsubsection{Results}

\autoref{fig:selector_convergence} and
\autoref{tab:selector_convergence} summarize convergence toward the
high-resolution \(L=40\) numerical reference using the diagnostics
defined above.

At \(L=30\), the median relative Frobenius error over the 12
collections is \(6.97\times10^{-8}\) for
\(d_{\mathrm{SK},30}\) and \(4.50\times10^{-6}\) for
\(W_{\Gamma,30}\). \Scorrection{These median errors, of order \(10^{-6}\) or smaller, are close to
the numerical error of single-precision floating-point computations
and can therefore be regarded as virtually zero.} The corresponding worst-case errors are
\(1.41\times10^{-5}\) and \(6.63\times10^{-5}\), respectively.

For every collection and both constructions, the Spearman
correlation between the strict upper-triangular entries of the
\(L=30\) and \(L=40\) matrices rounds to \(1.000000\), while
\(\operatorname{NN@3}=1.000\). Hence, at the reported numerical
precision, no change is detected in the global ranking of the
distinct pairwise dissimilarities, and every sample retains exactly
the same three nearest neighbours as at \(L=40\).

By contrast, at \(L=10\), the worst-case relative errors reach
\(0.811\) for \(d_{\mathrm{SK},10}\) and \(1.11\) for
\(W_{\Gamma,10}\), while the worst-collection Spearman correlations
decrease to \(0.605\) and \(0.492\), respectively. We therefore fix
\[
L=30
\]
for all subsequent experiments in this section.

\begin{figure}[t]
  \centering
  \includegraphics[width=0.83\linewidth]
  {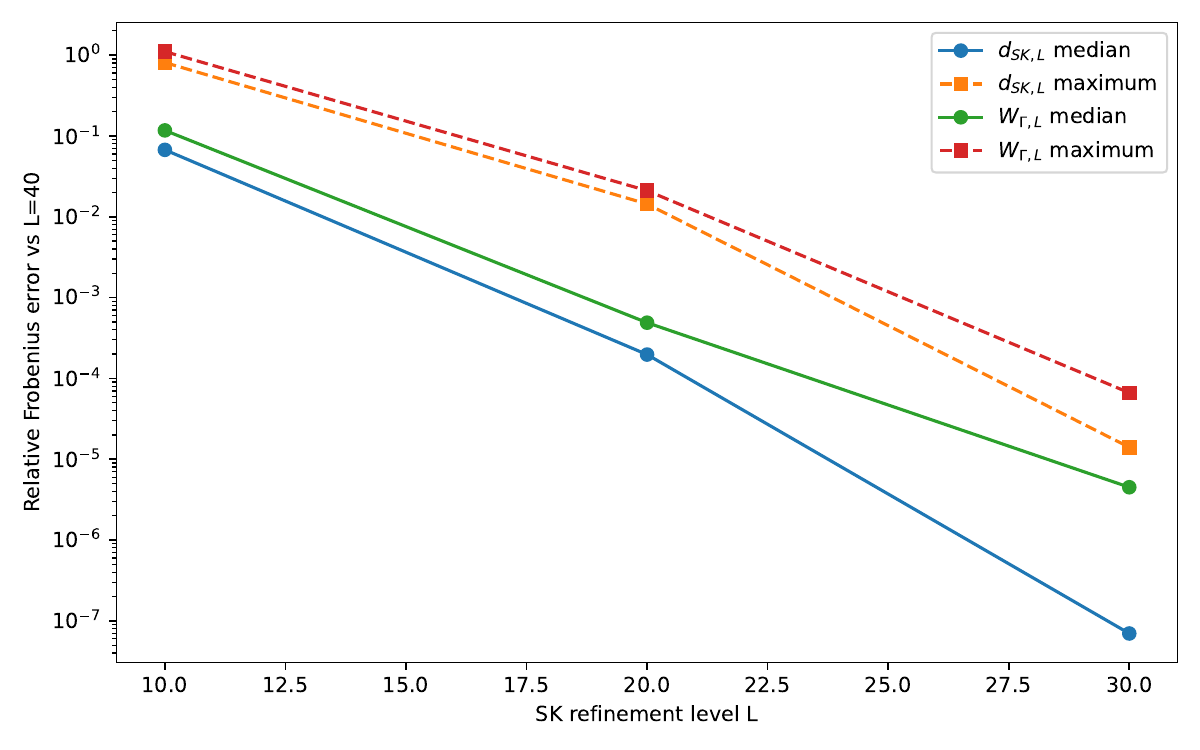}
  \caption{
  Convergence of \(d_{\mathrm{SK},L}\) and \(W_{\Gamma,L}\) toward
  their high-resolution \(L=40\) numerical reference matrices. For
  each method and refinement level \(L\), the curves report the
  median and maximum, over the 12 collections, of the relative
  Frobenius error \(E_c(L)\) defined in
  \autoref{sec:selector_convergence_results}. Lower values indicate closer
  agreement with the \(L=40\) reference; the reference level itself
  has zero error by construction.
  }
  \label{fig:selector_convergence}
\end{figure}

\input{tables/table_convergence.tex}

\subsection{Faithfulness to \(W_2\) and computational performance}
\label{sec:w2_faithfulness_results}

\subsubsection{Protocol}

Following the selector-refinement study in
\autoref{sec:selector_convergence_results}, we fix
\[
L=30
\]

and compare the two finite-level SK-based constructions
\(
d_{\mathrm{SK},30}
\text{ and }
W_{\Gamma,30}
\) in terms of both their faithfulness to the numerical \(W_2\)
reference defined in \autoref{sec:scientific_ensembles} and their
computational performance relative to the corresponding TTK
\(W_2\) computation. Across the 12 collections, this
yields \(3{,}004\) distinct pairwise diagram comparaisons for each dissimilarity.

Let
\[
D_c^{W_2}
\]
denote the pairwise \(W_2\) matrix for collection \(c\).

\paragraph{Pairwise faithfulness and numerical bound checks.}

For each construction
\[
m\in\{d_{\mathrm{SK}},W_\Gamma\},
\]
we assess faithfulness to the numerical \(W_2\) reference through two
complementary diagnostics:

\begin{itemize}

\item 
We measure global rank preservation with
\[
\rho_{c,W_2}^{m}
:=
\rho_{\mathrm{S}}
\left(
  \operatorname{ut}\!\left(D_{c,30}^{m}\right),
  \operatorname{ut}\!\left(D_c^{W_2}\right)
\right).
\]
This coefficient evaluates whether the \(L=30\) \Scorrection{SK-based constructions} and
\(W_2\) rank all distinct pairs of diagrams in a similar order.

\item We also measure preservation of the local \(W_2\) neighbourhoods
through
\[
\operatorname{NN@3}_{c,W_2}^{m}
:=
\frac{1}{3n_c}
\sum_{i=1}^{n_c}
\left|
  \mathcal{N}_3\!\left(D_{c,30}^{m},i\right)
  \cap
  \mathcal{N}_3\!\left(D_c^{W_2},i\right)
\right|.
\]
Thus, this NN@3 measure compares the three nearest neighbours
obtained with the \(L=30\) \Scorrection{SK-based constructions} to those obtained with \(W_2\).
\end{itemize}

For every distinct diagram pair, we additionally verify numerically
the ordering corresponding to the theoretical inequalities

\[
W_2(X,Y)
\leq
W_{\Gamma,30}(X,Y)
\leq
\sqrt{2}\,d_{\mathrm{SK},30}(X,Y).
\]

\paragraph{Computational performance.}

We measure the wall-clock computation time of each complete
distance-matrix filter on the same 20-thread workstation. For
\(m\in\{d_{\mathrm{SK}},W_\Gamma\}\), the collection-level speedup is
defined as
\[
S_c^{m}
:=
\frac{
  T_c^{W_2}
}{
  T_{c,30}^{m}
},
\]
where \(T_c^{W_2}\) is the computation time of the \(W_2\) matrix and
\(T_{c,30}^{m}\) that of the corresponding \(L=30\) SK matrix. The reported values are average wall-clock filter timings over five
independent runs.

\paragraph{Agreement with \(W_2\) average-linkage partitions.}

To determine whether differences between the numerical distances
affect a downstream grouping task, we additionally apply
average-linkage clustering to
\[
D_{c,30}^{d_{\mathrm{SK}}},
\qquad
D_{c,30}^{W_\Gamma},
\qquad\text{and}\qquad
D_c^{W_2},
\]
using the number of groups \(k\) provided by the collection metadata.
Each SK-based
partition is then compared directly with the corresponding
\(W_2\)-based partition using the Adjusted Rand Index (ARI).

\subsubsection{Results}

\paragraph{Pairwise faithfulness and numerical bound checks.}

The pairwise comparisons in \autoref{fig:w2_scatter} and the
dataset-level statistics in \autoref{tab:w2_comparison} show that,
across the complete 12-collection benchmark,
\(d_{\mathrm{SK},30}\) attains a median Spearman correlation of
\(0.879\) with \(W_2\), while \(W_{\Gamma,30}\) attains \(0.924\).
The corresponding median NN@3 overlaps are \(0.886\) and \(0.914\),
respectively. Hence, \(W_{\Gamma,30}\) is more faithful in median to the
\(W_2\) geometry according to both diagnostics. However,
\(d_{\mathrm{SK},30}\) still achieves a median Spearman correlation
of \(0.879\) and a median NN@3 overlap of \(0.886\), while additionally
providing the metric and Hilbertian guarantees established in
\autoref{sec:skot_hilbert_kernel}.

The collection-level results in \autoref{tab:w2_comparison} also
illustrate why global and local diagnostics are complementary. For
example, Isabel has a moderate global Spearman correlation of
\(0.603\) between \(d_{\mathrm{SK},30}\) and \(W_2\), but its
three-nearest-neighbour sets are exactly preserved. Thus, a moderate
correlation over the complete ranking of pairwise distances does not
necessarily imply that the local neighbourhood structure has been
lost.

Finally, the numerical bound checks accompanying
\autoref{fig:w2_scatter} confirm the theoretical ordering on all
\(3{,}004\) distinct pairwise diagram comparisons:
\[
W_2(X_i,X_j)
\leq
W_{\Gamma,30}(X_i,X_j)
\leq
\sqrt{2}\,d_{\mathrm{SK},30}(X_i,X_j).
\]
No value of \(W_{\Gamma,30}\) is observed below \(W_2\), and
\[
\max_{i<j}
\frac{
  W_{\Gamma,30}(X_i,X_j)
}{
  \sqrt{2}\,d_{\mathrm{SK},30}(X_i,X_j)
}
=
0.957.
\]
Moreover,
\[
\max_{i<j}
\frac{
  W_2(X_i,X_j)
}{
  \sqrt{2}\,d_{\mathrm{SK},30}(X_i,X_j)
}
=
0.902.
\]
Thus, both inequalities are satisfied for every pairwise comparison
available in the benchmark.

\begin{figure}[t]
  \centering
  \includegraphics[width=.99\linewidth]{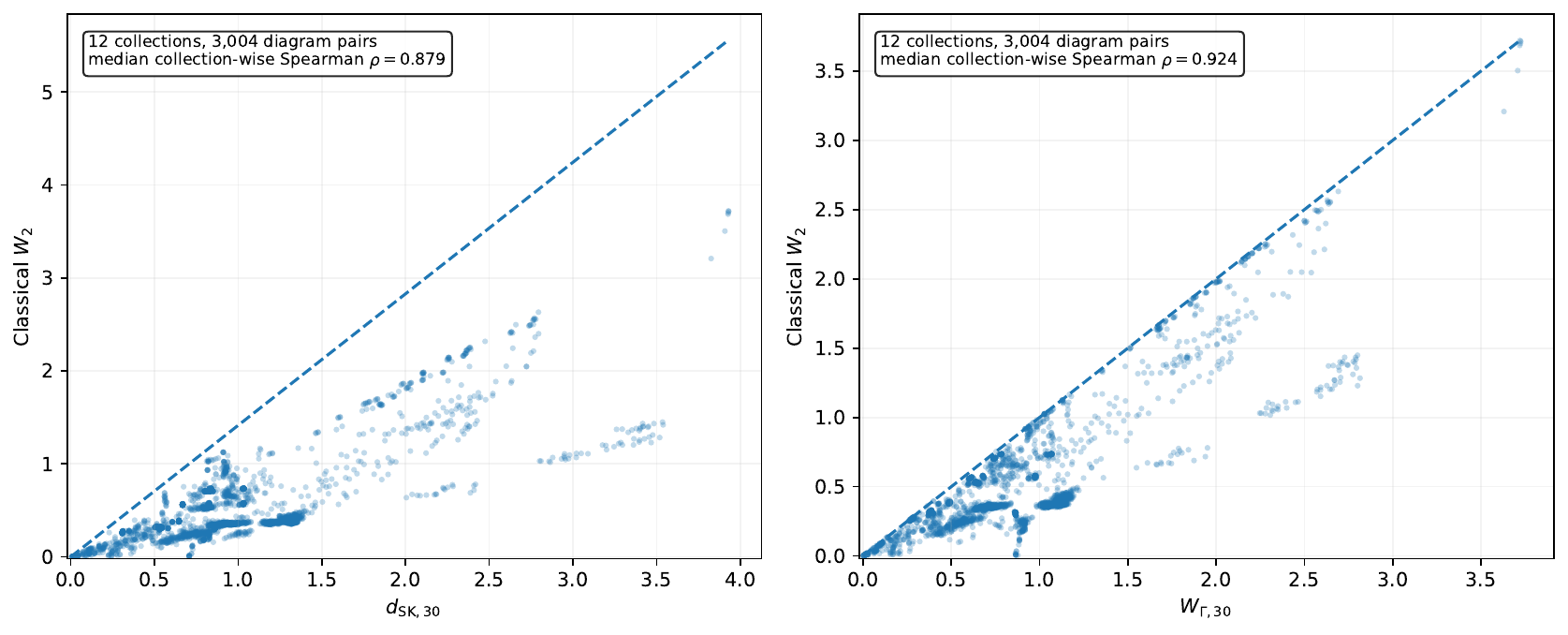}
  \caption{
Pairwise comparison with \(W_2\) on all 12 collections, comprising
\(3{,}004\) distinct pairwise diagram
evaluations. Left:
\(d_{\mathrm{SK},30}\) versus \(W_2\), together with the theoretical
upper-bound line \(W_2=\sqrt{2}\,d_{\mathrm{SK},30}\). Right:
\(W_{\Gamma,30}\) versus \(W_2\), together with the identity line.
The annotation in each panel reports the median, over the 12
collections, of the collection-wise Spearman rank correlation with
\(W_2\).
}\label{fig:w2_scatter}
\end{figure}

\input{tables/table_w2_comparison.tex}

\paragraph{Computational performance}

The computational results are reported per collection in
\autoref{tab:w2_comparison} and summarized graphically in
\autoref{fig:w2_speedup}. The speedup of \(d_{\mathrm{SK},30}\)
over the TTK \(W_2\) computation ranges from \(34.7\times\) to
\(5007\times\), with a median of \(626\times\). For
\(W_{\Gamma,30}\), the speedup ranges from \(40.8\times\) to
\(4561\times\), with a median of \(586\times\).

Summed over the complete benchmark, the recorded filter-computation
times are \(3401\) seconds for \(W_2\) and \(1.62\) seconds for
\(d_{\mathrm{SK},30}\). As a representative large-diagram example,
the Earthquake collection has a median diagram cardinality of
approximately \(49{,}568\) persistence pairs. Its complete \(W_2\)
matrix requires \(659\) seconds, whereas the corresponding
\(d_{\mathrm{SK},30}\) matrix requires only \(0.132\) seconds, as
reported in \autoref{tab:w2_comparison}. These timings illustrate the
practical difference between repeatedly solving planar assignment
problems and evaluating the sorting-based SK construction.

\begin{figure}[t]
  \centering
  \includegraphics[width=.99\linewidth]
  {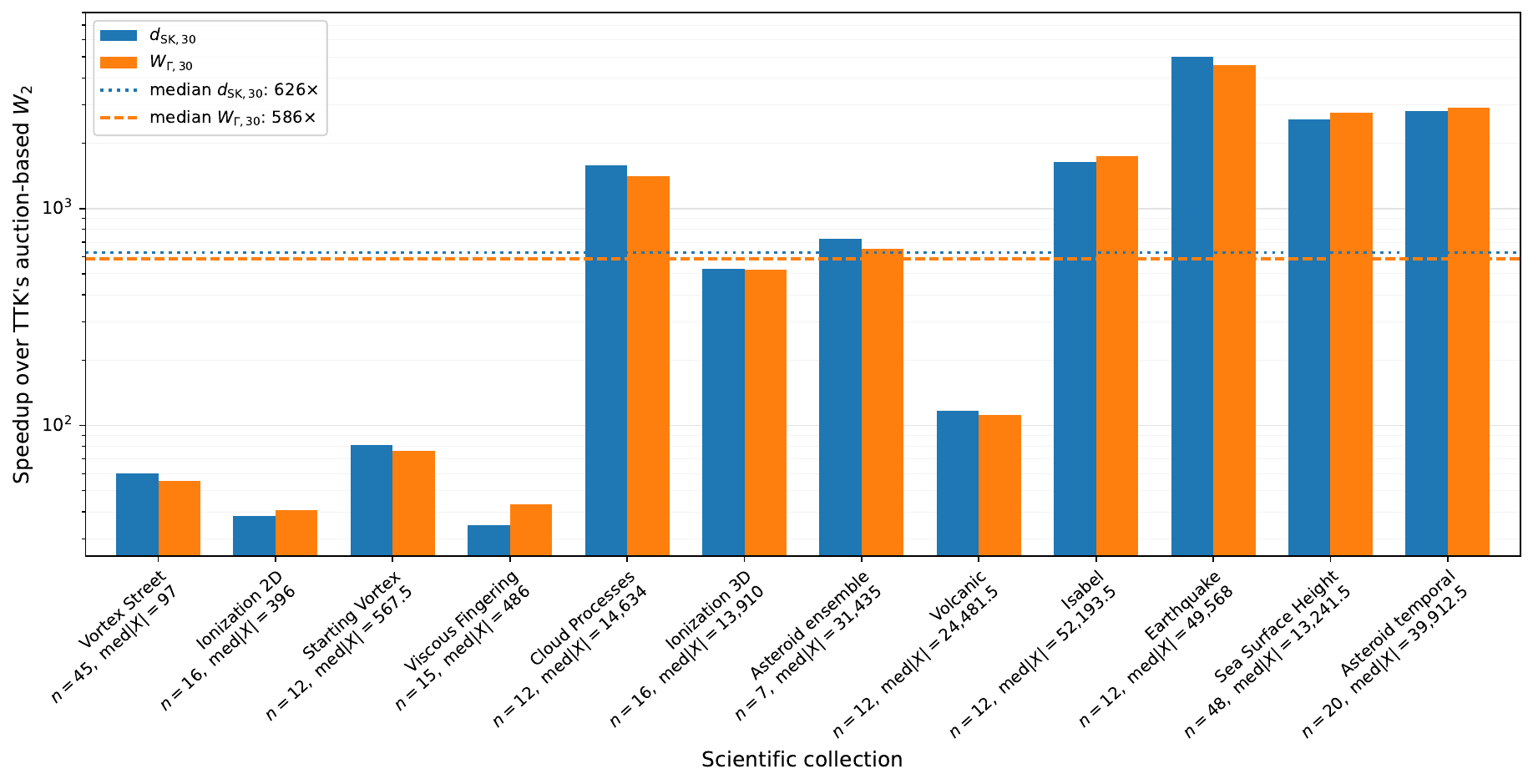}
  \caption{
Filter-computation speedup relative to TTK's auction-based
\(W_2\) computation at \(L=30\). Collections are ordered from left
to right by increasing total number of persistence pairs. For each
collection, the horizontal-axis label reports the number \(n\) of
diagrams and the median diagram cardinality \(\mathrm{med}|X|\). The
vertical axis is logarithmic; horizontal lines indicate median
speedups of \(626\times\) for \(d_{\mathrm{SK},30}\) and
\(586\times\) for \(W_{\Gamma,30}\).
}
  \label{fig:w2_speedup}
\end{figure}

\paragraph{Agreement with \(W_2\) average-linkage partitions}

The detailed average-linkage results are reported in
\autoref{tab:w2_partition_agreement} and displayed graphically in
\autoref{fig:w2_partition_agreement}. The partition obtained from
\(d_{\mathrm{SK},30}\) is exactly identical to the \(W_2\) partition
on 8 of the 12 collections. The same exact-agreement count is obtained
for \(W_{\Gamma,30}\). For both constructions, the median ARI with the
\(W_2\) partition is \(1.000\); the corresponding mean ARI values are
\(0.897\) for \(d_{\mathrm{SK},30}\) and \(0.905\) for
\(W_{\Gamma,30}\).

The four collections for which the \(d_{\mathrm{SK},30}\) partition
is not exactly identical to the \(W_2\) partition are Viscous
Fingering, Asteroid Impact temporal, Sea Surface Height, and Vortex
Street. The largest discrepancies occur on Asteroid Impact temporal,
Sea Surface Height, and Vortex Street. Overall, the \(d_{\mathrm{SK}}\) constructions
therefore preserve the average-linkage partition induced by \(W_2\)
on most collections, with perfect median agreement, but this
preservation is not universal.

\input{tables/table_w2_partition_agreement_supplementary.tex}

\begin{figure}[t]
  \centering
  \includegraphics[width=.86\linewidth]
  {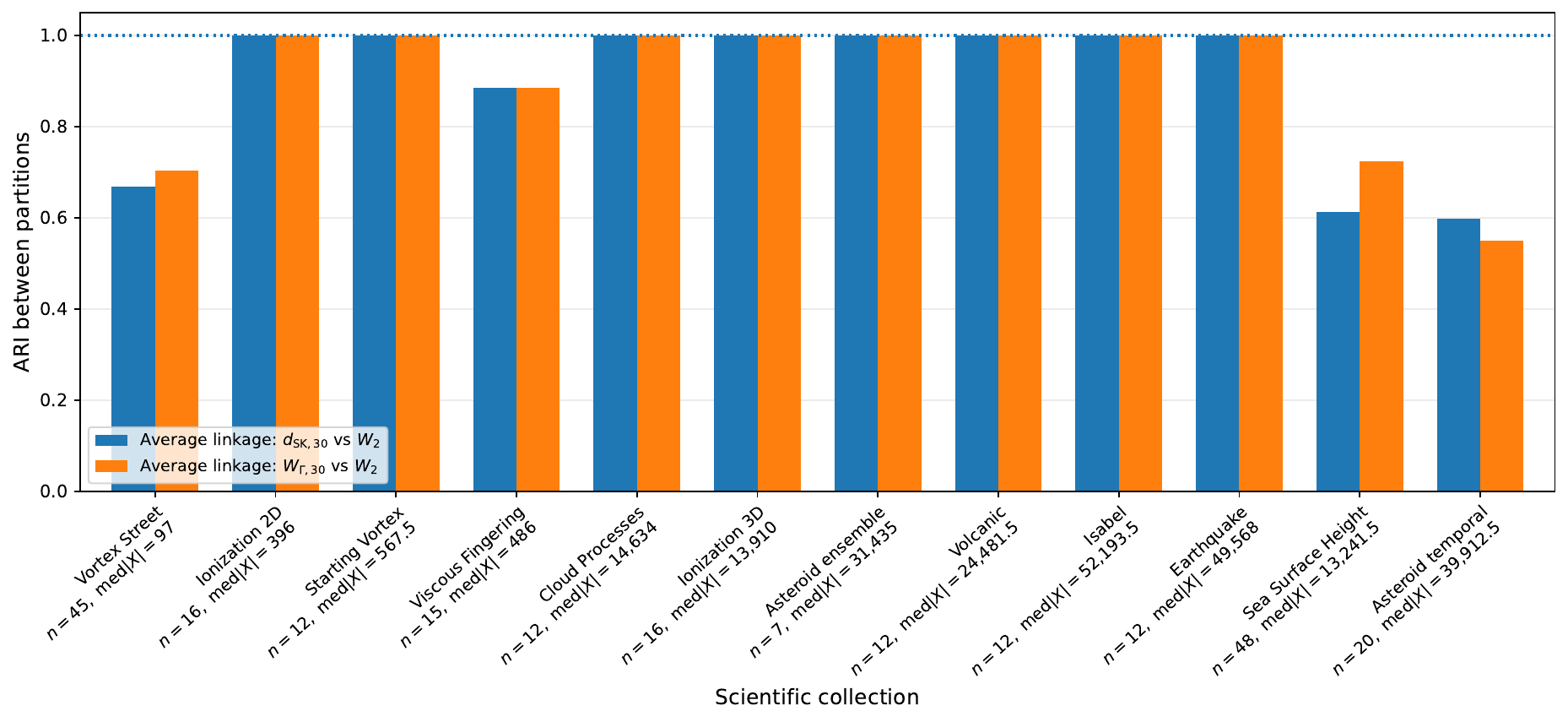}
  \caption{
  Agreement between the average-linkage partitions obtained from
  \(d_{\mathrm{SK},30}\) and \(W_{\Gamma,30}\), and the corresponding
  partition obtained from \(W_2\). Agreement is measured by the
  adjusted Rand index (ARI), with the number of groups fixed from the
  benchmark metadata. An ARI of \(1\) indicates identical partitions.
  }
  \label{fig:w2_partition_agreement}
\end{figure}

\subsection{Application to Hilbertian and kernel-based clustering}
\label{sec:kernel_ensemble_results}

\subsubsection{Protocol}

To assess the practical value of the Hilbertian and kernel structures
established in \autoref{sec:skot_hilbert} and
\autoref{sec:skot_kernel}, we apply them to clustering the 12
scientific collections. We compare Hilbert \(k\)-means and Gaussian spectral clustering, both
based on \(d_{\mathrm{SK}}\), with average linkage on
\(d_{\mathrm{SK}}\) as a distance-based baseline.

Following the convergence study in \autoref{sec:selector_convergence_results}, we set the refinement
level to
\[
L=30
\]
throughout the remainder of this subsection.

For each collection \(c\), containing \(n_c\) persistence diagrams,
let
\[
D_{c,30}(i,j)
=
d_{\mathrm{SK},30}(X_i,X_j).
\]
denote its $\Scorrection{d_{\mathrm{SK},30}}$ distance matrix.

To obtain explicit Euclidean coordinates associated with the
Hilbertian geometry, we construct the centered Gram matrix
\[
G_{c,30}
=
-\frac12
J_cD_{c,30}^{\circ2}J_c,
\qquad
J_c
=
I-\frac1{n_c}\mathbf{1}\mathbf{1}^{\top},
\]
where \(D_{c,30}^{\circ2}\) denotes entrywise squaring. Let
\[
G_{c,30}
=
V\Lambda V^{\top}
\]
be its eigendecomposition. Retaining all coordinates associated with
positive eigenvalues gives the full-dimensional classical-MDS
realization
\[
Z_{c,30}
=
V_{+}\Lambda_{+}^{1/2}.
\]

We also construct the Gaussian kernel associated with
\(d_{\mathrm{SK},30}\), as introduced in
\autoref{sec:skot_kernel}:
\[
K_{c,30}(i,j)
=
\exp\left(
-\frac{
  D_{c,30}(i,j)^2
}{
  2\sigma_{c,30}^2
}
\right),
\]
where
\[
\sigma_{c,30}
=
\operatorname{median}
\left\{
D_{c,30}(i,j):
i<j,\;
D_{c,30}(i,j)>0
\right\}.
\]
The bandwidth is therefore selected independently for each collection
as the median of its positive pairwise
\(d_{\mathrm{SK},30}\) distances.

For each collection, we compare three unsupervised clustering
pipelines:

\begin{itemize}

\item \emph{Average linkage on the \(d_{\mathrm{SK}}\) distance.}
We apply average-linkage clustering directly to the pairwise distance
matrix
\[
D_{c,30}.
\]
This provides a baseline that exploits only the $d_{\mathrm{SK}}$ distance values.

\item \emph{Hilbert \(k\)-means.}
We apply standard Euclidean \(k\)-means to the full-dimensional
classical-MDS coordinates
\[
Z_{c,30}.
\]
All coordinates associated with positive eigenvalues are retained.
Thus, this pipeline uses the complete Hilbertian realization of
\(d_{\mathrm{SK}}\).

\item \emph{Gaussian \(d_{\mathrm{SK}}\) spectral clustering.}
We apply normalized spectral clustering to
\[
K_{c,30}
\]
as a similarity matrix~\cite{ng2002spectral,dhillon2004kernelkm}.
Although spectral clustering can also be applied to a general
symmetric nonnegative affinity matrix, the $\Scorrection{d_{\mathrm{SK}}^{}}$ theory guarantees
that \(K_{c,30}\) is a positive semidefinite Gram matrix and therefore
admits an RKHS interpretation.

\end{itemize}

For all three \(d_{\mathrm{SK}}\)-based pipelines, the number of groups \(k\) is
read from the collection metadata. After clustering, each resulting
partition is compared with the reference partition supplied by the
metadata using ARI. Therefore, the ARI
values reported for average linkage, Hilbert \(k\)-means, and Gaussian
\(d_{\mathrm{SK}}\) spectral clustering measure agreement with the benchmark
reference groups.

For contextual comparison, we also report the ARI of the
average-linkage partition obtained from the numerical \(W_2\) matrix
with respect to the same benchmark reference partition.

\subsubsection{Results}

Averaged over the 12 collections, average linkage on
\(d_{\mathrm{SK},30}\) obtains a mean ARI of \(0.667\), Hilbert
\(k\)-means obtains \(0.756\), and Gaussian
\(d_{\mathrm{SK}}\) spectral clustering obtains \(0.800\), as
reported in \autoref{tab:kernel_clustering} and illustrated in
\autoref{fig:clustering_ari}. Gaussian \(d_{\mathrm{SK}}\)
spectral clustering therefore achieves the highest mean ARI among the
three evaluated pipelines. These three \(d_{\mathrm{SK}}\)-based pipelines achieve perfect agreement
with the reference partitions on 5, 6, and 7 of the 12 collections,
respectively. For comparison, average-linkage clustering
on \(W_2\) obtains a mean ARI of \(0.750\) and perfect agreement on 7
collections, as reported in
\autoref{tab:w2_partition_agreement}.

Compared to direct average linkage on \(d_{\mathrm{SK}}\), Gaussian \(d_{\mathrm{SK}}\)
spectral clustering improves 5 collections and worsens none. The largest gains occur for Viscous Fingering ($0.441\rightarrow1$), Sea Surface Height ($0.613\rightarrow1$), Volcanic Eruptions ($0.408\rightarrow0.737$), and Vortex Street ($0.668\rightarrow0.850$). Thus, the positive-definite Gaussian \(d_{\mathrm{SK}}\) kernel established in \autoref{sec:skot_kernel} is not merely formal: it enables a nonlinear learning pipeline that extracts additional ensemble structure from the \(d_{\mathrm{SK}}\) matrix.

The result is not uniformly perfect. In particular, the seven-member Asteroid Impact clustering collection is not aligned with the supplied two-class metadata for any of the three \(d_{\mathrm{SK}}\)-based pipelines (ARI $=-0.077$). Earthquake also remains only partially aligned with its three metadata groups (ARI $=0.368$). However, as reported in
\autoref{tab:w2_partition_agreement}, average-linkage clustering on
the \(W_2\) matrix also achieves ARI values of \(-0.077\) and
\(0.368\) with respect to the benchmark reference partitions on the
Asteroid Impact clustering and Earthquake collections, respectively.
These are exactly the same ARI values as those obtained by
average-linkage clustering on \(d_{\mathrm{SK},30}\). Therefore, in
this clustering experiment, the limited agreement observed on these
two collections cannot be attributed to replacing \(W_2\) with
\(d_{\mathrm{SK},30}\).

\begin{figure}[!htbp]
  \centering
  \includegraphics[width=.94\linewidth]
  {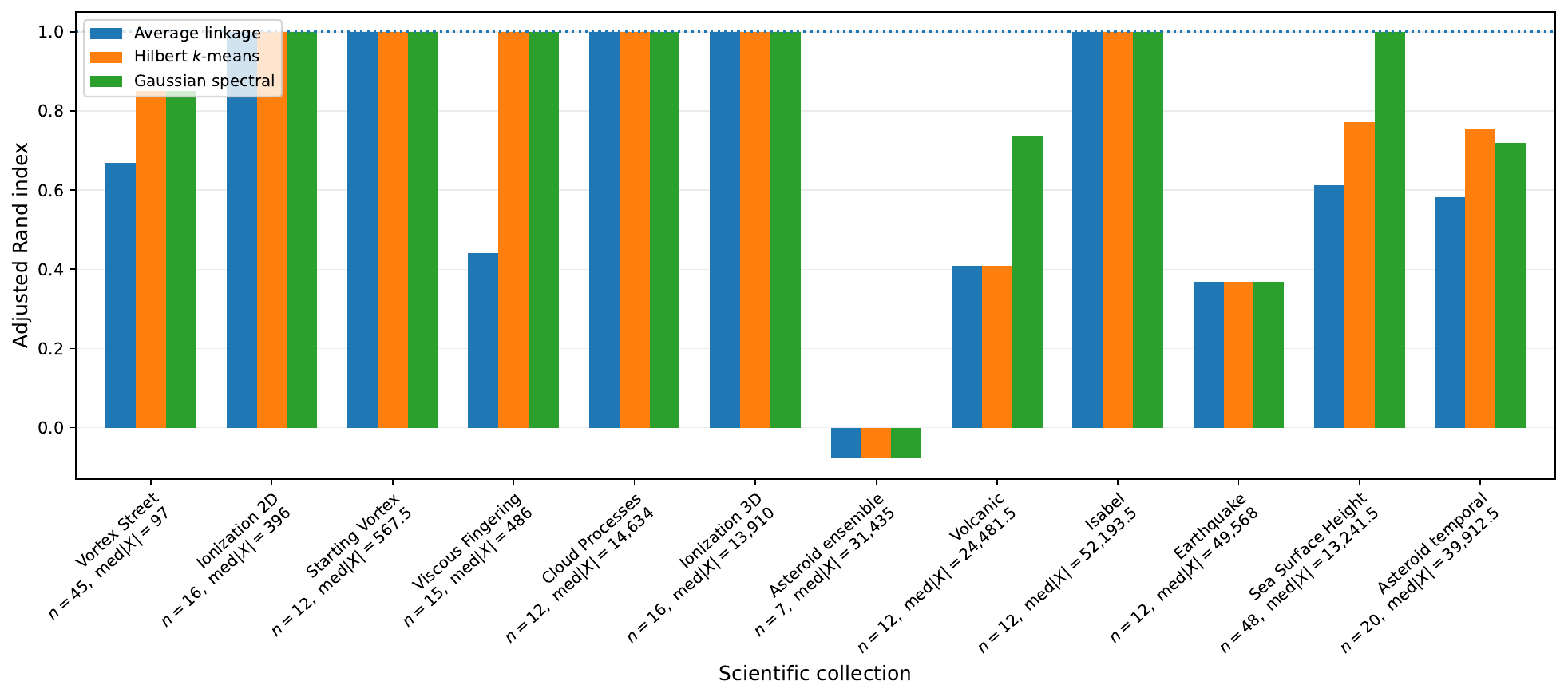}
  \caption{
  ARI with the benchmark reference partitions for the three
  SK-based clustering pipelines at \(L=30\): average linkage on
  \(d_{\mathrm{SK}}\), \(d_{\mathrm{SK}}\) Hilbert \(k\)-means on the
  full-dimensional classical-MDS coordinates, and normalized
  spectral clustering with the Gaussian $\Scorrection{d_{\mathrm{SK}}^{}}$ kernel. Gaussian \(d_{\mathrm{SK}}\)
  spectral clustering achieves the highest mean ARI and does not
  decrease the ARI relative to average linkage on
  \(d_{\mathrm{SK}}\) on any of the 12 collections.
  }
  \label{fig:clustering_ari}
\end{figure}

\input{tables/table_kernel_clustering.tex}

\subsection{Temporal Hilbertian case study}
\label{sec:temporal_case_study}

\subsubsection{Protocol}

We use this temporal case study to illustrate the practical use of the
Hilbertian realization and Gaussian kernel associated with
\(d_{\mathrm{SK}}\) for trajectory visualization and contiguous
segmentation of ordered diagram sequences. We consider the 20-diagram Asteroid
Impact temporal collection as an ordered diagram-valued sequence.

\paragraph{Hilbertian and kernel visualization.}

We use the Hilbertian realization of
\(d_{\mathrm{SK},30}\) as the main temporal representation and the
associated Gaussian kernel as a complementary similarity view.
Specifically, we visualize the sequence through:

\begin{itemize}

\item the first two coordinates of the full-dimensional Hilbertian
realization \(Z_{c,30}\), with consecutive diagrams connected in
temporal order;

\item the Gaussian \(d_{\mathrm{SK}}\) kernel matrix
\(K_{c,30}\), displayed as a similarity heatmap.

\end{itemize}

\paragraph{Contiguous kernel segmentation.}

We apply contiguous kernel
segmentation to the 20-diagram Asteroid Impact temporal sequence,
whose reference metadata labels form four contiguous blocks in
temporal order. The number of segments is therefore fixed to four
from the metadata, and the bandwidth of the Gaussian \(d_{\mathrm{SK}}\) kernel is set to the median of
the positive pairwise \(d_{\mathrm{SK},30}\) distances, without
label-dependent tuning. The segment boundaries are obtained by
dynamic programming using the kernel within-segment dispersion
\cite{arlot2019kernelchangepoint}. This experiment illustrates how
the Gaussian \(d_{\mathrm{SK}}\) kernel can be used directly for
kernel change-point detection and contiguous segmentation of a
diagram-valued temporal sequence.

\subsubsection{Results}

\paragraph{Hilbertian and kernel visualization.}

\autoref{fig:asteroid_hilbert_kernel} illustrates the diagram-valued
trajectory of the 20-step Asteroid Impact temporal sequence. The
two-dimensional Hilbert view, computed from the Hilbert Gram matrix,
reveals a smooth progression between the four annotated phases. The
Gaussian kernel displays the same evolution as blocks of high
within-phase similarity and gradual cross-phase decay. 

\begin{figure}[t]
  \centering
  \subfloat[Two-dimensional projection of the Hilbert embedding. Consecutive samples are connected in temporal order.]{\includegraphics[width=.49\linewidth]{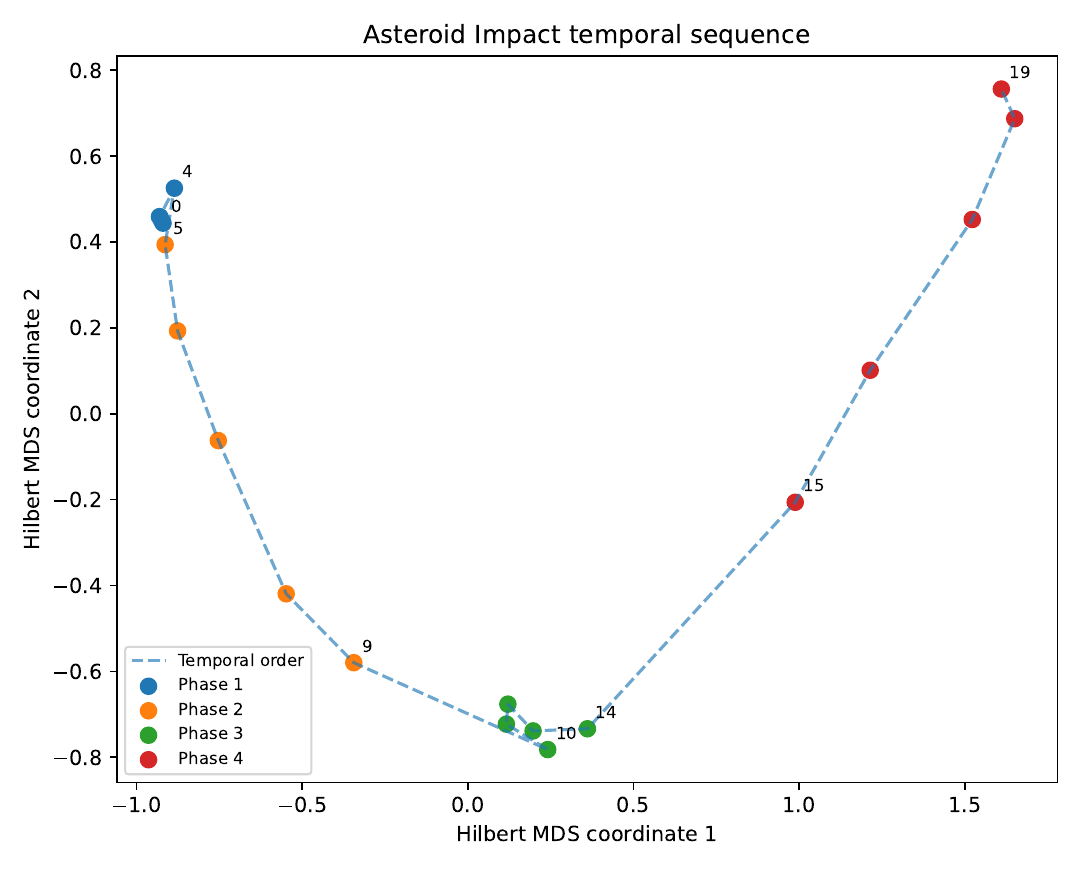}}
  \hfill
  \subfloat[Gaussian \(d_{\mathrm{SK}}\) kernel. Solid lines indicate the annotated phases and dashed lines the kernel segmentation.]{\includegraphics[width=.49\linewidth]{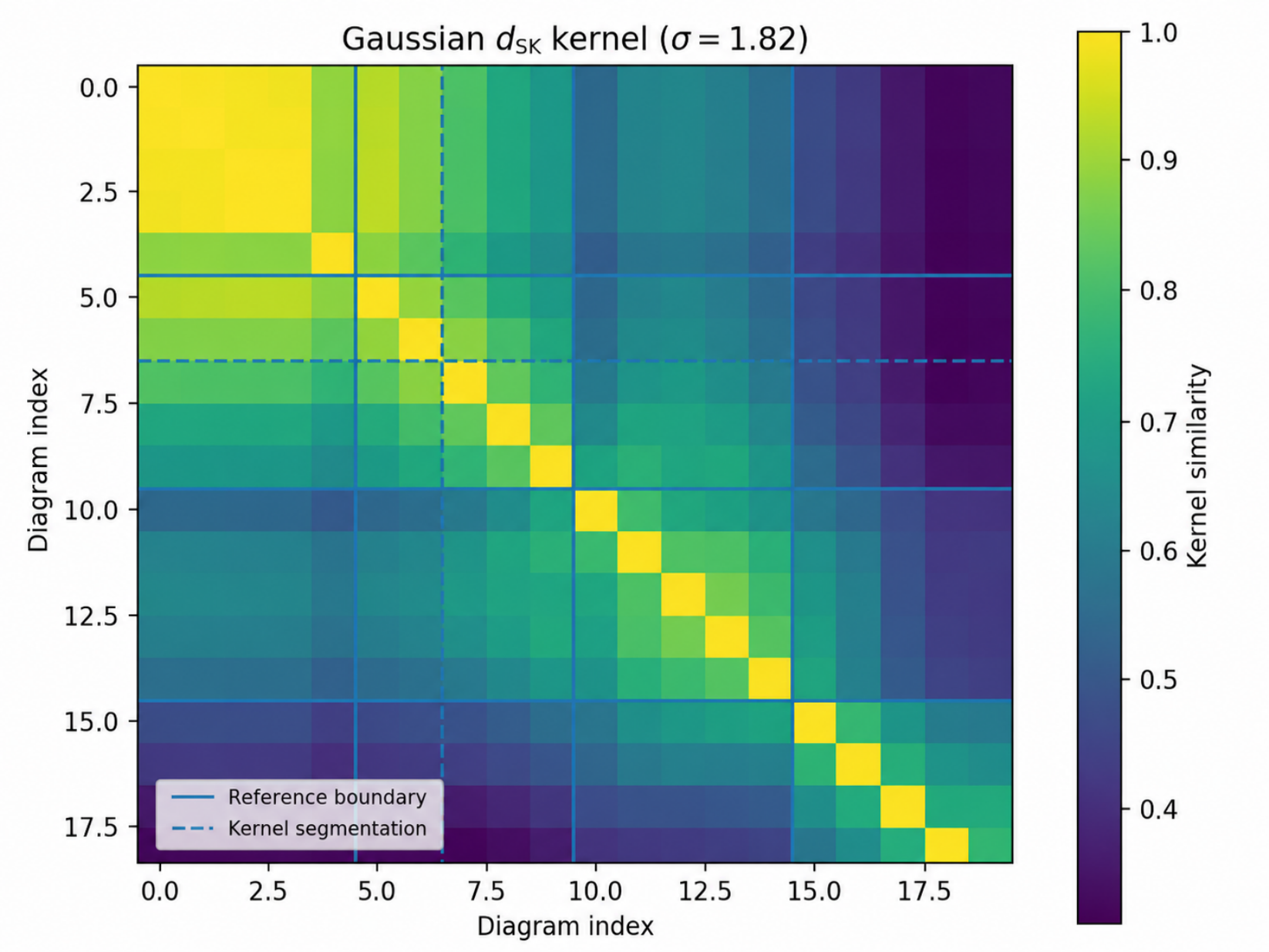}}
  \caption{Hilbertian and kernel views of the Asteroid Impact temporal sequence.}
  \label{fig:asteroid_hilbert_kernel}
\end{figure}

\paragraph{Contiguous kernel segmentation.}

The reference boundaries are
\[
[5,10,15],
\]
whereas the kernel segmentation predicts
\[
[7,10,15].
\]

Kernel segmentation recovers two reference boundaries exactly and
places the first boundary two samples late, yielding a boundary MAE
of \(0.67\) and an ARI of \(0.756\), tying the best ARI achieved on
this collection among the three clustering methods evaluated in
\autoref{sec:kernel_ensemble_results}, namely the \(0.756\) obtained
by Hilbert \(k\)-means. This figure focuses on the Asteroid Impact
temporal case study. The supplementary material
(\autoref{sec:kernel_segmentation_supp}) extends the evaluation to all
eight eligible ordered collections and details the kernel-based
contiguous-segmentation method employed; the complete results are
reported in \autoref{tab:kernel_segmentation_supp}, with exact
recovery of the full reference segmentation on five collections.
Across the eight eligible collections, kernel segmentation ties the
best ARI among the three clustering methods evaluated in
\autoref{sec:kernel_ensemble_results} on six collections and exceeds
all three on Earthquake; on the remaining collection, Sea Surface
Height, it ties the second-best ARI.

%% file: tables/table_dataset_summary.tex
\begin{table}[t]
\centering
\caption{
Scientific collections used in the evaluation.
Here, \(n\) denotes the number of persistence diagrams in the
collection, \(k\) the number of classes supplied by the
benchmark ground-truth, and \(|X|\) the number of off-diagonal
persistence pairs in a diagram \(X\). Diagram cardinalities aggregate
all retained critical-pair types and homological dimensions returned
by the Discrete Morse Sandwich backend. 
}
\label{tab:dataset_summary}
\scriptsize
\begin{tabular}{lrrrrr}
\toprule
Dataset
& \(n\)
& \(k\)
& Min. \(|X|\)
& Median \(|X|\)
& Max. \(|X|\) \\
\midrule
Isabel
& 12 & 3 & 6\,344 & 52\,194 & 54\,763 \\

Earthquake
& 12 & 3 & 36\,904 & 49\,568 & 65\,334 \\

Ionization 2D
& 16 & 4 & 171 & 396 & 463 \\

Ionization 3D
& 16 & 4 & 3\,376 & 13\,910 & 18\,667 \\

Volcanic
& 12 & 3 & 23\,897 & 24\,482 & 24\,759 \\

Viscous Fingering
& 15 & 3 & 409 & 486 & 589 \\

Cloud Processes
& 12 & 3 & 14\,540 & 14\,634 & 14\,872 \\

Asteroid ensemble
& 7 & 2 & 10\,563 & 31\,435 & 102\,735 \\

Asteroid temporal
& 20 & 4 & 1\,132 & 39\,912 & 70\,557 \\

Sea Surface Height
& 48 & 4 & 12\,903 & 13\,242 & 13\,438 \\

Starting Vortex
& 12 & 2 & 466 & 568 & 693 \\

Vortex Street
& 45 & 5 & 22 & 97 & 144 \\
\bottomrule
\end{tabular}
\end{table}

%% file: tables/table_convergence.tex
\begin{table}[t]
\centering
\caption{
Convergence to the high-resolution \(L=40\) numerical reference over
the 12 collections. For each collection \(c\),
\(E_c(L)\) is the relative Frobenius error,
\(\rho_c(L)\) is the Spearman correlation, between the strict
upper-triangular entries of the \(L\), and \(L=40\), level matrices,
and \(\operatorname{NN@3}_c(L)\) is the mean overlap of their
three-nearest-neighbour sets. The table reports the median and
maximum of \(E_c(L)\), and the minimum values of \(\rho_c(L)\) and
\(\operatorname{NN@3}_c(L)\), over the collections. Thus, the last
two columns report the worst collection.
}
\label{tab:selector_convergence}
\small
\begin{tabular}{lccccc}
\toprule
Method
& \(L\)
& Median \(E_c(L)\)
& Max. \(E_c(L)\)
& Min. \(\rho_c(L)\)
& Min. \(\operatorname{NN@3}_c(L)\) \\
\midrule
\(d_{\mathrm{SK},L}\)
& 10
& \(6.74\times10^{-2}\)
& \(8.11\times10^{-1}\)
& 0.604599
& 0.815 \\

\(d_{\mathrm{SK},L}\)
& 20
& \(1.97\times10^{-4}\)
& \(1.44\times10^{-2}\)
& 0.998330
& 0.952 \\

\(d_{\mathrm{SK},L}\)
& 30
& \(6.97\times10^{-8}\)
& \(1.41\times10^{-5}\)
& 1.000000
& 1.000 \\
\midrule
\(W_{\Gamma,L}\)
& 10
& \(1.17\times10^{-1}\)
& \(1.11\)
& 0.492412
& 0.583 \\

\(W_{\Gamma,L}\)
& 20
& \(4.89\times10^{-4}\)
& \(2.11\times10^{-2}\)
& 0.993988
& 0.917 \\

\(W_{\Gamma,L}\)
& 30
& \(4.50\times10^{-6}\)
& \(6.63\times10^{-5}\)
& 1.000000
& 1.000 \\
\bottomrule
\end{tabular}
\end{table}

%% file: tables/table_w2_comparison.tex
\begin{table}[t]
\centering
\caption{Faithfulness and computational performance at $L=30$ on the 12-collection benchmark (\(3{,}004\) distinct pairwise diagram
comparisons for each dissimilarity). $\rho$ is Spearman rank correlation, NN@3 is the mean overlap of the three-nearest-neighbour sets, and speedups use filter compute time.}
\label{tab:w2_comparison}
\scriptsize
\setlength{\tabcolsep}{2.6pt}
\resizebox{\linewidth}{!}{%
\begin{tabular}{lrrrrrrrr}
\toprule
Dataset & $n$ & Med. $|X|$ & $\rho(d_{\mathrm{SK}},W_2)$ & $\rho(W_\Gamma,W_2)$ & NN@3 $d_{\mathrm{SK}}$ & NN@3 $W_\Gamma$ & Speedup $d_{\mathrm{SK}}$ & Speedup $W_\Gamma$ \\
\midrule
Isabel & 12 & 52\,194 & 0.603 & 0.643 & 1.000 & 1.000 & 1640$\times$ & 1745$\times$ \\
Earthquake & 12 & 49\,568 & 0.896 & 0.939 & 0.694 & 0.722 & 5007$\times$ & 4561$\times$ \\
Ionization 2D & 16 & 396 & 0.922 & 0.878 & 1.000 & 0.958 & 38$\times$ & 41$\times$ \\
Ionization 3D & 16 & 13\,910 & 0.625 & 0.835 & 1.000 & 1.000 & 528$\times$ & 523$\times$ \\
Volcanic & 12 & 24\,482 & 0.885 & 0.915 & 0.778 & 0.778 & 117$\times$ & 112$\times$ \\
Viscous Fingering & 15 & 486 & 0.840 & 0.932 & 0.889 & 0.911 & 35$\times$ & 43$\times$ \\
Cloud Processes & 12 & 14\,634 & 0.953 & 0.955 & 1.000 & 1.000 & 1581$\times$ & 1404$\times$ \\
Asteroid ensemble & 7 & 31\,435 & 0.979 & 0.986 & 0.952 & 0.952 & 724$\times$ & 649$\times$ \\
Asteroid temporal & 20 & 39\,912 & 0.875 & 0.941 & 0.883 & 0.917 & 2818$\times$ & 2914$\times$ \\
Sea Surface Height & 48 & 13\,242 & 0.857 & 0.862 & 0.792 & 0.785 & 2578$\times$ & 2765$\times$ \\
Starting Vortex & 12 & 568 & 0.883 & 0.969 & 0.806 & 0.806 & 81$\times$ & 76$\times$ \\
Vortex Street & 45 & 97 & 0.763 & 0.675 & 0.696 & 0.719 & 60$\times$ & 56$\times$ \\
\midrule
Median & -- & -- & 0.879 & 0.924 & 0.886 & 0.914 & 626$\times$ & 586$\times$ \\ 
\bottomrule
\end{tabular}}
\end{table}

%% file: tables/table_w2_partition_agreement_supplementary.tex
\begin{table}[t]
\centering
\caption{
Agreement between average-linkage partitions on the complete
12-collection benchmark. The first two ARI columns compare the
partitions obtained from the SK constructions with the partition
obtained from \(W_2\). The last three columns compare each partition
with the reference groups supplied by the benchmark metadata. The
number of groups \(k\) is fixed from these metadata.
}
\label{tab:w2_partition_agreement}
\scriptsize
\begin{tabular}{lrrrrr}
\toprule
Dataset
& ARI \(d_{\mathrm{SK}}/W_2\)
& ARI \(W_\Gamma/W_2\)
& ARI ref./\(W_2\)
& ARI ref./\(d_{\mathrm{SK}}\)
& ARI ref./\(W_\Gamma\) \\
\midrule
Isabel
& 1.000 & 1.000 & 1.000 & 1.000 & 1.000 \\

Earthquake
& 1.000 & 1.000 & 0.368 & 0.368 & 0.368 \\

Ionization 2D
& 1.000 & 1.000 & 1.000 & 1.000 & 1.000 \\

Ionization 3D
& 1.000 & 1.000 & 1.000 & 1.000 & 1.000 \\

Volcanic
& 1.000 & 1.000 & 0.408 & 0.408 & 0.408 \\

Viscous Fingering
& 0.886 & 0.886 & 0.441 & 0.441 & 0.441 \\

Cloud Processes
& 1.000 & 1.000 & 1.000 & 1.000 & 1.000 \\

Asteroid ensemble
& 1.000 & 1.000 & -0.077 & -0.077 & -0.077 \\

Asteroid temporal
& 0.598 & 0.550 & 0.859 & 0.582 & 0.477 \\

Sea Surface Height
& 0.613 & 0.724 & 1.000 & 0.613 & 0.724 \\

Starting Vortex
& 1.000 & 1.000 & 1.000 & 1.000 & 1.000 \\

Vortex Street
& 0.668 & 0.704 & 1.000 & 0.668 & 0.704 \\
\midrule
Mean
& 0.897 & 0.905 & 0.750 & 0.667 & 0.670 \\
\bottomrule
\end{tabular}
\end{table}

%% file: tables/table_kernel_clustering.tex
\begin{table}[t]
\centering
\caption{Adjusted Rand index (ARI) at $L=30$. The number of clusters $k$ is fixed from the dataset metadata. ``Hilbert $k$-means'' uses the full-dimensional exact Euclidean realization of $d_{\mathrm{SK}}$; ``Gaussian spectral'' uses $k_\sigma(X,Y)=\exp[-d_{\mathrm{SK}}(X,Y)^2/(2\sigma^2)]$ with the median-distance heuristic. The ``Average linkage'' column here
reproduces the ``ARI ref./\(d_{\mathrm{SK}}\)'' column of
\autoref{tab:w2_partition_agreement}; it is repeated here as the
metric-space baseline for comparison with Hilbert \(k\)-means and
Gaussian $\Scorrection{d_{\mathrm{SK}}^{}}$ spectral clustering.}
\label{tab:kernel_clustering}
\scriptsize
\setlength{\tabcolsep}{4pt}
\begin{tabular}{lrrrrr}
\toprule
Dataset & $n$ & $k$ & Average linkage & Hilbert $k$-means & Gaussian spectral \\
\midrule
Isabel & 12 & 3 & \textbf{1.000} & \textbf{1.000} & \textbf{1.000} \\ 
Earthquake & 12 & 3 & \textbf{0.368} & \textbf{0.368} & \textbf{0.368} \\ 
Ionization 2D & 16 & 4 & \textbf{1.000} & \textbf{1.000} & \textbf{1.000} \\ 
Ionization 3D & 16 & 4 & \textbf{1.000} & \textbf{1.000} & \textbf{1.000} \\ 
Volcanic & 12 & 3 & 0.408 & 0.408 & \textbf{0.737} \\ 
Viscous Fingering & 15 & 3 & 0.441 & \textbf{1.000} & \textbf{1.000} \\ 
Cloud Processes & 12 & 3 & \textbf{1.000} & \textbf{1.000} & \textbf{1.000} \\ 
Asteroid ensemble & 7 & 2 & \textbf{-0.077} & \textbf{-0.077} & \textbf{-0.077} \\ 
Asteroid temporal & 20 & 4 & 0.582 & \textbf{0.756} & 0.720 \\ 
Sea Surface Height & 48 & 4 & 0.613 & 0.773 & \textbf{1.000} \\ 
Starting Vortex & 12 & 2 & \textbf{1.000} & \textbf{1.000} & \textbf{1.000} \\ 
Vortex Street & 45 & 5 & 0.668 & \textbf{0.850} & \textbf{0.850} \\ 
\midrule
Mean & -- & -- & 0.667 & 0.756 & \textbf{0.800} \\ 
\bottomrule
\end{tabular}
\end{table}

%% file: limitations_section.tex
\section{Limitations}\label{sec:limitations}

\Scorrection{A practical limitation of the current construction is that all
diagrams to be compared must be supported in the same normalized
persistence triangle \(\Atri\). When this is not already the case,
normalizing the collection requires a common bounding box containing all birth--death pairs,
to be known in advance. This requirement is natural in offline
collection analysis, as considered in our experiments, but may be
restrictive in online or iterative settings, such as persistence
optimization, where future diagrams are not available beforehand.
Expanding the normalization range would modify the scalar encodings of
previously processed diagrams and may therefore require recomputing
their \(d_{\mathrm{SK}}\) distances and derived representations.}

\Scorrection{A second limitation is that the bound
\[
W_2(X,Y)
\leq
\sqrt{2}\,d_{\mathrm{SK}}(X,Y)
\]
provides only one-sided control: it prevents \(d_{\mathrm{SK}}\) from
arbitrarily underestimating \(W_2\), but provides no
dataset-independent relative upper bound on \(d_{\mathrm{SK}}\) in
terms of \(W_2\). This discrepancy has two structural sources. First, the monotone assignment
induced by the SK curve may differ from the
\(W_2\)-optimal planar assignment. Second,
\(d_{\mathrm{SK}}\) evaluates this assignment using scalar
differences on \([0,1]\), rather than the diagonal-aware quadratic
cost in the persistence triangle. In particular, an assignment
between two distinct diagonal projections may contribute a positive
cost to \(d_{\mathrm{SK}}\), whereas any diagonal-to-diagonal
assignment has zero cost in \(W_2\). Consequently,
\(d_{\mathrm{SK}}\) may substantially overestimate \(W_2\).}

\Scorrection{A separate limitation concerns \(W_\Gamma\). Although it re-evaluates
the SK-induced assignment using the same diagonal-aware quadratic
cost as \(W_2\), this assignment is selected independently for each
pair of diagrams according to
the SK ordering, and the resulting pairwise
assignments can be mutually incompatible across triples.
Consequently, the triangle inequality may fail, and \(W_\Gamma\)
must be treated as a dissimilarity rather than a metric. It remains
usable by methods that accept arbitrary dissimilarities, such as
average-linkage clustering. However, it cannot support methods or
guarantees that rely on the triangle inequality, such as
metric-based nearest-neighbour indexing and pruning or approximation
guarantees for metric clustering. Moreover, unlike
\(d_{\mathrm{SK}}\), \(W_\Gamma\) has no Hilbertian guarantee, so a
Gaussian transformation of \(W_\Gamma\) is not guaranteed to be a
positive-definite kernel.}

\Scorrection{Finally, the quadratic \(W_2\) cost assigns disproportionate weight to
large assignment displacements, making \(W_2\)-based aggregation
sensitive to outlier diagrams
\cite{sisouk2026robustbarycenterspersistencediagrams}.
\(W_\Gamma\) inherits this limitation because it evaluates the
SK-induced coupling with the same quadratic planar cost. A natural
direction for future work is therefore to develop
space-filling-curve-based surrogates for \(W_p\), with
\(1\leq p<2\), whose lower-order transport costs may provide greater
robustness to outliers.}

%% file: conclusion_Section8_.tex
\section{Conclusion}\label{sec:conclusion}

We introduced the Sierpi\'nski--Knopp Wasserstein distance
\(d_{\mathrm{SK}}\), a diagonal-aware metric between normalized
persistence diagrams. The construction uses the first-hit selector of
the SK space-filling curve to replace the usual two-dimensional
partial-assignment problem by a one-dimensional \(1\)-Wasserstein
assignment problem on the unit interval. The resulting monotone
optimal assignment can be evaluated by sorting two lists of scalar values,
yielding an \(O(N\log N)\) algorithm for diagrams of total cardinality
\(N\). \Scorrection{By
retaining the diagram point associated with each scalar value, the
same monotone pairing also provides an explicit partial assignment
between the two input diagrams, with unmatched points assigned to the
diagonal. This correspondence can be reused by downstream applications
requiring point correspondences.}

The theoretical analysis establishes that \(d_{\mathrm{SK}}\)
controls the classical diagram \(2\)-Wasserstein distance:
\[
W_2(X,Y)
\leq
\sqrt{2}\,d_{\mathrm{SK}}(X,Y).
\]
Moreover, the cumulative \(L^1\) representation of
\(d_{\mathrm{SK}}^2\) yields an explicit isometric embedding of
\(d_{\mathrm{SK}}\) into a Hilbert space and a positive-definite
Gaussian kernel, this makes
\Scorrection{\(d_{\mathrm{SK}}\)} directly compatible with
Euclidean embedding methods and learning pipelines that require vector
representations, and
kernel-based learning methods on normalized persistence diagrams. We also introduced the assignment-induced planar
surrogate \(W_\Gamma\), which re-evaluates the monotone SK assignment
in the birth--death plane. Although \(W_\Gamma\) does not retain the
metric and Hilbertian guarantees of \(d_{\mathrm{SK}}\), it provides
a tighter numerical surrogate for \(W_2\) in median.

We evaluated both SK-based constructions on 12 scientific collections
containing 227 persistence diagrams and \(3.38\) million persistence
pairs, using complete pairwise numerical \(W_2\) reference matrices
computed with TTK's auction-based solver. The finite-level approximation stabilizes at refinement
level \(L=30\). Across \(3{,}004\) distinct pairwise diagram
comparisons for each dissimilarity, the median Spearman correlations with \(W_2\) are
\(0.879\) for \(d_{\mathrm{SK},30}\) and \(0.924\) for
\(W_{\Gamma,30}\), while the median NN@3 overlaps are \(0.886\) and
\(0.914\), respectively. No violation of
\[
W_2
\leq
W_{\Gamma,30}
\leq
\sqrt{2}\,d_{\mathrm{SK},30}
\]
is observed. The median per-collection speedup of
\(d_{\mathrm{SK},30}\) over \(W_2\) is \(626\times\), and the ratio
of total computation times over the full benchmark is approximately
\(2100\times\). Average-linkage partitions obtained from
\(d_{\mathrm{SK},30}\) and \(W_{\Gamma,30}\) each reproduce the
corresponding \(W_2\) partition exactly on 8 of the 12 collections.

The Hilbertian and kernel structures also lead to effective analysis
pipelines. Average linkage on \(d_{\mathrm{SK},30}\), \(k\)-means in
its full-dimensional Hilbertian realization, and spectral clustering
with its Gaussian kernel obtain mean ARI values of \(0.667\),
\(0.756\), and \(0.800\), respectively, with respect to the benchmark
reference partitions. These three pipelines achieve perfect agreement
with the reference partitions on 5, 6, and 7 of the 12 collections,
respectively. For comparison, average-linkage clustering on \(W_2\)
obtains a mean ARI of \(0.750\) and perfect agreement on 7
collections. The Gaussian \(d_{\mathrm{SK}}\) kernel does not decrease the ARI relative to average linkage on any
collection, and additionally
supports contiguous segmentation of ordered diagram collections,
with exact recovery of all reference boundaries on 5 of the 8
eligible collections.

Overall, \(d_{\mathrm{SK}}\) combines computational efficiency,
control of the \(W_2\) geometry, and compatibility with Euclidean and
kernel methods, whereas \(W_\Gamma\) favors closer numerical agreement
with \(W_2\). \Scorrection{Future work will investigate normalization protocols for online and
iterative settings, space-filling-curve-based surrogates for the
potentially more outlier-robust distances \(W_p\), \(1\leq p<2\),} investigate efficient computation and
reconstruction of \(d_{\mathrm{SK}}\)-barycenters, together with their
use as computationally efficient approximations of
\(W_2\)-barycenters.

%% file: appendix_skot_from_ieee.tex
\input{Supplementary_material}

\section{Proofs}\label{app:sec:skot}

\paragraph{Ambient triangle and persistence diagrams.}
Let
\[
\Atri:=\{(x,y)\in[0,1]^2:\ x\le y\},
\qquad
\Delta:=\{(u,u):u\in[0,1]\}\subset \Atri
\]
be the diagonal.
A (finite) persistence diagram supported in $\Atri$ is represented as an \emph{integer atomic measure}
\[
D=\sum_{r=1}^m a_r\,\delta_{z_r},
\qquad z_r\in \Atri\setminus\Delta,\ \ a_r\in\mathbb{N}.
\]
We write $|D|:=\sum_{r=1}^m a_r$ for the total mass (total number of points counting multiplicities).

\paragraph{Diagonal projection.}
Define the Euclidean orthogonal projection onto the diagonal by
\[
\Pi:\Atri\to\Delta,\qquad \Pi(x,y):=\Big(\frac{x+y}{2},\frac{x+y}{2}\Big).
\]

\paragraph{A space-filling SK curve and a canonical selector.}
Fix a continuous surjection (e.g.\ the standard Sierpi\'nski--Knopp curve on $\Atri$)
\[
S:[0,1]\to \Atri.
\]
For each $z\in\Atri$, define the \emph{first-hit selector}
\begin{equation}\label{app:eq:def_iota}
\iota(z):=\min S^{-1}(\{z\})\in[0,1].
\end{equation}

\begin{lemma}[Well-defined section and injectivity of $\iota$]\label{app:lem:iota_section}
The map $\iota:\Atri\to[0,1]$ in \eqref{app:eq:def_iota} is well-defined and satisfies
\[
S(\iota(z))=z\quad\text{for all }z\in\Atri,
\qquad\text{hence }\iota\text{ is injective.}
\]
Moreover, the two sets $\iota(\Delta)$ and $\iota(\Atri\setminus\Delta)$ are disjoint.
\end{lemma}

\begin{proof}
Fix $z\in\Atri$. Since $S$ is continuous and $[0,1]$ is compact, the fiber $S^{-1}(\{z\})$ is closed in $[0,1]$,
hence compact. As $S$ is surjective, this fiber is nonempty, so its minimum exists: $\iota(z)$ is well-defined.
By definition, $\iota(z)\in S^{-1}(\{z\})$, hence $S(\iota(z))=z$.

If $\iota(z)=\iota(z')$, then applying $S$ gives $z=S(\iota(z))=S(\iota(z'))=z'$, so $\iota$ is injective.
Finally, if $t\in \iota(\Delta)\cap \iota(\Atri\setminus\Delta)$, then $t=\iota(z)=\iota(z')$ for some
$z\in\Delta$ and $z'\in\Atri\setminus\Delta$, contradicting injectivity. Thus the two images are disjoint.
\end{proof}

\begin{lemma}[Borel measurability of the selector]\label{app:lem:iota_borel}
The selector $\iota:\Atri\to[0,1]$ defined by $\iota(z):=\min S^{-1}(\{z\})$ is Borel measurable.
More precisely, for every $t\in[0,1]$ one has
\[
\{z\in\Atri:\ \iota(z)\le t\}=S([0,t]).
\]
\end{lemma}

\begin{proof}
Let $t\in[0,1]$. If $\iota(z)\le t$, then there exists $s\in[0,t]$ such that $S(s)=z$, hence
$z\in S([0,t])$. Conversely, if $z\in S([0,t])$, then $S(s)=z$ for some $s\le t$, so
$\iota(z)=\min S^{-1}(\{z\})\le s\le t$. Thus $\{z:\iota(z)\le t\}=S([0,t])$.
Since $S$ is continuous and $[0,t]$ is compact, $S([0,t])$ is compact, hence Borel. Therefore $\iota$ is Borel.
\end{proof}

\paragraph{1D transport on \texorpdfstring{$[0,1]$}{[0,1]}.}
Let $\mathcal{M}_+([0,1])$ be the cone of finite nonnegative Borel measures on $[0,1]$.
For $\mu,\nu\in\mathcal{M}_+([0,1])$ with equal total mass $\mu([0,1])=\nu([0,1])$, define the 1-Wasserstein distance
(with ground cost $|s-t|$) by
\[
W_1(\mu,\nu)
:=\inf_{\pi\in\Pi(\mu,\nu)}\ \int_{[0,1]^2}|s-t|\,d\pi(s,t),
\]
where $\Pi(\mu,\nu)$ denotes the set of couplings of $\mu$ and $\nu$.

\paragraph{SK encodings of diagrams.}
For a diagram $D$ on $\Atri\setminus\Delta$, define two measures on $[0,1]$:
\[
\mu_D:=\iota_{\#}D\in\mathcal{M}_+([0,1]),
\qquad
\nu_D:=\iota_{\#}(\Pi_{\#}D)\in\mathcal{M}_+([0,1]).
\]
Note that $\mu_D([0,1])=\nu_D([0,1])=|D|$.
Define the associated \emph{signed} measure
\[
\sigma_D:=\mu_D-\nu_D,
\qquad\text{so that }\sigma_D([0,1])=0.
\]

\begin{definition}[$d_{\mathrm{SK}}$ distance]\label{app:def:SKOT}
For two diagrams $D_1,D_2$ on $\Atri\setminus\Delta$, define the \emph{$d_{\mathrm{SK}}$ distance} by
\begin{equation}\label{app:eq:def_SKOT}
d_{\mathrm{SK}}(D_1,D_2)^2
:=W_1\big(\mu_{D_1}+\nu_{D_2},\ \mu_{D_2}+\nu_{D_1}\big).
\end{equation}
Equivalently, in multiset notation, the left marginal consists of the times
$\{\iota(x):x\in D_1\}\cup\{\iota(\Pi(y)):y\in D_2\}$ (with multiplicities),
and the right marginal consists of $\{\iota(y):y\in D_2\}\cup\{\iota(\Pi(x)):x\in D_1\}$.
\end{definition}

\begin{lemma}\label{app:lem:diff_identity}
For any diagrams $D_1,D_2$ on $\Atri\setminus\Delta$,
\[
(\mu_{D_1}+\nu_{D_2})-(\mu_{D_2}+\nu_{D_1})=\sigma_{D_1}-\sigma_{D_2}.
\]
\end{lemma}

\begin{proof}
Expand and regroup:
\[
(\mu_{D_1}+\nu_{D_2})-(\mu_{D_2}+\nu_{D_1})
=(\mu_{D_1}-\nu_{D_1})-(\mu_{D_2}-\nu_{D_2})
=\sigma_{D_1}-\sigma_{D_2}.
\]
\end{proof}

\begin{proposition}[1D earthmover formula and monotone matching]\label{app:prop:1D_EMD}
Let $\mu,\nu\in\mathcal{M}_+([0,1])$ be \emph{integer atomic} measures with equal total mass:
\[
\mu=\sum_{i=1}^M \delta_{a_i},\qquad \nu=\sum_{i=1}^M \delta_{b_i}
\quad\text{(multisets, repetitions allowed).}
\]
Define the cumulative discrepancy
\[
H(t):=(\mu-\nu)([0,t])\qquad (t\in[0,1]).
\]
Then
\begin{equation}\label{app:eq:emd_cumulative}
W_1(\mu,\nu)=\int_0^1 |H(t)|\,dt.
\end{equation}
Moreover, if $(a_{(1)}\le\cdots\le a_{(M)})$ and $(b_{(1)}\le\cdots\le b_{(M)})$ are the sorted lists, then
\begin{equation}\label{app:eq:emd_sorting}
W_1(\mu,\nu)=\sum_{k=1}^M |a_{(k)}-b_{(k)}|.
\end{equation}
\end{proposition}

\begin{proof}
\emph{Step 1.}
Let $\pi\in\Pi(\mu,\nu)$ be any coupling.
For $t\in[0,1]$, set $L_t:=[0,t]$ and $R_t:=(t,1]$, and define
\[
A_t:=\pi(L_t\times R_t),\qquad B_t:=\pi(R_t\times L_t).
\]
Using the marginal constraints,
\[
\mu(L_t)=\pi(L_t\times[0,1])=\pi(L_t\times L_t)+A_t,
\qquad
\nu(L_t)=\pi([0,1]\times L_t)=\pi(L_t\times L_t)+B_t,
\]
hence $H(t)=\mu(L_t)-\nu(L_t)=A_t-B_t$ and therefore $|H(t)|\le A_t+B_t$.

Now define the ``crossing event'' at level $t$:
\[
C_t:=\big(L_t\times R_t\big)\ \cup\ \big(R_t\times L_t\big),
\quad\text{so that }\pi(C_t)=A_t+B_t.
\]
Thus $|H(t)|\le \pi(C_t)$ for all $t$.

Next, note the identity (valid for all $s,u\in[0,1]$):
\[
|s-u|=\int_0^1 \mathbf{1}_{C_t}(s,u)\,dt,
\]
because $\mathbf{1}_{C_t}(s,u)=1$ iff $t$ lies strictly between $s$ and $u$.
By Fubini,
\[
\int_{[0,1]^2}|s-u|\,d\pi(s,u)
=\int_0^1 \pi(C_t)\,dt
\ \ge\ \int_0^1 |H(t)|\,dt.
\]
Taking the infimum over $\pi\in\Pi(\mu,\nu)$ gives
\begin{equation}\label{app:eq:lower_bound_emd}
W_1(\mu,\nu)\ \ge\ \int_0^1 |H(t)|\,dt.
\end{equation}

\emph{Step 2.}
Let $(a_{(1)}\le\cdots\le a_{(M)})$ and $(b_{(1)}\le\cdots\le b_{(M)})$ be sorted.
Consider the monotone (quantile) coupling
\[
\pi^\star:=\sum_{k=1}^M \delta_{(a_{(k)},\,b_{(k)})}\ \in\ \Pi(\mu,\nu).
\]
Its transport cost is $\int |s-u|\,d\pi^\star=\sum_{k=1}^M |a_{(k)}-b_{(k)}|$, which proves \eqref{app:eq:emd_sorting}
once we show it equals the right-hand side of \eqref{app:eq:emd_cumulative}.

Fix $t\in[0,1]$ and set $c_a(t):=\#\{k:\ a_{(k)}\le t\}$ and $c_b(t):=\#\{k:\ b_{(k)}\le t\}$.
Then $H(t)=c_a(t)-c_b(t)$.
A pair $(a_{(k)},b_{(k)})$ crosses the cut at $t$ (i.e.\ belongs to $C_t$) iff exactly one of $a_{(k)},b_{(k)}$
is $\le t$. If $c_a(t)\ge c_b(t)$, then this happens precisely for indices $k\in\{c_b(t)+1,\dots,c_a(t)\}$,
so $\pi^\star(C_t)=c_a(t)-c_b(t)=|H(t)|$. The case $c_b(t)\ge c_a(t)$ is symmetric, hence
\[
\pi^\star(C_t)=|H(t)|\qquad\forall\,t\in[0,1].
\]
Therefore,
\[
\int_{[0,1]^2}|s-u|\,d\pi^\star(s,u)
=\int_0^1 \pi^\star(C_t)\,dt
=\int_0^1 |H(t)|\,dt.
\]
Combining with \eqref{app:eq:lower_bound_emd} yields \eqref{app:eq:emd_cumulative}, and \eqref{app:eq:emd_sorting} follows as well.
\end{proof}

\begin{theorem}[$d_{\mathrm{SK}}^2$ is a metric]\label{app:thm:SKOT_metric}
Fix the curve $S$ and the selector $\iota$ as above.
Then $d_{\mathrm{SK}}^2$ in \eqref{app:eq:def_SKOT} defines a metric on the set of finite diagrams supported in $\Atri\setminus\Delta$.
\end{theorem}

\begin{proof}
\emph{Nonnegativity and symmetry.}
By definition, $d_{\mathrm{SK}}(D_1,D_2)^2$ is a 1-Wasserstein distance on $[0,1]$, hence $d_{\mathrm{SK}}^2\ge 0$ and
$d_{\mathrm{SK}}(D_1,D_2)^2=d_{\mathrm{SK}}(D_2,D_1)^2$.

\emph{Identity of indiscernibles.}
Clearly $d_{\mathrm{SK}}(D,D)^2=W_1(\mu_D+\nu_D,\mu_D+\nu_D)=0$.
Conversely, assume $d_{\mathrm{SK}}(D_1,D_2)^2=0$. Then
\[
\mu_{D_1}+\nu_{D_2}=\mu_{D_2}+\nu_{D_1}.
\]
Rearranging gives $\sigma_{D_1}=\sigma_{D_2}$ (equivalently, use Lemma~\ref{app:lem:diff_identity}).
By Lemma~\ref{app:lem:iota_section}, $\iota(\Delta)$ and $\iota(\Atri\setminus\Delta)$ are disjoint; moreover
$\mu_{D_i}$ is supported in $\iota(\Atri\setminus\Delta)$ while $\nu_{D_i}$ is supported in $\iota(\Delta)$.
Restricting the identity $\sigma_{D_1}=\sigma_{D_2}$ to $\iota(\Atri\setminus\Delta)$ yields $\mu_{D_1}=\mu_{D_2}$.
Since $\iota$ is injective, equality of the pushforwards $\iota_\# D_1=\iota_\# D_2$ implies $D_1=D_2$
as integer atomic measures on $\Atri\setminus\Delta$.

\emph{Triangle inequality.}
Let $D_1,D_2,D_3$ be diagrams and set
\[
\alpha_{ij}:=\mu_{D_i}+\nu_{D_j},\qquad \beta_{ij}:=\mu_{D_j}+\nu_{D_i}.
\]
By Lemma~\ref{app:lem:diff_identity}, $\alpha_{ij}-\beta_{ij}=\sigma_{D_i}-\sigma_{D_j}$.
Since the measures $\alpha_{ij},\beta_{ij}$ are integer atomic on $[0,1]$ with equal total mass $|D_i|+|D_j|$,
Proposition~\ref{app:prop:1D_EMD} gives
\[
d_{\mathrm{SK}}(D_i,D_j)^2=W_1(\alpha_{ij},\beta_{ij})
=\int_0^1 \big|(\sigma_{D_i}-\sigma_{D_j})([0,t])\big|\,dt.
\]
Hence, using $(\sigma_{D_1}-\sigma_{D_3})=(\sigma_{D_1}-\sigma_{D_2})+(\sigma_{D_2}-\sigma_{D_3})$ and $|u+v|\le |u|+|v|$,
\[
\begin{aligned}
d_{\mathrm{SK}}(D_1,D_3)^2
&=\int_0^1 \big|(\sigma_{D_1}-\sigma_{D_3})([0,t])\big|\,dt\\
&\le \int_0^1 \big|(\sigma_{D_1}-\sigma_{D_2})([0,t])\big|\,dt\\
&\quad +\int_0^1 \big|(\sigma_{D_2}-\sigma_{D_3})([0,t])\big|\,dt\\
&=d_{\mathrm{SK}}(D_1,D_2)^2+d_{\mathrm{SK}}(D_2,D_3)^2.
\end{aligned}
\]
This proves the triangle inequality.
\end{proof}

\begin{theorem}
[\(d_{\mathrm{SK}}\) is a metric]
\label{app:cor:dSK_metric}
The function \(d_{\mathrm{SK}}\) defined in
\eqref{app:eq:def_SKOT} is a metric on the set of finite persistence
diagrams supported in \(\Atri\setminus\Delta\).
\end{theorem}

\begin{proof}
Nonnegativity, symmetry, and identity of indiscernibles follow
immediately from the corresponding properties of
\(d_{\mathrm{SK}}^2\) established in
Theorem~\ref{app:thm:SKOT_metric}. For any diagrams
\(D_1,D_2,D_3\), the triangle inequality for
\(d_{\mathrm{SK}}^2\) gives
\[
\begin{aligned}
d_{\mathrm{SK}}(D_1,D_3)
&=
\sqrt{d_{\mathrm{SK}}(D_1,D_3)^2}\\
&\leq
\sqrt{
d_{\mathrm{SK}}(D_1,D_2)^2
+
d_{\mathrm{SK}}(D_2,D_3)^2
}\\
&\leq
d_{\mathrm{SK}}(D_1,D_2)
+
d_{\mathrm{SK}}(D_2,D_3),
\end{aligned}
\]
where the last inequality follows from
\(\sqrt{a+b}\leq\sqrt a+\sqrt b\) for \(a,b\geq0\).
Thus, \(d_{\mathrm{SK}}\) is a metric.
\end{proof}

\begin{corollary}[Sorting formula and complexity for $d_{\mathrm{SK}}^2$]\label{app:cor:SKOT_complexity}
Let $D_1,D_2$ be finite persistance diagrams on $\Atri\setminus\Delta$ and set $N:=|D_1|+|D_2|$.
Form the two multisets of times (with multiplicities)
\[
A:=\{\iota(x):x\in D_1\}\ \cup\ \{\iota(\Pi(y)):\ y\in D_2\},
\qquad
B:=\{\iota(y):y\in D_2\}\ \cup\ \{\iota(\Pi(x)):\ x\in D_1\}.
\]
Let $(a_{(1)}\le\cdots\le a_{(N)})$ and $(b_{(1)}\le\cdots\le b_{(N)})$ be the sorted lists of $A$ and $B$.
Then
\[
d_{\mathrm{SK}}(D_1,D_2)^2=\sum_{k=1}^N |a_{(k)}-b_{(k)}|.
\]
Consequently, \emph{given the values of $\iota$ on the relevant points},
$d_{\mathrm{SK}}(D_1,D_2)^2$ can be computed in $O(N\log N)$ time by sorting and a linear pass.
\end{corollary}

\begin{proof}
This is exactly Proposition~\ref{app:prop:1D_EMD} applied to the two measures
$\alpha_{12}=\mu_{D_1}+\nu_{D_2}=\sum_{a\in A}\delta_a$ and
$\beta_{12}=\mu_{D_2}+\nu_{D_1}=\sum_{b\in B}\delta_b$ appearing in Definition~\ref{app:def:SKOT}.
The complexity follows because sorting two length-$N$ lists costs $O(N\log N)$ and the final sum costs $O(N)$.
\end{proof}

\begin{remark}[Dependence on the chosen space-filling curve]
The distance $d_{\mathrm{SK}}$ depends on the fixed choice of the surjection $S$ and the selector $\iota$.
For any fixed choice, Theorem~\ref{app:cor:dSK_metric} guarantees that $d_{\mathrm{SK}}$ is a metric.
\end{remark}

\begin{remark}[Diagonal points]
If one allows diagram points \emph{on} the diagonal $\Delta$, then they contribute
$\delta_{\iota(z)}-\delta_{\iota(\Pi(z))}=0$ to $\sigma_D$, so $d_{\mathrm{SK}}$ becomes a pseudo-metric unless one
identifies diagrams up to diagonal points. This is consistent with the usual convention in persistence-diagram theory.
\end{remark}

\paragraph{Diagonal-aware quadratic cost.}
Let $\Delta:=\{(u,u):u\in[0,1]\}$ and $d_\Delta(z):=\inf_{u\in[0,1]}\|z-(u,u)\|_2$.
Define the diagonal-aware squared cost $c_\Delta:\Atri\times\Atri\to[0,\infty)$ by
\[
c_\Delta(x,y):=
\begin{cases}
\|x-y\|_2^2, & x\notin\Delta,\ y\notin\Delta,\\[2pt]
d_\Delta(x)^2, & x\notin\Delta,\ y\in\Delta,\\[2pt]
d_\Delta(y)^2, & x\in\Delta,\ y\notin\Delta,\\[2pt]
0, & x\in\Delta,\ y\in\Delta.
\end{cases}
\]
For two persistance diagrams $D_1,D_2$ (integer atomic measures on $\Atri\setminus\Delta$), set
\[
\bar D_1:=D_1+\Pi_{\#}D_2,
\qquad
\bar D_2:=D_2+\Pi_{\#}D_1,
\]
so that $|\bar D_1|=|\bar D_2|=|D_1|+|D_2|$, and define
\[
W_{2,\Delta}^2(D_1,D_2)
:=\inf_{\gamma\in\Pi(\bar D_1,\bar D_2)}\ \int_{\Atri\times\Atri} c_\Delta(x,y)\,d\gamma(x,y).
\]

\begin{remark}
[SK $1/2$-H\"older]\label{app:ass:SK_holder}
The (fixed) SK curve $S:[0,1]\to\Atri$ satisfies
\[
\|S(t)-S(s)\|_2^2 \le C_{\mathrm{SK}}\,|t-s|
\qquad\forall\,t,s\in[0,1].
\]
\end{remark}

\begin{theorem}[Quadratic Wasserstein with diagonal is controlled by $d_{\mathrm{SK}}^2$]\label{app:thm:W2_le_SKOT}
Under Assumption~\ref{app:ass:SK_holder}, for any two diagrams $D_1,D_2$ on $\Atri\setminus\Delta$,
\[
W_{2,\Delta}^2(D_1,D_2)\ \le\ C_{\mathrm{SK}}\ d_{\mathrm{SK}}(D_1,D_2)^2.
\]
In particular, if $C_{\mathrm{SK}}=2$, then $W_{2,\Delta}^2(D_1,D_2)\le 2\,d_{\mathrm{SK}}(D_1,D_2)^2$.
\end{theorem}

\begin{proof}
Recall that $d_{\mathrm{SK}}(D_1,D_2)^2=W_1(\iota_{\#}\bar D_1,\iota_{\#}\bar D_2)$ with ground cost $|t-s|$ on $[0,1]$.
Let $\pi^\star\in\Pi(\iota_{\#}\bar D_1,\iota_{\#}\bar D_2)$ be an optimal coupling, so that
\[
d_{\mathrm{SK}}(D_1,D_2)^2=\int_{[0,1]^2}|t-s|\,d\pi^\star(t,s).
\]
Define the pushforward coupling on $\Atri\times\Atri$
\[
\gamma:=(S,S)_{\#}\pi^\star.
\]
Since $S\circ\iota=\mathrm{id}_{\Atri}$, we have $S_{\#}(\iota_{\#}\bar D_i)=\bar D_i$ for $i=1,2$,
hence $\gamma\in\Pi(\bar D_1,\bar D_2)$.

Now fix $(t,s)\in[0,1]^2$ and set $x:=S(t)$, $y:=S(s)$.
If $x,y\notin\Delta$, then $c_\Delta(x,y)=\|x-y\|_2^2$.
If (say) $y\in\Delta$, then $c_\Delta(x,y)=d_\Delta(x)^2\le \|x-y\|_2^2$ because $y$ is a diagonal point.
If $x,y\in\Delta$, then $c_\Delta(x,y)=0\le \|x-y\|_2^2$.
Thus in all cases,
\[
c_\Delta(S(t),S(s))\ \le\ \|S(t)-S(s)\|_2^2\ \le\ C_{\mathrm{SK}}\,|t-s|
\]
by Assumption~\ref{app:ass:SK_holder}.
Integrating against $\pi^\star$ gives
\[
\int c_\Delta(x,y)\,d\gamma(x,y)
=\int c_\Delta(S(t),S(s))\,d\pi^\star(t,s)
\le C_{\mathrm{SK}}\int |t-s|\,d\pi^\star(t,s)
= C_{\mathrm{SK}}\,d_{\mathrm{SK}}(D_1,D_2)^2.
\]
Taking the infimum over all $\gamma\in\Pi(\bar D_1,\bar D_2)$ yields
$W_{2,\Delta}^2(D_1,D_2)\le C_{\mathrm{SK}}\,d_{\mathrm{SK}}(D_1,D_2)^2$.
\end{proof}

\subsubsection{Kernelization: conditional negative definiteness of \texorpdfstring{$d_{\mathrm{SK}}^2$}{dSK squared}}

\begin{definition}[Conditionally negative definite function (CND)]
Let $X$ be a set. A symmetric function $f:X\times X\to\mathbb{R}$ with $f(x,x)=0$ for all $x\in X$
is called \emph{conditionally negative definite} if for every $n\ge 1$, every $x_1,\dots,x_n\in X$
and every $c_1,\dots,c_n\in\mathbb{R}$ such that $\sum_{i=1}^n c_i=0$, one has
\[
\sum_{i=1}^n\sum_{j=1}^n c_i c_j\, f(x_i,x_j)\le 0.
\]
\end{definition}

\begin{lemma}[The absolute value is CND on $\mathbb{R}$]\label{app:lem:abs_CND}
For all $x_1,\dots,x_n\in\mathbb{R}$ and all $c_1,\dots,c_n\in\mathbb{R}$ with $\sum_{i=1}^n c_i=0$,
\[
\sum_{i=1}^n\sum_{j=1}^n c_i c_j\,|x_i-x_j|\le 0.
\]
\end{lemma}

\begin{proof}
For $u\in\mathbb{R}$ set $b_i(u):=\mathbf{1}_{(u,\infty)}(x_i)\in\{0,1\}$. For any $x,y\in\mathbb{R}$,
\[
|x-y|=\int_{\mathbb{R}}\big|\mathbf{1}_{(u,\infty)}(x)-\mathbf{1}_{(u,\infty)}(y)\big|\,du,
\]
because the integrand equals $1$ exactly when $u$ lies strictly between $x$ and $y$ (a set of Lebesgue measure $|x-y|$),
and equals $0$ otherwise. Hence, exchanging a finite sum with the integral,
\[
\sum_{i,j} c_i c_j\,|x_i-x_j|
=\int_{\mathbb{R}} \sum_{i,j} c_i c_j\,|b_i(u)-b_j(u)|\,du.
\]
Since $b_i(u)\in\{0,1\}$, we have $|b_i(u)-b_j(u)|=(b_i(u)-b_j(u))^2$. Expanding and using $\sum_i c_i=0$ yields,
for every $u\in\mathbb{R}$,
\[
\sum_{i,j} c_i c_j\,(b_i(u)-b_j(u))^2
=-2\Big(\sum_{i=1}^n c_i\, b_i(u)\Big)^2\le 0.
\]
Integrating over $u$ proves the claim.
\end{proof}

\begin{proposition}[$d_{\mathrm{SK}}^2$ is conditionally negative definite]\label{app:prop:SKOT_CND}
The function $(D,E)\mapsto d_{\mathrm{SK}}(D,E)^2$ is conditionally negative definite
on the set of finite diagrams supported in $\Atri\setminus\Delta$.
\end{proposition}

\begin{proof}
For a diagram $D$, define the cumulative signature
\[
H_D(t):=\sigma_D([0,t])\qquad (t\in[0,1]),
\]
where $\sigma_D=\mu_D-\nu_D$. By Lemma~\ref{app:lem:diff_identity} and Proposition~\ref{app:prop:1D_EMD}, for any $D,E$,
\[
d_{\mathrm{SK}}(D,E)^2
=W_1(\mu_D+\nu_E,\mu_E+\nu_D)
=\int_0^1 \big|(\sigma_D-\sigma_E)([0,t])\big|\,dt
=\int_0^1 |H_D(t)-H_E(t)|\,dt.
\]
Moreover $H_D$ is a bounded step function, and $|H_D(t)|\le |D|$ for all $t\in[0,1]$; in particular the integrand
$|H_D(t)-H_E(t)|$ is integrable and we may exchange finite sums and the integral.

Let $D_1,\dots,D_n$ be diagrams and let $c_1,\dots,c_n\in\mathbb{R}$ with $\sum_i c_i=0$.
Using Lemma~\ref{app:lem:abs_CND} pointwise (with $x_i=H_{D_i}(t)$) and integrating over $t\in[0,1]$ gives
\[
\sum_{i,j} c_i c_j\,d_{\mathrm{SK}}(D_i,D_j)^2
=\int_0^1 \sum_{i,j} c_i c_j\,|H_{D_i}(t)-H_{D_j}(t)|\,dt
\le 0.
\]
Thus $d_{\mathrm{SK}}^2$ is CND.
\end{proof}

\begin{corollary}[A positive definite Gaussian $d_{\mathrm{SK}}$ kernel]\label{app:cor:SKOT_kernel_PD}
For every $\sigma>0$, the function
\[
k_\sigma(D,E):=\exp\!\Big(-\frac{d_{\mathrm{SK}}(D,E)^2}{2\sigma^2}\Big)
\]
is a positive definite kernel on the set of finite diagrams supported in $\Atri\setminus\Delta$.
\end{corollary}

\begin{proof}
By Proposition~\ref{app:prop:SKOT_CND}, $d_{\mathrm{SK}}^2$ is conditionally negative definite. By Schoenberg's theorem,
for every $\lambda>0$ the kernel $(D,E)\mapsto \exp(-\lambda\,d_{\mathrm{SK}}(D,E)^2)$ is positive definite.
Taking $\lambda=\frac{1}{2\sigma^2}$ yields the claim.
\end{proof}

\begin{remark}[Computation]
Evaluating $k_\sigma(D_1,D_2)$ requires one computation of $d_{\mathrm{SK}}(D_1,D_2)^2$ followed by one exponential.
Hence, given the values of $\iota$ on the relevant points, the cost is
$O(N\log N)$ with $N=|D_1|+|D_2|$ (see Corollary~\ref{app:cor:SKOT_complexity}).
\end{remark}

\begin{proposition}[Isometric $L^1$ representation of $d_{\mathrm{SK}}^2$]\label{app:prop:SKOT_L1}
For a diagram $D$, recall the signed measure $\sigma_D=\mu_D-\nu_D$ on $[0,1]$ and define its cumulative
\[
H_D(t):=\sigma_D([0,t])\qquad (t\in[0,1]).
\]
Then for any persistence diagrams $D,E$ supported in $\Atri\setminus\Delta$,
\begin{equation}\label{app:eq:SKOT_L1}
d_{\mathrm{SK}}(D,E)^2=\int_0^1 |H_D(t)-H_E(t)|\,dt=\|H_D-H_E\|_{L^1([0,1])}.
\end{equation}
In particular, the map $D\mapsto H_D$ is an injective embedding of diagrams into $L^1([0,1])$,
and $d_{\mathrm{SK}}^2$ is the pullback of the $L^1$ distance through this embedding.
\end{proposition}

\begin{proof}
By Definition~\ref{app:def:SKOT} we have
\[
d_{\mathrm{SK}}(D,E)^2=W_1(\mu_D+\nu_E,\ \mu_E+\nu_D).
\]
By Proposition~\ref{app:prop:1D_EMD},
\[
W_1(\mu_D+\nu_E,\mu_E+\nu_D)=\int_0^1 \big|(\mu_D+\nu_E-\mu_E-\nu_D)([0,t])\big|\,dt.
\]
Using Lemma~\ref{app:lem:diff_identity}, $(\mu_D+\nu_E)-(\mu_E+\nu_D)=\sigma_D-\sigma_E$.
Hence the integrand equals $|H_D(t)-H_E(t)|$, proving \eqref{app:eq:SKOT_L1}.

Injectivity of $D\mapsto H_D$ follows from injectivity of $D\mapsto \sigma_D$,
which is immediate since $\mu_D$ and $\nu_D$ live on the disjoint sets
$\iota(\Atri\setminus\Delta)$ and $\iota(\Delta)$ (Lemma~\ref{app:lem:iota_section}).
\end{proof}

\begin{corollary}[Algebraic properties and trivial bounds]\label{app:cor:SKOT_algebra}
Let $D,E,F$ be diagrams on $\Atri\setminus\Delta$ and let $m\in\mathbb{N}$.
Then:
\begin{enumerate}
\item[(i)] (\emph{Common-addition invariance}) $d_{\mathrm{SK}}(D+F,E+F)^2=d_{\mathrm{SK}}(D,E)^2$.
\item[(ii)] (\emph{Integer homogeneity}) $d_{\mathrm{SK}}(mD,mE)^2=m\,d_{\mathrm{SK}}(D,E)^2$.
\item[(iii)] (\emph{Uniform upper bound}) $d_{\mathrm{SK}}(D,E)^2\le |D|+|E|$.
\end{enumerate}
\end{corollary}

\begin{proof}
All three items follow from the $L^1$ representation \eqref{app:eq:SKOT_L1}.

(i) Linearity of pushforwards gives $\sigma_{D+F}=\sigma_D+\sigma_F$, hence $H_{D+F}=H_D+H_F$ and
$\|H_{D+F}-H_{E+F}\|_{L^1}=\|H_D-H_E\|_{L^1}$.

(ii) Similarly, $\sigma_{mD}=m\sigma_D$ so $H_{mD}=mH_D$ and the $L^1$ distance scales by $m$.

(iii) Since $|H_D(t)|\le |D|$ and $|H_E(t)|\le |E|$ for all $t\in[0,1]$, we have
$|H_D(t)-H_E(t)|\le |D|+|E|$, hence $d_{\mathrm{SK}}(D,E)^2\le |D|+|E|$.
\end{proof}

\begin{proposition}[Dual (Kantorovich--Rubinstein) formulation]\label{app:prop:SKOT_dual}
For any diagrams $D,E$ on $\Atri\setminus\Delta$,
\[
d_{\mathrm{SK}}(D,E)^2
=\sup\Big\{\int_0^1 f(t)\,d(\sigma_D-\sigma_E)(t)\ :\ f:[0,1]\to\mathbb{R}\ \text{$1$-Lipschitz}\Big\}.
\]
Equivalently, $d_{\mathrm{SK}}(D,E)^2$ is the Kantorovich--Rubinstein norm
of the signed measure $\sigma_D-\sigma_E$ (which has total mass $0$).
\end{proposition}

\begin{proof}
This is the Kantorovich--Rubinstein duality for $W_1$ on the compact metric space $([0,1],|\cdot|)$,
applied to the pair of measures $(\mu_D+\nu_E,\mu_E+\nu_D)$; their difference is $\sigma_D-\sigma_E$
(Lemma~\ref{app:lem:diff_identity}).
\end{proof}

\begin{proposition}[An explicit Hilbert embedding]\label{app:prop:SKOT_Hilbert_embed}
Define, for any diagram $D$, the function $\Phi(D):[0,1]\times\mathbb{R}\to\mathbb{R}$ by
\[
\Phi(D)(t,u):=\mathbf{1}_{\{u<H_D(t)\}}-\mathbf{1}_{\{u<0\}}.
\]
Then $\Phi(D)\in L^2([0,1]\times\mathbb{R})$, and for any diagrams $D,E$,
\begin{equation}\label{app:eq:SKOT_Hilbert}
d_{\mathrm{SK}}(D,E)^2=\|\Phi(D)-\Phi(E)\|_{L^2([0,1]\times\mathbb{R})}^2.
\end{equation}
In particular, $d_{\mathrm{SK}}$ is a Hilbertian metric.
\end{proposition}

\begin{proof}
For fixed $t$, the function $u\mapsto \Phi(D)(t,u)$ is supported on the interval between $0$ and $H_D(t)$,
hence $\int_{\mathbb{R}}|\Phi(D)(t,u)|^2\,du=|H_D(t)|$ and $\Phi(D)\in L^2$ because $H_D$ is bounded.

Moreover, for any reals $a,b$,
\[
\int_{\mathbb{R}}\Big(\mathbf{1}_{\{u<a\}}-\mathbf{1}_{\{u<b\}}\Big)^2\,du = |a-b|,
\]
since the integrand equals $1$ exactly when $u$ lies strictly between $a$ and $b$.
Applying this with $a=H_D(t)$ and $b=H_E(t)$ and integrating over $t\in[0,1]$ yields
\[
\|\Phi(D)-\Phi(E)\|_{L^2}^2
=\int_0^1 |H_D(t)-H_E(t)|\,dt
=d_{\mathrm{SK}}(D,E)^2
\]
by Proposition~\ref{app:prop:SKOT_L1}.
\end{proof}

\begin{remark}[Kernel distance is a monotone transform of $d_{\mathrm{SK}}^2$]\label{app:rem:SKOT_kernel_distance}
For $\sigma>0$, the kernel $k_\sigma(D,E)=\exp\!\big(-d_{\mathrm{SK}}(D,E)^2/(2\sigma^2)\big)$ satisfies
$k_\sigma(D,D)=1$ for all $D$, and $k_\sigma(D,E)=1$ if and only if $D=E$
(since $d_{\mathrm{SK}}^2$ is a metric).
The induced RKHS distance is therefore
\[
d_{k_\sigma}(D,E)^2
:=\|k_\sigma(D,\cdot)-k_\sigma(E,\cdot)\|_{\mathcal{H}}^2
=k_\sigma(D,D)+k_\sigma(E,E)-2k_\sigma(D,E)
=2-2\exp\!\Big(-\frac{d_{\mathrm{SK}}(D,E)^2}{2\sigma^2}\Big).
\]
Equivalently, with $g(r):=\sqrt{2-2e^{-r/(2\sigma^2)}}$, one has
$d_{k_\sigma}(D,E)=g(d_{\mathrm{SK}}(D,E)^2)$.
The map $g$ is continuous and strictly increasing on $[0,\infty)$, with inverse
$g^{-1}(s)=-2\sigma^2\log(1-s^2/2)$ for $s\in[0,\sqrt{2})$; hence $d_{k_\sigma}$ and $d_{\mathrm{SK}}^2$
induce the same topology, which is also the topology induced by $d_{\mathrm{SK}}$.
\end{remark}

\begin{proposition}[Point-separating property of the Gaussian $d_{\mathrm{SK}}$ kernel]
\label{app:prop:SKOT_kernel_point_separating}
Fix $\sigma>0$ and define $k_\sigma(D,E):=\exp\!\big(-d_{\mathrm{SK}}(D,E)^2/(2\sigma^2)\big)$.
Then $k_\sigma(D,E)=1$ if and only if $D=E$.
In particular, the canonical feature map $D\mapsto k_\sigma(D,\cdot)$ is injective.
\end{proposition}

\begin{proof}
Since $d_{\mathrm{SK}}^2$ is a metric (Theorem~\ref{app:thm:SKOT_metric}), we have
$d_{\mathrm{SK}}(D,E)^2=0$ if and only if $D=E$.
As $r\mapsto \exp(-r/(2\sigma^2))$ is strictly decreasing on $[0,\infty)$,
$k_\sigma(D,E)=1$ iff $d_{\mathrm{SK}}(D,E)^2=0$ iff $D=E$.

If $k_\sigma(D,\cdot)=k_\sigma(E,\cdot)$, evaluating at $D$ yields
$1=k_\sigma(D,D)=k_\sigma(E,D)$, hence $E=D$ by the first part.
\end{proof}

\begin{theorem}[Hilbert--Gaussian form, ISPD, and characteristicness]
\label{app:thm:SKOT_kernel_characteristic}
Let $\sigma>0$ and let $k_\sigma(D,E)=\exp\!\big(-d_{\mathrm{SK}}(D,E)^2/(2\sigma^2)\big)$ on
$\mathcal{D}$ (finite diagrams in $\Atri\setminus\Delta$).
Then there exist a real Hilbert space $\mathcal{H}_0$ and a continuous injective map
$\Psi:\mathcal{D}\to\mathcal{H}_0$ such that
\[
k_\sigma(D,E)=\exp\!\Big(-\frac{\|\Psi(D)-\Psi(E)\|_{\mathcal{H}_0}^2}{2\sigma^2}\Big)
\qquad\forall\,D,E\in\mathcal{D}.
\]
Moreover, $k_\sigma$ is integrally strictly positive definite (ISPD) on $\mathcal{D}$ and therefore
characteristic.
\end{theorem}

\begin{proof}
By Proposition~\ref{app:prop:SKOT_Hilbert_embed}, there exists an explicit map
$\Phi:\mathcal{D}\to L^2([0,1]\times\mathbb{R})$ such that
$d_{\mathrm{SK}}(D,E)^2=\|\Phi(D)-\Phi(E)\|_{L^2}^2$ for all $D,E\in\mathcal{D}$.
Set $\mathcal{H}_0:=L^2([0,1]\times\mathbb{R})$ and $\Psi:=\Phi$.
Then, with the Gaussian kernel on $\mathcal{H}_0$,
\[
\tilde{k}_\sigma(x,y):=\exp\!\Big(-\frac{\|x-y\|_{\mathcal{H}_0}^2}{2\sigma^2}\Big),
\qquad x,y\in\mathcal{H}_0,
\]
we have $k_\sigma=\tilde{k}_\sigma\circ(\Psi,\Psi)$.

By \cite[Thm.~3.2]{Guella2020GaussianHilbert}, $\tilde{k}_\sigma$ is ISPD on $\mathcal{H}_0$.
Let $\lambda$ be a nonzero finite signed Borel measure on $\mathcal{D}$ and set
$\tilde{\lambda}:=\Psi_{\#}\lambda$. Since $\Psi$ is an isometric embedding, $\tilde{\lambda}\neq 0$ and
\[
\iint_{\mathcal{D}\times\mathcal{D}} k_\sigma(D,E)\,d\lambda(D)\,d\lambda(E)
=
\iint_{\mathcal{H}_0\times\mathcal{H}_0} \tilde{k}_\sigma(x,y)\,d\tilde{\lambda}(x)\,d\tilde{\lambda}(y)
>0,
\]
so $k_\sigma$ is ISPD on $\mathcal{D}$.

Finally, integrally strictly positive definite kernels are characteristic
(see, e.g., \cite[Thm.~7]{Sriperumbudur2010HilbertEmbedding}),
hence $k_\sigma$ is characteristic on $\mathcal{D}$.
\end{proof}

\begin{corollary}[Universality on compacts]
\label{app:cor:SKOT_kernel_universal}
Let $K\subset\mathcal{D}$ be such that $\Psi(K)$ is compact in $\mathcal{H}_0$
(equivalently, $K$ is compact for the metric $d_{\mathrm{SK}}$).
Then the restriction of $k_\sigma$ to $K\times K$ is universal on $K$:
its RKHS is dense in $C(K)$ with respect to the uniform norm.
Consequently, $k_\sigma$ is characteristic on $K$.
\end{corollary}

\begin{proof}
Let $\tilde{k}_\sigma$ be the Gaussian kernel on $\mathcal{H}_0$ as in
Theorem~\ref{app:thm:SKOT_kernel_characteristic}.
By \cite[Thm.~3.1]{Guella2020GaussianHilbert}, $\tilde{k}_\sigma$ is universal on $\mathcal{H}_0$,
hence its restriction to the compact set $\Psi(K)$ is universal on $\Psi(K)$.
Since $\Psi:K\to\Psi(K)$ is continuous and injective, the pullback kernel
$k_\sigma=\tilde{k}_\sigma\circ(\Psi,\Psi)$ is universal on $K$
(see, e.g., \cite[Lem.~7.4]{Guella2020GaussianHilbert}).

Characteristicness on $K$ follows either from Theorem~\ref{app:thm:SKOT_kernel_characteristic}
(restriction of a characteristic kernel remains characteristic on measures supported on $K$),
or directly from ISPD $\Rightarrow$ characteristic \cite[Thm.~7]{Sriperumbudur2010HilbertEmbedding}.
\end{proof}

\begin{remark}[Practical implication]
Theorems~\ref{app:thm:SKOT_kernel_characteristic}--\ref{app:cor:SKOT_kernel_universal} justify using $k_\sigma$
for kernel machines (SVM, kernel ridge, GP) and for distributional tasks (MMD two-sample tests)
on persistence diagrams endowed with $d_{\mathrm{SK}}$.
\end{remark}

\begin{proposition}[Robustness to an approximate selector]\label{app:prop:SKOT_selector_approx}
Let $\tilde\iota:\Atri\to[0,1]$ be any map and define $\widetilde d_{\mathrm{SK}}$ by repeating
Definition~\ref{app:def:SKOT} with $\iota$ replaced by $\tilde\iota$.
Assume that for all atoms $z$ appearing in $D$ or $E$ and for their diagonal projections,
\[
|\tilde\iota(z)-\iota(z)|\le \varepsilon,\qquad
|\tilde\iota(\Pi(z))-\iota(\Pi(z))|\le \varepsilon.
\]
Then
\[
\big|\widetilde d_{\mathrm{SK}}(D,E)^2-d_{\mathrm{SK}}(D,E)^2\big|
\le 2\,(|D|+|E|)\,\varepsilon.
\]
\end{proposition}

\begin{proof}
Let $\alpha:=\mu_D+\nu_E$ and $\beta:=\mu_E+\nu_D$, and let $\tilde\alpha,\tilde\beta$ be the corresponding measures
built with $\tilde\iota$. Since $W_1$ is a metric,
\[
|W_1(\alpha,\beta)-W_1(\tilde\alpha,\tilde\beta)|
\le W_1(\alpha,\tilde\alpha)+W_1(\beta,\tilde\beta).
\]
Under the assumption, each atom of $\alpha$ is moved by at most $\varepsilon$ to obtain $\tilde\alpha$,
so coupling each atom with its moved version yields $W_1(\alpha,\tilde\alpha)\le (|D|+|E|)\varepsilon$.
The same bound holds for $W_1(\beta,\tilde\beta)$, giving the claim.
\end{proof}

\begin{proposition}[Dataset-specific reverse bound for $d_{\mathrm{SK}}^2$]\label{app:prop:SKOT_reverse_finite_dataset}
Let $\mathcal F$ be a finite family of finite persistence diagrams supported in
$\Atri\setminus\Delta$, and define the finite set of relevant atoms
\[
\mathcal Z_{\mathcal F}
:=
\bigcup_{D\in\mathcal F}
\Bigl(\operatorname{spt}(D)\cup \Pi(\operatorname{spt}(D))\Bigr)
\subset \Atri.
\]
Since $\mathcal Z_{\mathcal F}$ is finite and $\iota$ is injective, the constant
\[
L_{\mathcal F}
:=
\max_{\substack{x,y\in\mathcal Z_{\mathcal F}\\ x\neq y}}
\frac{|\iota(x)-\iota(y)|}{\|x-y\|_2}
\]
is finite. Then, for every $D,E\in\mathcal F$,
\[
d_{\mathrm{SK}}(D,E)^2
\le
L_{\mathcal F}\,\sqrt{2(|D|+|E|)}\;W_{2,\Delta}(D,E).
\]
In particular, if all diagrams in $\mathcal F$ satisfy $|D|\le M$, then
\[
d_{\mathrm{SK}}(D,E)^2
\le
2L_{\mathcal F}\sqrt{M}\;W_{2,\Delta}(D,E)
\qquad\forall\,D,E\in\mathcal F.
\]
\end{proposition}

\begin{proof}
Fix $D,E\in\mathcal F$, and let $\mathcal M$ be an optimal diagonal-aware matching
between $D$ and $E$ for $W_{2,\Delta}$, with $U\subset D$ the unmatched points
of $D$ and $V\subset E$ the unmatched points of $E$. Then
\[
W_{2,\Delta}^2(D,E)
=
\sum_{(x,y)\in \mathcal M}\|x-y\|_2^2
+
\sum_{x\in U} d_\Delta(x)^2
+
\sum_{y\in V} d_\Delta(y)^2.
\]

We construct a coupling between
\[
\alpha:=\mu_D+\nu_E,
\qquad
\beta:=\mu_E+\nu_D
\]
as follows:
for each matched pair $(x,y)\in \mathcal M$, couple $\iota(x)$ with $\iota(y)$ and
$\iota(\Pi(y))$ with $\iota(\Pi(x))$; for each unmatched $x\in U$, couple
$\iota(x)$ with $\iota(\Pi(x))$; for each unmatched $y\in V$, couple
$\iota(\Pi(y))$ with $\iota(y)$. This gives an admissible coupling for
$W_1(\alpha,\beta)=d_{\mathrm{SK}}(D,E)^2$, hence
\[
d_{\mathrm{SK}}(D,E)^2
\le
\sum_{(x,y)\in \mathcal M}
\Bigl(
|\iota(x)-\iota(y)|
+
|\iota(\Pi(y))-\iota(\Pi(x))|
\Bigr)
+
\sum_{x\in U}|\iota(x)-\iota(\Pi(x))|
+
\sum_{y\in V}|\iota(y)-\iota(\Pi(y))|.
\]

All points involved belong to $\mathcal Z_{\mathcal F}$, so by definition of
$L_{\mathcal F}$,
\[
|\iota(x)-\iota(y)|\le L_{\mathcal F}\|x-y\|_2,
\qquad
|\iota(\Pi(y))-\iota(\Pi(x))|
\le
L_{\mathcal F}\|\Pi(y)-\Pi(x)\|_2
\le
L_{\mathcal F}\|x-y\|_2,
\]
because $\Pi$ is $1$-Lipschitz. Likewise,
\[
|\iota(x)-\iota(\Pi(x))|\le L_{\mathcal F}d_\Delta(x),
\qquad
|\iota(y)-\iota(\Pi(y))|\le L_{\mathcal F}d_\Delta(y).
\]
Therefore
\[
d_{\mathrm{SK}}(D,E)^2
\le
L_{\mathcal F}
\Bigl(
2\sum_{(x,y)\in \mathcal M}\|x-y\|_2
+
\sum_{x\in U} d_\Delta(x)
+
\sum_{y\in V} d_\Delta(y)
\Bigr).
\]

Now write the right-hand side as a sum of exactly $|D|+|E|$ nonnegative terms:
for each $(x,y)\in \mathcal M$, two copies of $\|x-y\|_2$, and for each $x\in U$,
resp.\ $y\in V$, one copy of $d_\Delta(x)$, resp.\ $d_\Delta(y)$. By
Cauchy--Schwarz,
\[
2\sum_{(x,y)\in \mathcal M}\|x-y\|_2
+
\sum_{x\in U} d_\Delta(x)
+
\sum_{y\in V} d_\Delta(y)
\le
\sqrt{|D|+|E|}
\Bigl(
2\sum_{(x,y)\in \mathcal M}\|x-y\|_2^2
+
\sum_{x\in U} d_\Delta(x)^2
+
\sum_{y\in V} d_\Delta(y)^2
\Bigr)^{1/2}.
\]
Since
\[
2\sum_{(x,y)\in \mathcal M}\|x-y\|_2^2
+
\sum_{x\in U} d_\Delta(x)^2
+
\sum_{y\in V} d_\Delta(y)^2
\le
2\,W_{2,\Delta}^2(D,E),
\]
we obtain
\[
d_{\mathrm{SK}}(D,E)^2
\le
L_{\mathcal F}\sqrt{2(|D|+|E|)}\,W_{2,\Delta}(D,E).
\]
If $|D|,|E|\le M$, then $|D|+|E|\le 2M$, hence
\[
d_{\mathrm{SK}}(D,E)^2\le 2L_{\mathcal F}\sqrt{M}\,W_{2,\Delta}(D,E).
\]
\end{proof}

\begin{proposition}[A simpler separation-based reverse bound]
\label{app:prop:SKOT_reverse_simple}
Let $D,E$ be finite diagrams on $\Atri\setminus\Delta$. Define
\[
\delta^{\times}_{D,E}
:=
\min\{\|x-y\|_2:\ x\in \operatorname{spt}(D),\ y\in \operatorname{spt}(E),\ x\neq y\},
\]
with the convention $\delta^{\times}_{D,E}=+\infty$ if no such distinct pair exists, and
\[
\eta_{D,E}:=
\min_{z\in \operatorname{spt}(D)\cup\operatorname{spt}(E)} d_\Delta(z),
\]
with the convention $\eta_{D,E}=+\infty$ when $D=E=\varnothing$.
Then
\[
d_{\mathrm{SK}}(D,E)^2
\le
\frac{\sqrt{2(|D|+|E|)}}{\min(\delta^{\times}_{D,E},\eta_{D,E})}
\,W_{2,\Delta}(D,E).
\]
In particular, if $|D|,|E|\le M$, then
\[
d_{\mathrm{SK}}(D,E)^2
\le
\frac{2\sqrt M}{\min(\delta^{\times}_{D,E},\eta_{D,E})}
\,W_{2,\Delta}(D,E).
\]
\end{proposition}

\begin{proof}
Let $\mathcal M$ be an optimal diagonal-aware matching between $D$ and $E$, and let
$U\subset D$, $V\subset E$ be the unmatched points. Set
\[
\lambda_{D,E}:=\min(\delta^{\times}_{D,E},\eta_{D,E}).
\]
As in Proposition~\ref{app:prop:SKOT_reverse_finite_dataset}, we construct a coupling between
\[
\alpha:=\mu_D+\nu_E,
\qquad
\beta:=\mu_E+\nu_D
\]
by coupling, for each matched pair $(x,y)\in \mathcal M$, the atoms $\iota(x)$ with $\iota(y)$ and
$\iota(\Pi(y))$ with $\iota(\Pi(x))$; for each $x\in U$, the atoms $\iota(x)$ and $\iota(\Pi(x))$;
and for each $y\in V$, the atoms $\iota(\Pi(y))$ and $\iota(y)$. Hence
\[
d_{\mathrm{SK}}(D,E)^2
\le
\sum_{(x,y)\in \mathcal M}
\Bigl(
|\iota(x)-\iota(y)|
+
|\iota(\Pi(y))-\iota(\Pi(x))|
\Bigr)
+
\sum_{x\in U}|\iota(x)-\iota(\Pi(x))|
+
\sum_{y\in V}|\iota(y)-\iota(\Pi(y))|.
\]

If $(x,y)\in \mathcal M$ and $x=y$, both matched contributions vanish. If $x\neq y$, then
$\|x-y\|_2\ge \delta^{\times}_{D,E}\ge \lambda_{D,E}$, so
\[
|\iota(x)-\iota(y)|\le 1\le \frac{\|x-y\|_2}{\lambda_{D,E}},
\qquad
|\iota(\Pi(y))-\iota(\Pi(x))|
\le 1\le \frac{\|x-y\|_2}{\lambda_{D,E}}.
\]
Likewise, if $x\in U$, then $d_\Delta(x)\ge \eta_{D,E}\ge \lambda_{D,E}$, hence
\[
|\iota(x)-\iota(\Pi(x))|
\le 1\le \frac{d_\Delta(x)}{\lambda_{D,E}},
\]
and similarly for $y\in V$. Therefore
\[
d_{\mathrm{SK}}(D,E)^2
\le
\frac{1}{\lambda_{D,E}}
\Bigl(
2\sum_{(x,y)\in \mathcal M}\|x-y\|_2
+
\sum_{x\in U} d_\Delta(x)
+
\sum_{y\in V} d_\Delta(y)
\Bigr).
\]

Now write the right-hand side as a sum of exactly $|D|+|E|$ nonnegative terms:
two copies of $\|x-y\|_2$ for each matched pair $(x,y)\in \mathcal M$, and one copy of each
$d_\Delta(x)$, $x\in U$, and $d_\Delta(y)$, $y\in V$. By Cauchy--Schwarz,
\[
2\sum_{(x,y)\in \mathcal M}\|x-y\|_2
+
\sum_{x\in U} d_\Delta(x)
+
\sum_{y\in V} d_\Delta(y)
\le
\sqrt{|D|+|E|}
\Bigl(
2\sum_{(x,y)\in \mathcal M}\|x-y\|_2^2
+
\sum_{x\in U} d_\Delta(x)^2
+
\sum_{y\in V} d_\Delta(y)^2
\Bigr)^{1/2}.
\]
Since
\[
2\sum_{(x,y)\in \mathcal M}\|x-y\|_2^2
+
\sum_{x\in U} d_\Delta(x)^2
+
\sum_{y\in V} d_\Delta(y)^2
\le
2\,W_{2,\Delta}(D,E)^2,
\]
we obtain
\[
d_{\mathrm{SK}}(D,E)^2
\le
\frac{\sqrt{2(|D|+|E|)}}{\lambda_{D,E}}\,W_{2,\Delta}(D,E).
\]
The last claim follows from $|D|+|E|\le 2M$.
\end{proof}

\begin{corollary}[A one-number constant on a finite dataset]
\label{app:cor:SKOT_reverse_dataset_simple}
Let $\mathcal F$ be a finite family of diagrams and set
\[
X_{\mathcal F}:=\bigcup_{D\in\mathcal F}\operatorname{spt}(D),
\qquad
\delta_{\mathcal F}:=
\min\{\|x-y\|_2:\ x,y\in X_{\mathcal F},\ x\neq y\},
\]
with the convention $\delta_{\mathcal F}=+\infty$ if $X_{\mathcal F}$ has at most one point, and
\[
\eta_{\mathcal F}:=\min_{x\in X_{\mathcal F}} d_\Delta(x).
\]
If $M:=\max_{D\in\mathcal F}|D|$, then for all $D,E\in\mathcal F$,
\[
d_{\mathrm{SK}}(D,E)^2
\le
\frac{2\sqrt M}{\min(\delta_{\mathcal F},\eta_{\mathcal F})}
\,W_{2,\Delta}(D,E).
\]
\end{corollary}

\begin{proof}
For any $D,E\in\mathcal F$, one has $\delta^{\times}_{D,E}\ge \delta_{\mathcal F}$ and
$\eta_{D,E}\ge \eta_{\mathcal F}$. Apply Proposition~\ref{app:prop:SKOT_reverse_simple}.
\end{proof}

%% file: Supplementary_material.tex
	\appendices

%% file: supplementary_experiments.tex
\section{Additional experimental results}\label{app:additional_experiments}

\subsection{Kernel-based contiguous segmentation of ordered collections}
\label{sec:kernel_segmentation_supp}

This supplementary experiment evaluates whether the Gaussian $d_{\mathrm{SK}}^{}$
kernel can be used to partition an ordered collection of persistence
diagrams into contiguous segments. For collections with an explicitly
temporal ordering, this corresponds to a change-point-detection
problem. More generally, the experiment should be interpreted as an
order-constrained segmentation problem, since the stored sample order
may represent time or another ordered parameterization.

For each collection \(c\), let
\[
X_0,\ldots,X_{n_c-1}
\]
denote the persistence diagrams in their stored order, and let
\[
K_c:=K_{c,30}
\]
be the Gaussian $d_{\mathrm{SK}}^{}$ kernel matrix constructed from
\(d_{\mathrm{SK},30}\). Its bandwidth is fixed independently for each
collection as
\[
\sigma_c
=
\operatorname{median}
\left\{
d_{\mathrm{SK},30}(X_i,X_j):
i<j,\;
d_{\mathrm{SK},30}(X_i,X_j)>0
\right\}.
\]
The bandwidth is therefore determined from the pairwise $d_{\mathrm{SK}}^{}$
distances without using the reference labels or boundary locations.

For a half-open candidate segment
\[
[a,b)
=
\{a,a+1,\ldots,b-1\},
\]
we use the within-segment kernel dispersion
\[
C_c([a,b))
=
\sum_{i=a}^{b-1}K_c(i,i)
-
\frac{1}{b-a}
\sum_{i,j=a}^{b-1}K_c(i,j).
\]
If \(\psi_c\) denotes a feature map associated with the Gaussian $d_{\mathrm{SK}}^{}$
kernel, so that
\[
K_c(i,j)
=
\left\langle
\psi_c(X_i),\psi_c(X_j)
\right\rangle,
\]
then this cost is exactly the sum of squared distances of the feature
vectors to their mean within the segment:
\[
C_c([a,b))
=
\sum_{i=a}^{b-1}
\left\|
\psi_c(X_i)-\overline{\psi}_{a:b}
\right\|^2,
\qquad
\overline{\psi}_{a:b}
=
\frac{1}{b-a}
\sum_{i=a}^{b-1}\psi_c(X_i).
\]
Consequently, a segment has a small cost when its persistence
diagrams are mutually similar according to the Gaussian $d_{\mathrm{SK}}^{}$ kernel.

The number \(k\) of segments is fixed to the number of reference
groups supplied by the collection metadata. Dynamic programming then
finds the globally optimal boundaries
\[
0=\widehat{\tau}_0
<
\widehat{\tau}_1
<
\cdots
<
\widehat{\tau}_{k-1}
<
\widehat{\tau}_k=n_c
\]
minimizing
\[
\sum_{\ell=1}^{k}
C_c
\left(
[\widehat{\tau}_{\ell-1},\widehat{\tau}_{\ell})
\right).
\]
Only the number \(k\) and the stored sample order are provided to the
algorithm. The reference labels and reference boundary locations are
not used to compute the segmentation; they are used only afterwards
for evaluation.

No minimum segment length is imposed, beyond requiring each segment
to contain at least one sample. Therefore, the optimization may
produce a one-sample segment when this decreases the total
within-segment dispersion.

A reported boundary \(b\) denotes the first sample index of the next
segment. For example, the boundary list
\[
[4,8]
\]
defines the three half-open segments
\[
[0,4),\qquad[4,8),\qquad[8,n_c).
\]
Thus, the corresponding groups contain samples \(0\)--\(3\),
\(4\)--\(7\), and \(8\)--\(n_c-1\), respectively.

Let
\[
\tau_1,\ldots,\tau_{k-1}
\]
be the reference boundaries and
\[
\widehat{\tau}_1,\ldots,\widehat{\tau}_{k-1}
\]
the predicted boundaries. The boundary mean absolute error is
\[
\operatorname{Boundary\ MAE}
=
\frac{1}{k-1}
\sum_{\ell=1}^{k-1}
\left|
\widehat{\tau}_{\ell}-\tau_{\ell}
\right|.
\]
It measures the average displacement of the predicted boundaries in
numbers of samples. The Adjusted Rand Index (ARI), by contrast,
compares the complete predicted segmentation with the complete
reference partition. The two measures are complementary: Boundary
MAE evaluates boundary localization, whereas ARI evaluates the
resulting assignments of all samples.

We include only collections whose metadata reference labels form
exactly \(k\) contiguous blocks in the stored sample order. This
criterion is satisfied by 8 of the 12 collections. 

\input{tables/table_kernel_segmentation_supplementary.tex}

Overall, the Gaussian $d_{\mathrm{SK}}^{}$ kernel segmentation exactly recovers all reference
boundaries on 5 of the 8 eligible collections.

For Sea Surface Height, the predicted boundaries are
\(
[12,24,31],
\)
compared with the reference boundaries
\(
[12,24,36].
\)
The first two boundaries are recovered exactly, while the last
boundary is detected five samples early, yielding a Boundary MAE of
\(1.67\) and an ARI of \(0.773\), tying the second-best ARI achieved
on this collection among the three clustering methods evaluated in
\autoref{sec:kernel_ensemble_results}, namely the \(0.773\) obtained
by Hilbert \(k\)-means.

Earthquake provides the least favorable result among the eight
collections. Its predicted boundaries are
\(
[3,4],
\)
instead of the reference boundaries
\(
[4,8].
\)
The predicted partition therefore contains the one-sample segment
\([3,4)\), which is permitted because no minimum segment length is
imposed. This results in a Boundary MAE of \(2.50\) and an ARI of
\(0.408\), exceeding the ARI of \(0.368\) obtained by each of the
three clustering methods evaluated in
\autoref{sec:kernel_ensemble_results}.

This supplementary experiment demonstrates that the Gaussian $d_{\mathrm{SK}}^{}$
kernel can be used directly in kernel change-point and contiguous
segmentation methods~\cite{arlot2019kernelchangepoint}. 

%% file: tables/table_kernel_segmentation_supplementary.tex
\begin{table}[t]
\centering
\caption{
Fixed-\(k\) contiguous segmentation using the Gaussian $d_{\mathrm{SK}}^{}$ kernel at
\(L=30\). For each collection, \(k\) is the number of reference groups
supplied by the metadata, and the Gaussian bandwidth is the median of
the positive pairwise \(d_{\mathrm{SK},30}\) distances, without
label-dependent tuning. Only the 8 collections whose reference
labels form \(k\) contiguous blocks in the stored sample order are
included. A boundary \(b\) denotes the first index of the next
half-open segment; for example, boundaries \([4,8]\) define
\([0,4)\), \([4,8)\), and \([8,n)\). Boundary MAE is the mean absolute
difference, measured in sample indices, between corresponding
reference and predicted boundaries. ARI compares the complete
predicted segmentation with the metadata reference partition. The
stored order is explicitly temporal for some collections and may
represent another ordered parameterization for others.
}
\label{tab:kernel_segmentation_supp}
\scriptsize
\begin{tabular}{lrrrr}
\toprule
Dataset
& Reference boundaries
& Predicted boundaries
& Boundary MAE
& ARI vs.\ reference \\
\midrule

Isabel
& [4, 8]
& [4, 8]
& 0.00
& 1.000 \\

Earthquake
& [4, 8]
& [3, 4]
& 2.50
& 0.408 \\

Ionization 2D
& [4, 8, 12]
& [4, 8, 12]
& 0.00
& 1.000 \\

Ionization 3D
& [4, 8, 12]
& [4, 8, 12]
& 0.00
& 1.000 \\

Cloud Processes
& [4, 8]
& [4, 8]
& 0.00
& 1.000 \\

Asteroid temporal
& [5, 10, 15]
& [7, 10, 15]
& 0.67
& 0.756 \\

Sea Surface Height
& [12, 24, 36]
& [12, 24, 31]
& 1.67
& 0.773 \\

Starting Vortex
& [6]
& [6]
& 0.00
& 1.000 \\

\bottomrule
\end{tabular}
\end{table}